\documentclass{article}

\usepackage{arxiv}

\usepackage[utf8]{inputenc} 
\usepackage[T1]{fontenc}    
\usepackage{hyperref}       
\usepackage{url}            
\usepackage{booktabs}       
\usepackage{amsfonts}       
\usepackage{nicefrac}       
\usepackage{microtype}      
\usepackage{lipsum}
\usepackage{latexsym,amsfonts,amssymb,amsthm,amsmath}
\usepackage{mathrsfs,mathtools}
\usepackage{commath}
\usepackage{tikz} 
\usepackage{comment}
\usepackage{natbib}
\usepackage{rotating}
\usepackage{algorithm}
\usepackage{algpseudocode}
\usepackage{bbm}
\usepackage{siunitx}
\usepackage{xcolor}
\usepackage[scientific-notation=true]{siunitx}
\newtheorem{theorem}{Theorem}
\newtheorem{Proposition}[theorem]{Proposition}
\newcommand{\Var}{\mathrm{Var}}
\newcommand{\Cov}{\mathrm{Cov}}
\newcommand{\Corr}{\mathrm{corr}}
\newcommand{\indep}{\perp \!\!\! \perp}

\title{Peak Inference for Gaussian Random Fields on a Lattice}
\date{\today}

\usepackage{xcolor}

\newcommand{\nt}[1]{\color{black} #1}
\newcommand{\ntb}[1]{\color{black} #1}

\author{
  Tuo~Lin \\
  Department of Biostatistics and UF Health Cancer Center\\
  University of Florida\\
  Gainesville, FL 32603 \\
  \texttt{tuolin@ufl.edu} \\
   \And
 Armin~Schwartzman \\
  Division of Biostatistics and Hal{\i}c{\i}o\u{g}lu Data Science Institute\\
  University of California, San Diego\\
  San Diego, CA 92122 \\
  \texttt{armins@ucsd.edu} \\
  \And
  Samuel~Davenport\\
  Division of Biostatistics\\
  University of California, San Diego\\
  San Diego, CA 92122 \\
  \texttt{sdavenport@health.ucsd.edu}
}

\begin{document}
\maketitle

\begin{abstract}
In this work we develop a Monte Carlo method to compute the height distribution of local maxima of a stationary Gaussian or \ntb{Gaussian-related} random field that is observed on a regular lattice. We show that our method can be used to provide valid peak based inference in \ntb{datasets with low levels of smoothness, where existing formulae derived for continuous domains are} not accurate. We also extend the methods in \cite{worsley2005improved} and \cite{taylor2007maxima} to compute the \nt{peak} height distribution and compare them with our approach. Lastly, we apply our method to a task fMRI dataset to show how it can be used in practice. 
\end{abstract}


\section{Introduction}
Statistical parametric mapping (SPM) is widely used as a tool to conduct statistical inference on neuroimaging data \citep{friston1989localisation, worsley1992three, worsley1996unified}. Recently, \cite{eklund2016cluster, eklund2019cluster} investigated the validity of cluster size and voxelwise inference based on random field theory (RFT) and found that a number of the assumptions that have been traditionally made do not hold in practice. \ntb{One important assumption}, which we address in this work, is that the data is sufficiently smooth so that it can be treated as a continuous random field. Inference based on peaks or local maxima, recognized as topological features of the statistical summary maps \citep{chumbley2009false, chumbley2010topological, friston1989localisation, schwartzman2011multiple, cheng2017multiple} strongly relies on this assumption. In this paper, we circumvent this assumption and develop a method for performing peak inference that is valid for data observed on a regular lattice.

The traditional approach to obtaining peak $p$-values in fMRI analysis has been to assume that the data is distributed as a smooth stationary Gaussian random field. Given this, \cite{nosko1969local, adler1981geometry, cheng2015distribution} showed that the distribution of the height of peaks above a peak-defining threshold $u \in \mathbb{R}$ is asymptotically exponential (as $u \rightarrow \infty$). The choice of $u$ is somewhat arbitrary and this result only holds in practice for reasonably large choices of $u$. Recently, \cite{cheng2015distribution} obtained a more general formula to calculate the exact height distribution of local maxima in an isotropic Gaussian random field, that is valid for all peak heights and does not require a pre-threshold. This distribution can be used to compute a $p$-value at each local maximum based on its height. The formula has a single parameter $\kappa$, which only depends on the shape of the auto-correlation function near the origin, and is invariant under spatial scaling. While elegant, the formula is only accurate when the Gaussian random field is sampled on a continuous domain, instead of a discrete lattice grid, which in practice can require a high level of applied smoothing. To give context, \cite{schwartzman2019peak} suggest that this formula is imprecise when data has an intrinsic FWHM that is lower than 7 voxels. However, since the typical smoothing kernel in an fMRI study has an FWHM of 3 voxels \nt{\citep{eklund2016cluster}}, using this formula provides conservative $p$-values in practice. Moreover, the isotropic assumption is rather strong and is unlikely to hold in practice. Thus it is desirable to directly calculate the height distribution of local maxima sampled on a discrete lattice, which we shall refer to as discrete local maxima (DLM).


In order to address the difference between a discrete lattice and a smooth random field, \cite{worsley2005improved} and \cite{taylor2007maxima} introduced a method that targets the distribution of the global maximum on a lattice in order to provide control of the voxelwise family-wise error rate. Although this method aims to infer on the global maximum, it can also be used, after some modifications which we develop here, to compute the height distribution of local maxima. However, their approach is limited in that it is only valid for a narrow class of Gaussian random fields, namely the ones that arise as the result of convolving Gaussian white noise with a separable kernel. In addition, they require local maxima to be defined as those voxels with height values larger than \nt{their} immediate neighbors along the coordinate axes, excluding diagonally adjacent neighbors. Figure \ref{fig1} gives a rough idea of why this assumption is restrictive in practice by comparing density plots from peak height distributions calculated from both Worsley and Taylor's analytical DLM approach (which we shall refer to as ADLM) and Cheng and Schwartzman's continuous RFT approach. 

\begin{figure}[!htp]
  \centering
  \includegraphics[trim=50 50 50 50, clip, width = 50mm]{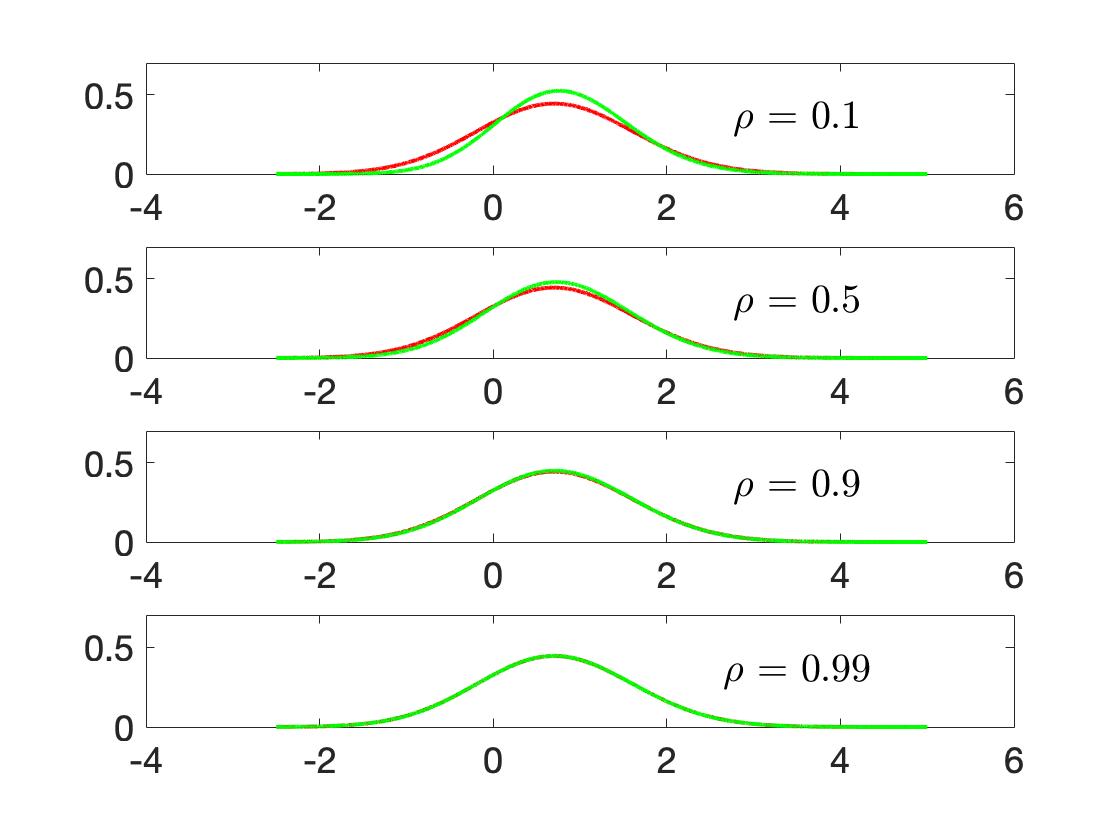}
  \includegraphics[trim=50 50 50 50, clip, width = 50mm]{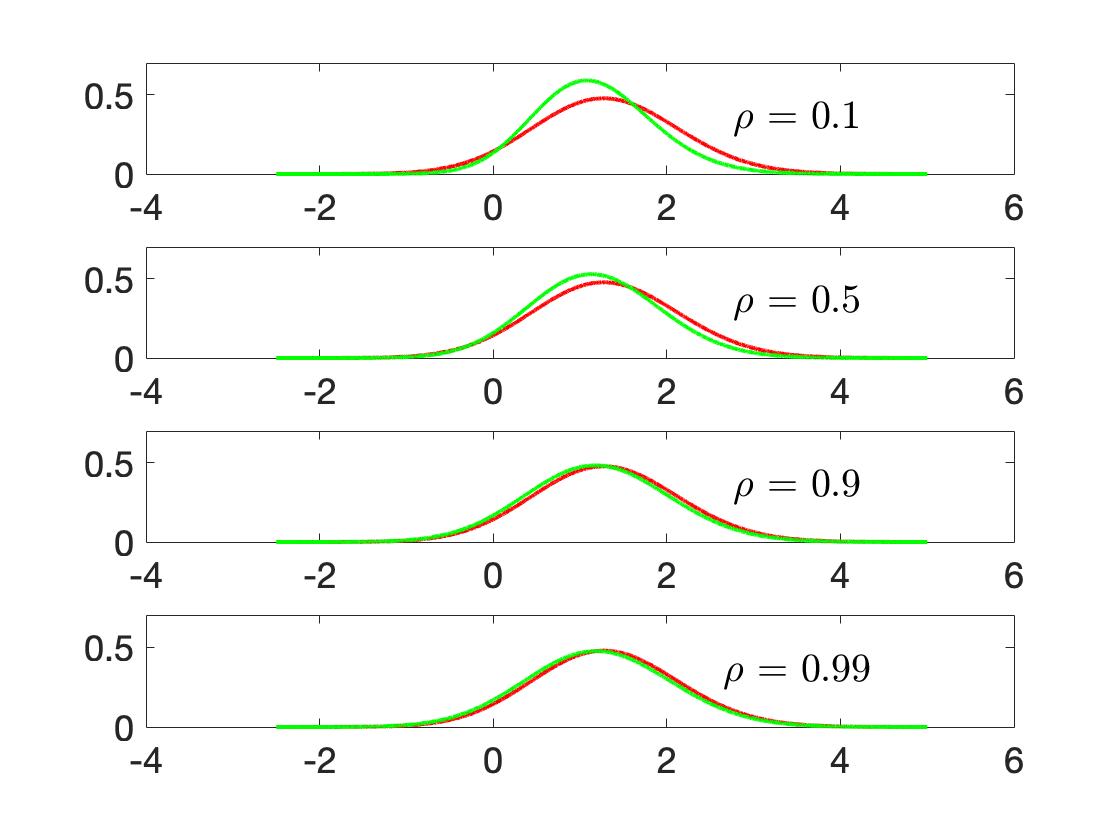}
  \includegraphics[trim=50 50 50 50, clip, width = 50mm]{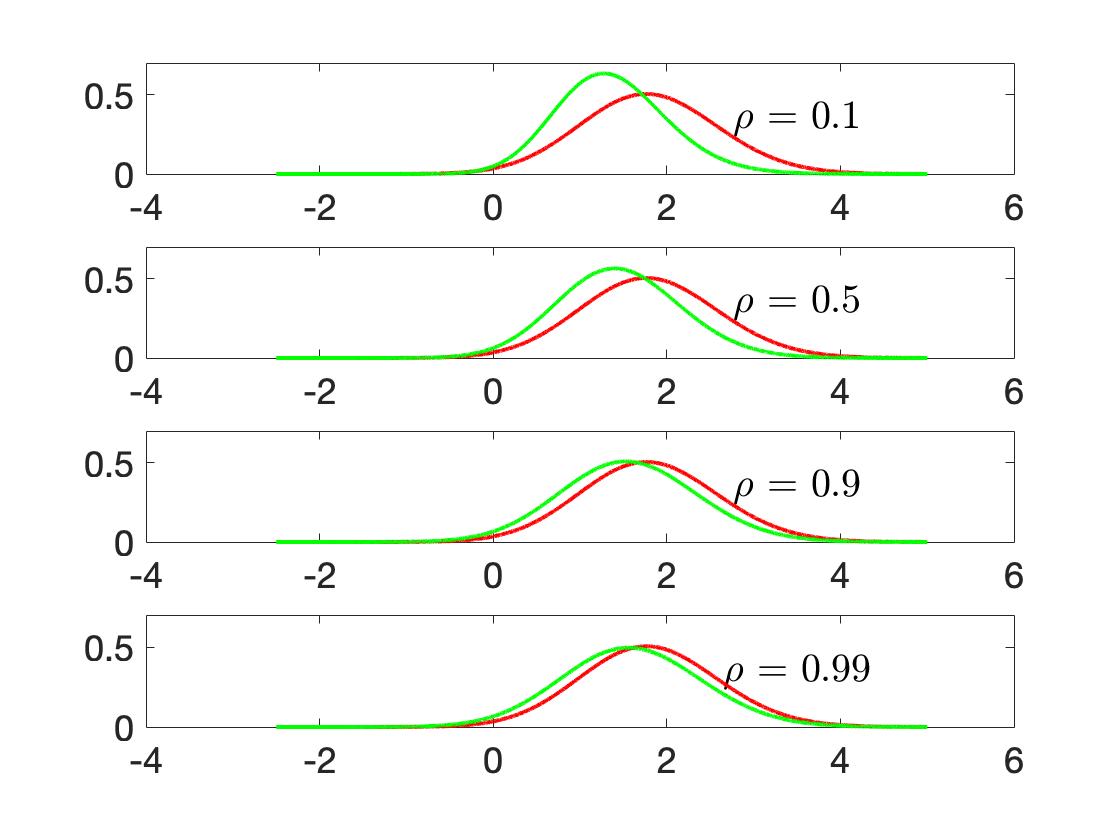}
  \caption{Theoretical peak height density function for local maxima. Left: 1D, Middle: 2D, Right: 3D. Each row is calculated using a different correlation between adjacent voxels. In each plot the green curve is from ADLM and the red curve is from the continuous RFT method. In \ntb{1D (left)}, when $\rho$ is small, the ADLM density is narrower, but as $\rho$ increases, the discrepancy disappears and the two methods converge. In 2D \ntb{(middle)} and 3D \ntb{(right)}, the differences between the two methods remain for high $\rho$, with the ADLM density shifted to the left. This occurs because ADLM does not consider diagonal voxels as neighbors, so \ntb{the} distribution of local maxima obtained from this method consists of smaller height values and \ntb{thus the left shift relative to the continuous method increases as the dimension increases}. Note that in 1D there are no diagonal voxels and so convergence occurs.}
\label{fig1}
\end{figure}

To address these issues, we propose a computation-based method called Monte Carlo DLM (MCDLM) that works for any stationary Gaussian random field under arbitrary connectivity (i.e. where local maxima are defined with respect to any desired neighborhood). This improves upon ADLM in that it  allows the accurate computation of the height distribution of local maxima on a lattice without assuming a separable covariance function and permitting a range of neighborhood \ntb{structures}. Our approach works by calculating the joint covariance of a voxel and its neighbors either theoretically or via empirical estimation, and then generating random samples from a multivariate Gaussian distribution with the calculated covariance and storing the samples that have larger height values than their neighbors. This provides an empirical cdf for the height of local maxima via numerical integration. A $p$-value for an observed peak in data can be computed \ntb{from the empirical cdf}. We also extend this approach to calculate the height distribution of local maxima of $t$-fields, by generating the random samples from a multivariate $t$-distribution.

The structure of this paper is as follows. Section \ref{sec2} provides details about how to calculate the peak height distribution using continuous RFT, the ADLM and the MCDLM method. \ntb{Sections} \ref{sec3} and \ref{sec4} apply MCDLM to isotropic Gaussian random fields, $t$-fields and stationary Gaussian fields with known nonseparable and unknown covariance and \ntb{compare} its performance with ADLM and continuous RFT. Section \ref{sec3} provides the simulation setup and Section \ref{sec4} \ntb{reports} all the simulation results. Section \ref{sec5} discusses a real \ntb{data} example of applying the proposed methods. Section \ref{sec6} gives the concluding remarks. Code to reproduce the results of this paper are available on GitHub (https://github.com/tuolin123/DLM-Code) and in the RFTtoolbox (https://github.com/sjdavenport/RFTtoolbox).

\section{Theory and methods for calculating the height distribution of local maxima}
\label{sec2}
Let $\{Z(s), s\in \mathcal{L}\}$ be a real-valued stationary Gaussian random field parametrized on a $D$-dimensional set $\mathcal{L}$, where $D \in \mathbb{Z}^{+}$. We assume that $\mathcal{L}$ is a regularly spaced discrete lattice, in particular that 
\begin{align*}
    \mathcal{L} \subset \left\{\sum_{d=1}^D k_dv_de_d: k_d \in \mathbb{Z}, v_d \in \mathbb{R}^{+} \text{ for } 1 \leq d \leq D\right\},
\end{align*}
where $(e_d)_{1 \leq d \leq D}$ is the standard basis in $\mathbb{R}^D$, $v_d$ represents the step size, and \ntb{$k_d$ is the scale of image} in the $d$th direction. Our interest lies in calculating the peak height distribution, which for $u \in \mathbb{R}$, is defined as
\begin{align}
P( Z(s) > u | s\ \text{is a local maximum}) = P[Z(s)>u| \ntb{Z(s) > Z(t)}, \forall t \in \mathcal{N}(s)], \label{eqn2.2}
\end{align}
where $\mathcal{N}(s)$ denotes the set of neighbors of $s\in \mathcal{L}$ in the discrete lattice. The most relevant neighborhoods are the partially connected (PC) and fully connected (FC) \ntb{neighborhoods} that are respectively defined as 
\begin{align}
    \quad\mathcal{N}_{PC}(s) &= \left\lbrace s + k_d v_d e_d: k_d \in \lbrace -1,1\rbrace \text{ for } 1 \leq d \leq D\right\rbrace \text{ and }\label{partial}\\ 
    \quad\mathcal{N}_{FC}(s) &= \left\lbrace s +  \sum_{d=1}^D k_d v_d e_d: k_d \in \lbrace -1,0,1\rbrace \text{ for } 1 \leq d \leq D \right\rbrace \mathbin{\bigg\backslash} \left\{s\right\}. \label{connectivity}
\end{align}

Figure \ref{fig0} illustrates the two types of neighborhoods for a point, $s_5$, on a 2D regular lattice. If $s_5$ is partially connected to the adjacent pixels in the horizontal and vertical directions, then $\mathcal{N}_{PC}(s_5) = \left\{s_2,s_4,s_6,s_8\right\}$, shown in the left of Figure \ref{fig0}. If $s_5$ is fully connected, meaning it is connected to pixels in the horizontal, vertical and diagonal directions, then $\mathcal{N}_{FC}(s_5) = \left\{s_1,s_2,s_3,s_4,s_6,s_7,s_8,s_9\right\},$ shown in the right of Figure \ref{fig0}.
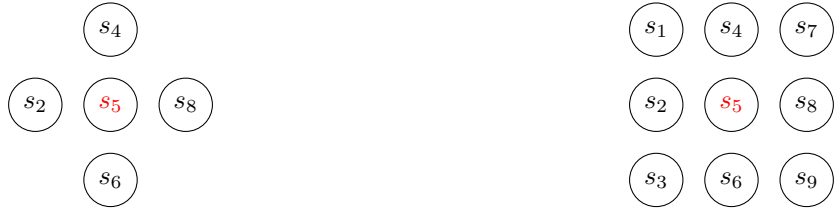
\begin{figure}[!htp]
\begin{minipage}{0.5\textwidth}
\centering
\begin{tikzpicture}[main/.style = {draw, circle}] 
\node[main] (2) at (1,2) {$s_2$}; 
\node[main] (4) at (2,3) {$s_4$}; 
\node[main] (5) at (2,2) {$\color{red}s_5$}; 
\node[main] (6) at (2,1) {$s_6$}; 
\node[main] (8) at (3,2) {$s_8$}; 
\end{tikzpicture} 
\end{minipage}\hfill
\begin{minipage}{0.5\textwidth}
\centering
\begin{tikzpicture}[main/.style = {draw, circle}] 
\node[main] (1) at (1,1) {$s_3$}; 
\node[main] (2) at (1,2) {$s_2$}; 
\node[main] (3) at (1,3) {$s_1$}; 
\node[main] (4) at (2,3) {$s_4$}; 
\node[main] (5) at (2,2) {$\color{red}s_5$}; 
\node[main] (6) at (2,1) {$s_6$}; 
\node[main] (7) at (3,3) {$s_7$}; 
\node[main] (8) at (3,2) {$s_8$}; 
\node[main] (9) at (3,1) {$s_9$}; 
\end{tikzpicture} 
\end{minipage}
\caption{Local pixel neighborhood in 2D. The partially and fully connected neighborhoods are shown on the left and right respectively. The point $s_5$ in red is considered a local maximum if its value is larger than its neighbors.}\label{fig0}
\end{figure}

\subsection{Analytical DLM \ntb{(ADLM)} method} 
\label{sec2.1}
The DLM approach of \cite{worsley2005improved} and \cite{taylor2007maxima} provides closed form expressions for the family-wise error rate in testing the \ntb{presence of} signals in data \nt{where the noise is assumed to arise from Gaussian white noise smoothed with a separable} Gaussian kernel. They do not explicitly focus on the peak height distribution. \ntb{However,} in inferring on the global maximum they calculate probabilities of the form $ P\left[\{Z(s) > u\}\cap_{t\in\mathcal{N}_{PC}(s)}\{Z(t) < Z(s)\}\right] $. These probabilities can be used to calculate a peak height distribution for the partially connected neighborhood since \eqref{eqn2.2} can be written as
\begin{equation}
P [Z(s) > u|Z(t) < Z(s), \forall t \in \mathcal{N}_{PC}(s)] = \frac{P\left[\{Z(s) > u\}\cap_{t\in\mathcal{N}_{PC}(s)}\{Z(t) < Z(s)\}\right]}{P[Z(t) < Z(s), \forall t \in \mathcal{N}_{PC}(s)]}. \label{eqn2.9}
\end{equation}
Using \nt{the} results of \cite{worsley2005improved} and \cite{taylor2007maxima}, we \nt{can} expand the left hand side of (\ref{eqn2.9}) as $\int_u^{\infty} f_{\text{DLM}}(z) dz$, where $f_{\text{DLM}}(z)$ is the density function of the \ntb{peak} height distribution. Under the assumption that $Z$ \nt{has a separable covariance function}, \nt{as we show in Appendix \ref{appendix.a2}}, $f_{\text{DLM}}(z)$ has the form
\begin{align}
f_{\text{DLM}}(z) = \frac{\prod_{d=1}^DQ(\rho_d,z)\phi(z)}{\int_{-\infty}^\infty\bigg(\prod_{d=1}^DQ(\rho_d,z)\bigg)\phi(z)dz} \label{eqn2.22},
\end{align}
where
\begin{align}
&h_d = \sqrt{\frac{1-\rho_d}{1+\rho_d}},\quad \alpha_d = \sin^{-1}\left(\sqrt{(1-\rho_d^2)/2}\right),\quad z^+ = \max(z,0) \nonumber\\
&Q(\rho_d, z) = 1 - 2\Phi(h_dz^+) + \frac{1}{\pi}\int_{0}^{\alpha_d}\exp\left(-\frac{1}{2}h^2_dz^2/\sin^2\theta\right)d\theta, \label{sdfsdf}
\end{align}
and $\rho_d$ is the correlation between \ntb{the values of $Z$ at} two voxels along each axis direction $d$, given by $\rho_d = \rho(s, s+v_de_d)$, where $\rho(\cdot,\cdot)$ is defined as 
\begin{align}
\rho(s,t) = \exp[-(s-t)'\Lambda(s-t)/2] \label{eqn2.3},
\end{align}
where $\Lambda = \operatorname{diag}(1/(2\eta_1^2),..., 1/(2\eta_D^2))$, and $(\eta_d)_{d=1,...,D}$ is the standard deviation of the Gaussian kernel in the $d$th direction. The correlation function in \eqref{eqn2.3} arises, for example, from integration of continuous white noise against a Gaussian kernel. This approach also allows for the calculation of the height of local maxima on the boundary of the image or a mask by substituting $Q(\rho_d, z)$ in \eqref{eqn2.22} with $1-\Phi(h_dz)$ if a voxel on the boundary only has one neighbor, and with 1 if it has no neighbors. Further details regarding the derivation of (\ref{eqn2.22}) are provided in Appendix \ref{appendix.a2}. Since this \nt{approach} provides a closed form density function, \nt{we shall refer to it as the} analytical DLM (ADLM) method.  

One critical assumption of ADLM is that the correlation function has a specific separable structure. Under this assumption things simplify because conditioned on the center voxel, the distribution of the height of neighboring voxels along a given axis are conditionally independent of the distribution of the height at neighboring voxels along the other (perpendicular) axis directions \nt{as described formally} in the proposition below.

\begin{Proposition}
\label{prop1}
Given data $\{Z(s), s\in \mathcal{L}\}$ such that the spatial correlation between $s$ and $t$ is defined as \eqref{eqn2.3}, \nt{and $d_1, d_2\in\left\{1,...,D\right\}$ such that $s\pm v_{d_1}e_{d_1}, s\pm v_{d_2}e_{d_2} \in S$ and letting $\indep$ denote independence, it follows that,}
\begin{align*}
    \begin{pmatrix}
    Z(s-v_{d_1}e_{d_1}) \\
    Z(s+v_{d_1}e_{d_1})           
    \end{pmatrix}
    \indep 
    \begin{pmatrix}
    Z(s-v_{d_2}e_{d_2}) \\
    Z(s+v_{d_2}e_{d_2})           
    \end{pmatrix}\bigg| Z(s).
\end{align*}
\end{Proposition}
The result in Proposition \ref{prop1} is briefly noted in \cite{taylor2007maxima} and we provide a \nt{formal} proof in Appendix \ref{appendix.a1}. \nt{To help understand this result with a visual example note that for the neighborhood structures shown in} Figure \ref{fig0}, under the required assumptions, we have
\begin{align*}
\begin{pmatrix}
Z(s_4) \\
Z(s_6)           
\end{pmatrix}
\indep 
\begin{pmatrix}
Z(s_2) \\
Z(s_8)           
\end{pmatrix}\bigg| Z(s_5).
\end{align*}

This conditional independence result holds along the horizontal and vertical axes and allows for an expansion for the distribution of partially connected local maxima. However, it does not imply independence when the diagonals are included, i.e.,
\begin{align*}
\left(Z(s_1), Z(s_3), Z(s_7), Z(s_9) \right)^\top
\not\!\perp\!\!\!\perp
\left(Z(s_2), Z(s_4), Z(s_6), Z(s_8) \right)^\top | Z(s_5).
\end{align*}
Thus, this ADLM method can only be used to calculate the height distribution of peaks that are greater than their directly adjacent neighbors. 

\nt{Under \ref{eqn2.3}, the correlation $\rho_d$ between two adjacent voxels along each lattice axis $d$ can be written as} 
\begin{align}
\rho_d = \rho(s, s+ v_de_d) = \exp\left[-\frac{1}{2}\left(\frac{v_d^2}{2\eta_d^2}\right)\right] = \exp\left[-\frac{v_d^2}{4\eta_d^2}\right], \label{eqn2.32}
\end{align}
and if we assume $Z(s)$ is isotropic with a common standard deviation of the Gaussian kernel $\eta_d = \eta$ and $v_d = 1$, $\rho_d$ can be further simplified to $\rho_d = \rho = \exp[-1/4\eta^2]$, a function which does not depend on $d$.

The ADLM approach allows the calculation of the height distribution of local maxima on a discrete lattice. However, the method makes restrictive assumptions and its validity is limited to partial connectivity.

\subsection{The correlation function on the lattice}\label{SScorrfn}
The methods which we will develop in what follows strongly rely on the correlation function.  In this section we provide some explicit expansions of this function under the assumption that the fields are derived by smoothing i.i.d white noise with a kernel (we will relax this assumption later on).

Define the correlation function $\rho(s, t): \mathcal{L} \times \mathcal{L} \rightarrow \mathbb{R}$ to be the function that maps $s, t\in \mathcal{L}$ to $\Corr(Z(s),Z(t))$. As a step toward our goal of calculating peak $p$-values for a Gaussian random field on a regular discrete lattice, we shall calculate the spatial correlation \nt{analytically for fields generated by smoothing noise, before extending to the more general setting in Section \ref{sec2.5}}. Assume that $W:\mathcal{L} \rightarrow \mathbb{R}$ is a Gaussian random field consisting of i.i.d. unit variance white noise and for some kernel $K:\mathbb{R}^D  \rightarrow \mathbb{R}$,
\begin{align*}
    Z(s) = \sum_{l \in \mathcal{L}} K(s-l)W(l) \quad \text{for each}\quad s\in \mathcal{L}.
\end{align*}
The correlation function $\rho(s,t)$ is then
\begin{align*}
    \rho(s, t)  &= \frac{E\left[\sum_{l \in \mathcal{L}} K(s-l)W(l)\sum_{l'\in \mathcal{L}}K(t-l')W(l')\right]}{\sqrt{\Var\left[\sum_{l \in \mathcal{L}} K(s-l)W(l)\right]\Var\left[\sum_{l'\in \mathcal{L}}K(t-l')W(l')\right]}}\nonumber\\
    &=\frac{ E\left[\sum_{l \in \mathcal{L}}\sum_{l'\in \mathcal{L}}K(s-l)K(t-l')W(l)W(l')\right]}{\Var\left[\sum_{l \in \mathcal{L}} K(s-l)W(l)\right]}\nonumber\\
    &= \frac{\sum_{l\in \mathcal{L}}K(s-l)K(t-l)}{\sum_{l \in \mathcal{L}} K(s-l)^2}
\end{align*}

since $E[W(l)W(l')] = 0$ for $l \neq l'$ and $EW(l)^2 = 1$ for all $l$. In particular when $K$ is an isotropic Gaussian kernel, i.e., $K(s) = \frac{1}{\eta^D}\phi_D\left(\frac{||s||}{\eta}\right)$, for some $\eta > 0$ and each $s, t \in \mathcal{L}$, 
\begin{align}
    \rho(s, t) = \frac{\sum_{l\in \mathcal{L}}\frac{1}{\eta^{2D}}\phi_D\left(\frac{||s-l||}{\eta}\right)\phi_D\left(\frac{||t-l||}{\eta}\right)}{\sum_{l\in \mathcal{L}}\frac{1}{\eta^{2D}}\left[\phi_D\left(\frac{||s-l||}{\eta}\right)\right]^2},\label{eqn2.7}
\end{align}

where $\phi_D$ is the density function for the $D$ dimensional standard Gaussian distribution. As is common in fMRI analysis we will typically refer to this kernel using its full width at half maximum (FWHM) which is defined as $\text{FWHM} = 2\sqrt{2\ln2}\eta$. Using (\ref{eqn2.7}) as the correlation function to ADLM defined in (\ref{eqn2.22}) improves the performance of the ADLM approach relative to using (\ref{eqn2.3}) - the correlation function used in \cite{worsley2005improved} - as discussed in Section \ref{appendix.d2}.

More generally if $K$ is an elliptical Gaussian kernel, i.e., $K(s) = \prod_{j=1}^D\frac{1}{\eta_j}\phi_1\left(\frac{[s]_j}{\eta_j}\right)$, then 
\begin{align}
    \rho(s, t) = \frac{\sum_{l\in \mathcal{L}}\prod_{j=1}^D\frac{1}{\eta_j^2}\phi_1\left(\frac{[s-l]_j}{\eta_j}\right)\phi_1\left(\frac{[t-l]_j}{\eta_j}\right)}{\sum_{l\in \mathcal{L}}\prod_{j=1}^D\frac{1}{\eta_j^2}\left[\phi_1\left(\frac{[s-l]_j}{\eta_j}\right)\right]^2}, \label{eqn2.8}
\end{align}
where $[s-w]_j$ and $[s+v-w]_j$ refer to the $j$th \nt{elements of the vectors} $s-w$ and $s+v-w$ respectively. 

\subsection{Monte Carlo DLM \ntb{(MCDLM)} method}
\label{sec2.5}
In this section we introduce a new method based on Monte Carlo simulation to calculate the height distribution of local maxima on a discrete lattice. Our approach is based on the observation that the probability that $s \in \mathcal{L}$ is a local maximum based entirely on the distribution of $s$ with its neighbors, \nt{c.f.} \eqref{eqn2.2}. Define 
\begin{equation*}
    \mathbf{Z}(s) = \left(Z(s), Z(n_1(s)), ..., Z(n_k(s))\right)^\top,
\end{equation*}
where we \nt{write} the neighborhood of $s$ as 
\nt{$\mathcal{N}(s) = \left\lbrace n_1(s), \dots, n_k(s)  \right\rbrace \subseteq \mathcal{N}_{FC}(s)$} with \ntb{$ k = |\mathcal{N}(s)|$}. \nt{For the} partially connected neighborhood $k = 2D$ and for the fully connected neighborhood, $k = 3^D-1$. \nt{Under stationarity}, $\mathbf{Z}(s) \sim N(\mathbf{0}, \Sigma)$ for each $s$ where the covariance, $\Sigma = \mathrm{Cov}(\mathbf{Z}(s))$, is constant over the domain. The covariance matrix $\Sigma$ can be derived analytically under certain assumptions or estimated from the data, see below. Given $\Sigma$, the method calculates the peak height distribution via Monte Carlo experiments \ntb{by generating $M$ multivariate Gaussian vectors $\mathbf{Z}_m = \left(Z_{m1}, Z_{m2}, ..., Z_{m(k+1)}\right)^\top, \ 1\leq m \leq M$ of size $k+1$ and recording the height of the central voxel in each realization if it is higher than its neighbors.} Details of these experiments are described in Algorithm \ref{alg1}.

\begin{algorithm}[!htp]
\caption{MCDLM}
\label{alg1}
\begin{algorithmic}[1]
\Require The number of iterations $M\in \mathbb{N}$, and the \nt{$(k+1) \times (k+1)$} covariance matrix $\Sigma$
\Procedure{simulateLocmax}{$M, \Sigma$}
  \State \ntb{Initialize an empty vector $\boldsymbol{h} \leftarrow []$}
  \For{$m = 1, \dots, M$}
    \State \texttt{Generate $\mathbf{Z}_m \sim N(\mathbf{0}, \Sigma)$}, independent of $\mathbf{Z}_1,...,\mathbf{Z}_{m-1}$
    \If{$Z_{m1} > \max_{2 \leq j \leq k+1} Z_{mj}$}
        \State $\boldsymbol{h} = \left[\boldsymbol{h}, Z_{m1}\right]$
    \EndIf
  \EndFor
  \State \textbf{return} $\boldsymbol{h}$
\EndProcedure
\end{algorithmic}
\end{algorithm}

After obtaining the vector $\boldsymbol{h} = \left( h_1, h_2,..., h_N \right)^\top$, \ntb{$N\le M$, where $N$ is the number of MC realizations that satisfies the condition in step 5 of Algorithm \ref{alg1}}, for $u \in \mathbb{R}$, the MCDLM approximation to the peak height distribution can be calculated as 
\begin{align*}
    \hat{F}_N(u) = \frac{1}{N}\sum_{i=1}^N \mathbbm{1}\{h_i \leq u\},
\end{align*}
where $\mathbbm{1}\{\cdot\}$ denotes the indicator function and $N$ is the length of $\boldsymbol{h}$. For an observed peak of height $u$, a peak $p$-value can be computed as $1 - \hat{F}_N(u)$. In order to make this empirical $p$-value as accurate as possible, $N$ should be taken to be as large as possible. 

Now we address the question of how to calculate $\Sigma$. We will first do this for the particular isotropic subcase, discussed in the previous sections, in which $Z$ is obtained by integrating continuous white noise against \ntb{an isotropic (spherically symmetric)} Gaussian kernel. In that case, assuming a fully connected neighborhood, it can be shown that, \ntb{after proper reordering} of $\bf Z$,
\begin{align}
\Sigma = \underbrace{A\otimes A\otimes... \otimes A}_{D\ \text{terms of}\ A} = A^{\otimes D} \label{eqn2.5}
\end{align}
where
\begin{align*}
A = \begin{pmatrix}
1 & \rho & \rho^4\\
\rho & 1 & \rho\\
\rho^4 & \rho & 1
\end{pmatrix}.
\end{align*}
with $\rho$ the correlation between adjacent voxels as defined in \eqref{eqn2.32}. \nt{See Theorem \ref{thm:thm} in Appendix \ref{appendix.dd} for \ntb{a} formal statement and proof of this result.}

Equation (\ref{eqn2.5}) holds under isotropy and form of covariance stated in \eqref{eqn2.32}. For a general \nt{mean-zero} stationary field we can instead use the data to estimate $\Sigma$. To do so, \nt{without loss of generality, assume we observe and standardized} i.i.d Gaussian random fields $Z_1, \dots, Z_n$ on $\mathcal{L}$. We wish to infer on the distribution of peaks of the mean $\frac{1}{n}\sum_{i = 1}^n Z_n$. \nt{We can then} estimate $\Sigma$ from the data as follows. For each $s, t \in \mathcal{L}$, 
\begin{equation}
  \Cov\left( \frac{1}{n} \sum_{i = 1}^n Z_i(s), \frac{1}{n} \sum_{i = 1}^n Z_i(t)\right) = \frac{1}{n}\Cov(Z_1(s), Z_1(t)),
\end{equation}
as such, using the assumption of stationarity, we can estimate this covariance as   
\begin{align}
    \widehat{\mathrm{Cov}}(Z_1(s), Z_1(t)) = \frac{1}{n|L|}\sum_{i=1}^n\sum_{(s',t')\in L}Z_i(s')Z_i(t') \label{eqn5.1}
\end{align}
where $L = \{(s', t') \in \mathcal{L} \times \mathcal{L}: s' - t' = s-t\}$. \nt{This allows us to build an estimate of $\Sigma$ by calculating \eqref{eqn5.1} between relevant pairs of voxels.} If \nt{further} we assume that the fields are isotropic, then we can improve the accuracy of this estimate by taking $L =  \{(s', t') \in S \times S: ||s' - t'|| = ||s-t||\}$.

\subsection{MCDLM for $t$-fields}
\label{sec2.4}
We can also use our MCDLM approach to calculate the height distribution of local maxima of a $t$-field. The $t$-fields are generated by voxelwise calculation of $t$-statistic using
\begin{align}
    T(s) = \frac{\varepsilon(s)}{\sqrt{\sum_{i=1}^N Z_i^2/N}},\quad s \in \mathcal{L},\label{eqn.t}
\end{align}
where $Z_1,...,Z_N$ and $\varepsilon(s)$ are i.i.d \nt{stationary} Gaussian random fields observed on the lattice $\mathcal{L}$. In practice, the $t$-statistic is typically used as the test statistic for regression coefficients. 

In this setting the local neighborhood has a multivariate $t$-distribution according to the definition in \cite{roth2012multivariate}. Thus, given an estimate of the neighborhood covariance we can still apply MCDLM using Algorithm \ref{alg1} by changing the simulation in line 3 from a multivariate Gaussian distribution to a multivariate $t$-distribution.

This approach works well in practice (see Section \ref{sec4.3}) however it is somewhat computational expensive (especially as $\rho$ and the degrees of freedom increase). To get around this we consider \ntb{as an alternative approach} a voxelwise Gaussianization transformation of the $t$-fields (as in \cite{schwartzman2019peak}) which acts using the distribution function as follows:
\begin{align}
    Z(s) = -\Phi^{-1}[F_{t,\nu}(-T(s))], \label{eqn.t2Gauss}
\end{align}
where $F_{t,\nu}$ is the cdf of the $t$-distribution with $\nu$ degrees of freedom. We then apply the MCDLM method for the Gaussian field to the transformed $t$-field resulting in improved speed without signficantly impacting performance, see Section \ref{SS:gauss}.


\subsection{Continuous Gaussian random field theory} \label{sec2.3}
Historically \citep{chumbley2009false, chumbley2010topological,schwartzman2019peak} it has been common to use the results of continuous Gaussian random fields \citep{adler1981geometry} to perform inference on the lattice. \ntb{For comparison, we} here briefly outline how this works and explain how it can be used to provide height distributions for local maxima in smooth random fields. \nt{In order to apply Random Field Theory we must assume that the noise is sufficiently smooth for the good lattice assumption to hold \citep{davenport2023robust}. Then given a domain $S \supset \mathcal{L}$, which is compact with non-empty interior $\mathring{S}$, the good lattice assumption \ntb{presumes} that $Z$ extends to a $C^3$ random field $Z$ on $S$ the peaks of which can be identified with and have similar heights to the peaks of $Z$ on the lattice $\mathcal{L}$.} Let 
\begin{equation*}
	\nabla{Z(s)} = \left(\frac{\partial{Z(s)}}{\partial{s_1}},...,\frac{\partial{Z(s)}}{\partial{s_D}}\right),\quad \text{and}\quad
	\nabla^2{Z(s)} = \left(\frac{\partial{Z(s)}}{\partial{s_{ij}}}\right)_{1\leq i,j\leq D}.
\end{equation*}
Then the local maxima of $Z$ on $S$ are defined to be the points $s\in \mathring{S}$ such that $\nabla{Z}(s) = 0 \text{ and }\nabla^2{Z}(s)<0$.

\ntb{Because the domain is assumed to be continuous,} event that a local maximum is observed at a given $s \in \mathring{S}$ has probability zero. As such, in order to obtain a conditional peak height distribution, Palm distributions must be used (see \cite{cheng2015distribution} for details). For $ u \in \mathbb{R}$, they provide formulae to calculate
\begin{align}
	\mathbb{P}[Z(s)>u\mid \nabla{Z}(s) = 0\ \text{and}\ \nabla^2{Z}(s)<0]. \label{eqn2.1}
\end{align}
In general these expressions are difficult to evaluate. However, under the assumption of isotropy, \cite{cheng2015explicit} showed that they can be obtained explicitly. Recently, \cite{cheng2020critical} extended these results to the case where the field arises as a diffeomorphic transformation of an isotropic field. Details of how to apply these methods to perform peak inference in fMRI data can be found in \cite{schwartzman2019peak}.

\section{Simulation Setup}
\label{sec3}
In this section we describe the different simulation settings we have considered in order to compare the performance of the three methods introduced in Section \ref{sec2}, i.e. ADLM, MCDLM \ntb{and} continuous RFT method. For each simulation setting we generate a large number of stationary Gaussian random fields (or $t$-fields), collect the heights of the peaks across all fields and combine these to obtain a reference peak height distribution, which will allow us to test the validity of each of the methods.

For each method, we calculate a $p$-value at each peak \ntb{with respect to the corresponding peak height distribution}. Local maxima \ntb{are} selected based on the criteria that their height values are
larger than their neighbors - we \nt{shall} consider both the fully connected and partially
connected neighborhoods. We compare the validity and accuracy of these $p$-values using $pp$ plots which compare the sorted $p$-values to the tail probability of the true peak height distribution. These are formally defined in Appendix \ref{appendix.d} for clarity. The closer \nt{the curve in each} plot is to the identity function, the closer the approximation is to the true distribution. \ntb{Curves} lying below the identity correspond to conservative $p$-values and \nt{lines} above the identity function correspond to liberal $p$-values. We use
these $pp$ plots to compare the performance of the three approaches in all of our simulation studies.

\subsection{Isotropic Gaussian random fields}
\label{sec3.1}
Our first set of simulations consists of isotropic Gaussian fields which are obtained by convolving Gaussian white noise with \ntb{an isotropic} Gaussian kernel with specified FWHM and normalizing so that the resulting fields have unit variance. To avoid any boundary effects the fields \ntb{are} initially generated on a $D$-dimensional large grid of size $\left(50 + 2\times \lceil 4*\eta\rceil\right)$ at each direction and the central subset is taken, as described in \cite{davenport2020selective}. To simulate imaging data, we generate simulations in 2D and 3D. The resulting 2D images are of size $50 \times 50$ and the resulting 3D images are of size $50 \times 50 \times 50$. In the simulations we choose the FWHM so that the correlation between adjacent voxels in each perpendicular direction is equal to $\rho \in \left\lbrace 0.01, 0.5, 0.99 \right \rbrace$ - \nt{which correspond to FWHM of $0.7$, $1.5$ and $11.7$ voxels respectively. For now we assume that $\rho$ is known and calculate the neighbourhood covariance $\Sigma$, needed for MCDLM, using \eqref{eqn2.5}.} In each setting we generate $10{,}000$ random fields and compare the different approaches using $pp$ plots - as described above and in Appendix \ref{appendix.d}. See Section \ref{sec4.2} for the corresponding results. 

\nt{In our examples here and in the following sections, when applying MCDLM, we choose $M$ large enough to ensure that the number of resulting empirical samples is at least $N = 10^6$ for FWHM $< 11.7$ and at least $N = 2\times 10^5$ for FWHM $= 11.7$. We also generate a look-up table, which provides the same results under reduced computation time (for the case of Gaussian white noise smoothed with an isotropic Gaussian kernel). The results of using the lookup table are shown in Appendix \ref{appendix.d3}.}

\subsection{Isotropic $t$-fields}
\label{sec3.2}
\nt{For our second \ntb{set} of simulations}, we consider the performance of the different approaches when it comes to evaluating the height distribution of peaks of $t$-fields. To do so we generate fields with $\nu$ degrees of freedom (taking $\nu = 20, 50, 100$) by simulating i.i.d isotropic Gaussian fields $Z_1,...,Z_\nu$ and $\varepsilon(s)$ in (\ref{eqn.t}) as in Section \ref{sec3.1}. In each setting we generate $10{,}000$ $t$-fields in both 2D and 3D and calculate peak-height $p$-values using the  MCDLM approach for $t$-fields discussed in Section \ref{sec2.3}. 

We also calculate $p$-values using the continuous RFT approach; note that this is designed for Gaussian random fields so we would not expect it to work as well in this setting. \nt{To help understand how close continuous RFT is to simulations from a Gaussian random field we compute a third set of $p$-values using simulated Gaussian random fields. As the number of degrees of freedom of the $t$-field increases the $t$-field converges to a Gaussian random field with the original covariance function so this comparison also allows us to understand how close the fields are to the theoretical limit. Note that in practice this third approach is not a viable measure for calculating a peak height distribution as it requires generating a large number of fields which is extremely computationally expensive but we include it in the plots as a reference. We do not compare to ADLM in this setting as it is only designed for Gaussian fields.} We compare the $p$-values obtained using these different approaches using $pp$-plots. The results are described in Section \ref{sec4.3}.

As discussed in Section \ref{sec2.4}, we also consider \nt{the use of a} Gaussianization transformation of the $t$-field to improve the computational efficiency. We perform the same set of simulations but where each of the $t$-fields is Gaussianized, with the same $t$-field as described above. We then calculate $p$-values using the MCDLM for Gaussian fields and the continuous RFT approach. The results are described in Appendix \ref{SS:gauss}.

\subsection{Stationary Gaussian fields with known nonseparable covariance}
Since MCDLM works for a general mean-zero stationary field, in this third set of simulations, we examine the performance of MCDLM on stationary Gaussian fields with known nonseparable covariance. To obtain such fields, we start by smoothing 2D Gaussian white noise using an elliptical Gaussian kernel with different smoothing FWHM along each axis. Next, we smooth a second 2D Gaussian white noise using another elliptical Gaussian kernel, this time with the FWHM values swapped between the two axes. Finally, we calculate the voxelwise average of the two smoothed fields to obtain our desired fields. Let the correlations corresponding to the smoothing FWHM in the vertical and horizontal directions be denoted by $\rho_1$ and $\rho_2$. This pair of correlations can be interchangeable because of the two smoothing steps introduced. We consider two sets of parameters for this set of simulations, the first with $\rho_1 = 0.01$ and $\rho_2 = 0.5$, and the second with $\rho_1 = 0.5$ and $\rho_2 = 0.99$. For the neighborhood covariance, we calculate it using $\operatorname{Var}\left(\frac{1}{2}\left\{Z_1(s)+Z_2(s)\right\}\right) = \frac{1}{4}\left(\Sigma_1 + \Sigma_2\right)$, where $\Sigma_1$ and $\Sigma_2$ are the neighborhood covariances for the two fields obtained in intermediate steps and can be calculated using equation \eqref{eqn2.8}. Similar to the previous experiments, we generate $10,000$ fields and validate the MCDLM performance using $pp$ plots. The results are shown in Section \ref{sec.nonseparable}.

\subsection{Stationary Gaussian fields with unknown covariance}
\label{sec3.3}
\begin{figure}[!htp]
\centering
\includegraphics[width = 8cm]{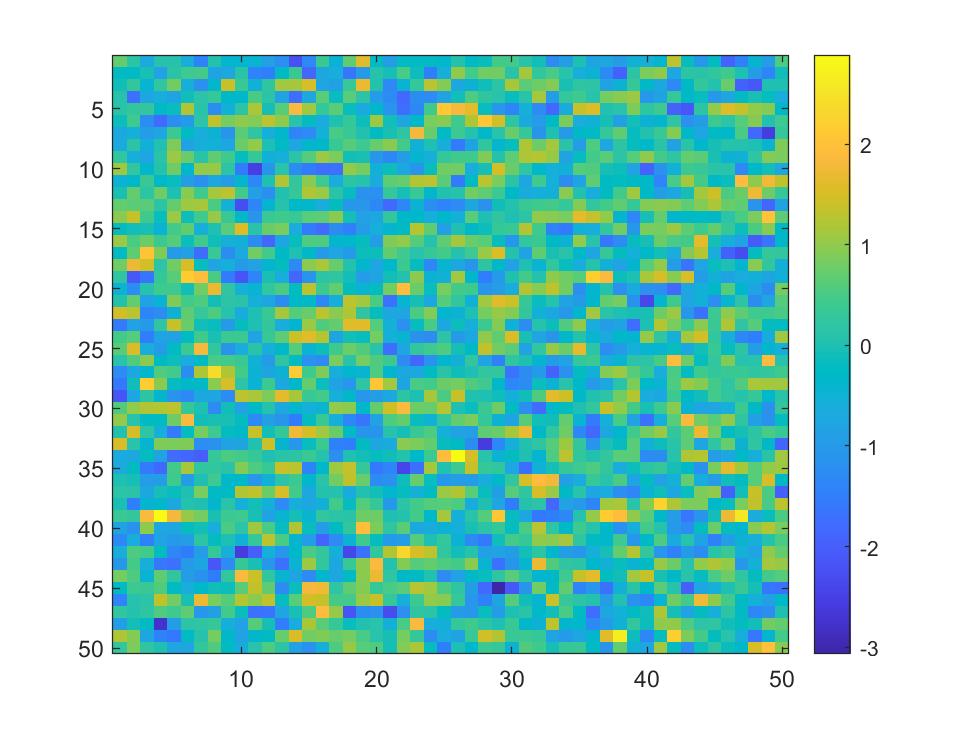}
\includegraphics[width = 8cm]{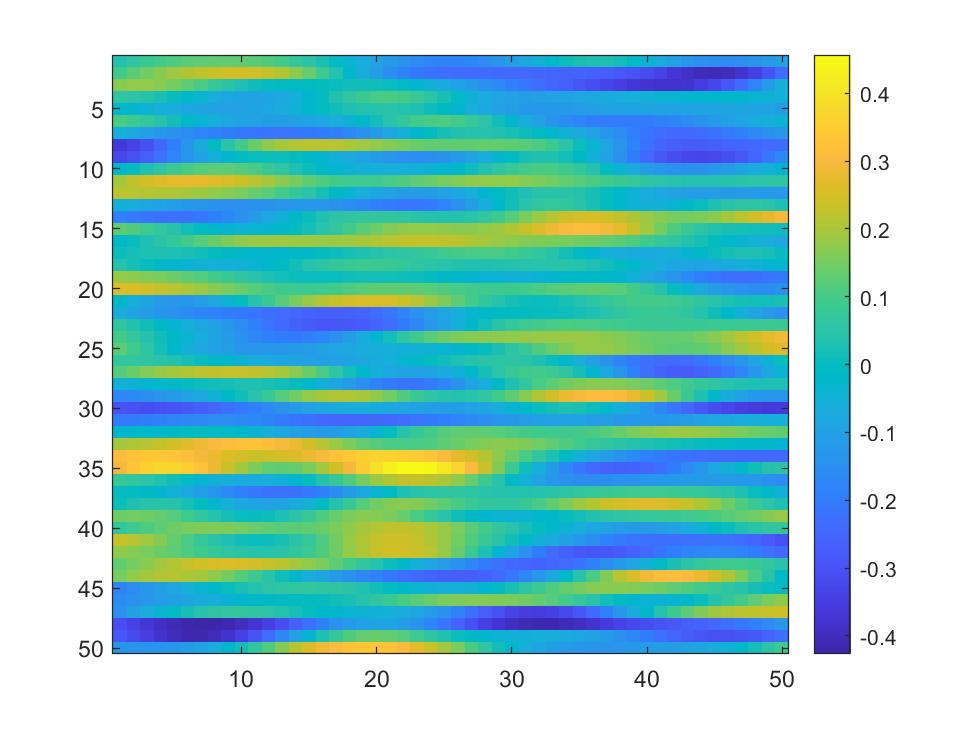}
\caption{Examples of stationary Gaussian random fields which are obtained by convolving white noise with an elliptical Gaussian kernel with $\rho_1 = 0.01, \rho_2 = 0.5$ in the left plot and $\rho_1 = 0.5, \rho_2 = 0.99$ in the right plot.\label{fig16}}
\end{figure}
\nt{To study the performance of MCDLM when the covariance is unknown, and must therefore be estimated we consider several further simulation settings.} Firstly we use $n$ fields ($n = 20, 50, 100, 200$ for small $\rho$ and $n = 20, 50, 100, 200, 1000$ for large $\rho$) to estimate the neighborhood covariance using the isotropic version of equation (\ref{eqn5.1}). We study the performance of MCDLM with this estimated covariance across different sample sizes (\nt{running $10{,}000$ simulations} in each of the settings described in Section \ref{sec3.1}). Secondly we consider 2D non-isotropic Gaussian fields. To generate these we smooth Gaussian white noise with an elliptical Gaussian kernel with smoothing FWHM in each direction chosen such that the correlation between adjacent voxels in the vertical and horizontal directions is $\rho_1$ and $\rho_2$ respectively. We estimate the neighborhood covariance using equation (\ref{eqn2.8}) using $n$ fields ($n = 20, 50, 100, 200$) \nt{and then apply MCDLM to generate peak $p$-values.} We consider two \nt{2D non-isotropic scenarios}, one where $\rho_1 = 0.01$ and $\rho_2 = 0.5$ and a second where $\rho_1 = 0.5$ and $\rho_2 = 0.99$ (example realizations of these fields are \nt{shown} in Figure \ref{fig16}). The results are shown in Section \ref{sec4.4}. 

\nt{For each of these simulations we estimate $\Sigma$ from the data as described in Section \ref{sec2.5}.  We take advantage of stationarity in order to do so as there is a lot of structure that can be taken advantage of when estimating $\Sigma$. In particular the neighborhood covariance matrix has a block Toeplitz structure which makes it easier to estimate, see the examples in Appendix \ref{SS:cov}.  As a pracitcal note, we observe that once the covariance function has been estimated we must ensure that it is positive semi-definite (p.s.d) in order to simulate from it. In order to ensure this we push the negative eigenvalues of estimated covariance matrix to a small positive value, \num{1e-10}.}

\section{Simulation Results}
\label{sec4}
\subsection{Results for isotropic Gaussian random fields}
\label{sec4.2}

\begin{figure}[!htp]
\centering
\begin{sideways}
\phantom{------------------}2D
\end{sideways}
\includegraphics[trim=80 5 80 5, clip,width=0.3\textwidth]{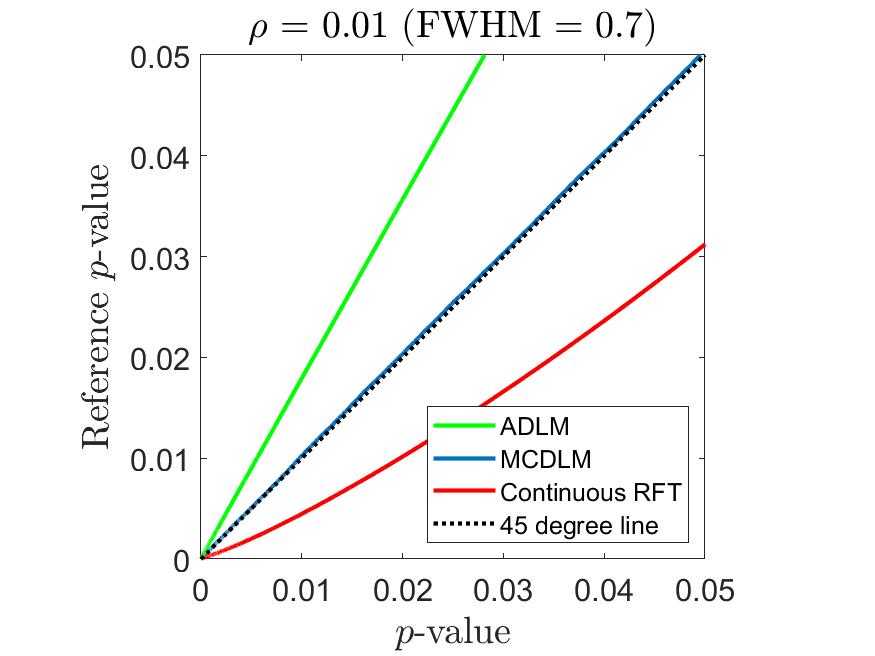}
\includegraphics[trim=80 5 80 5, clip,width=0.3\textwidth]{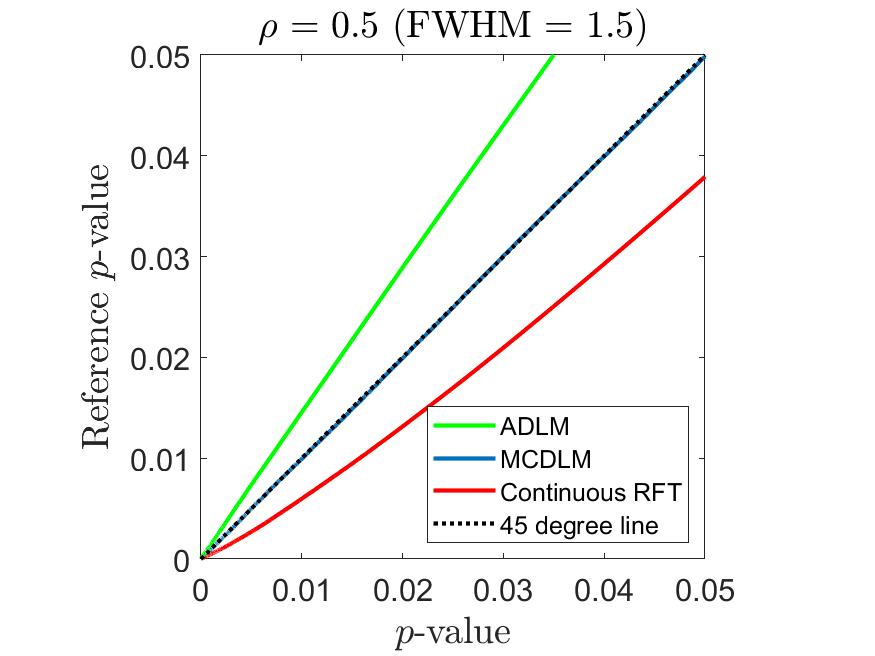}
\includegraphics[trim=100 5 100 5, clip,width=0.3\textwidth]{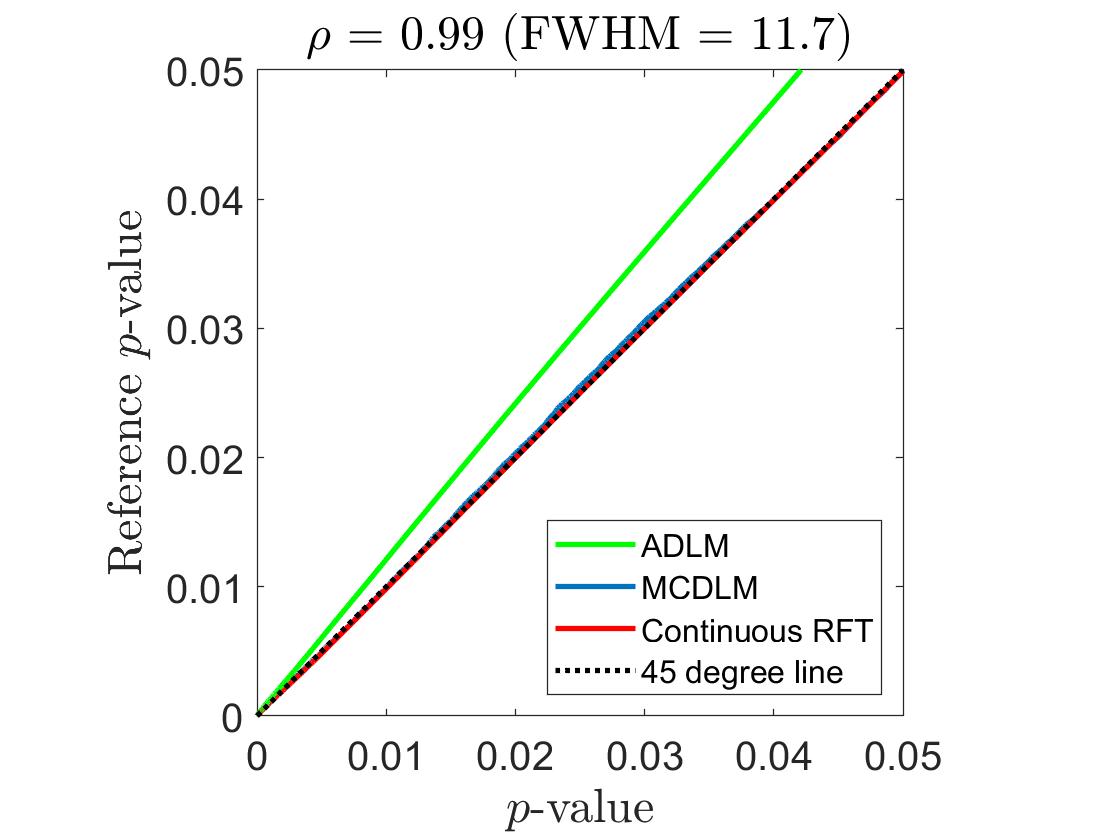}

\begin{sideways}
\phantom{------------------}3D
\end{sideways}
\includegraphics[trim=80 5 80 5, clip,width=0.3\textwidth]{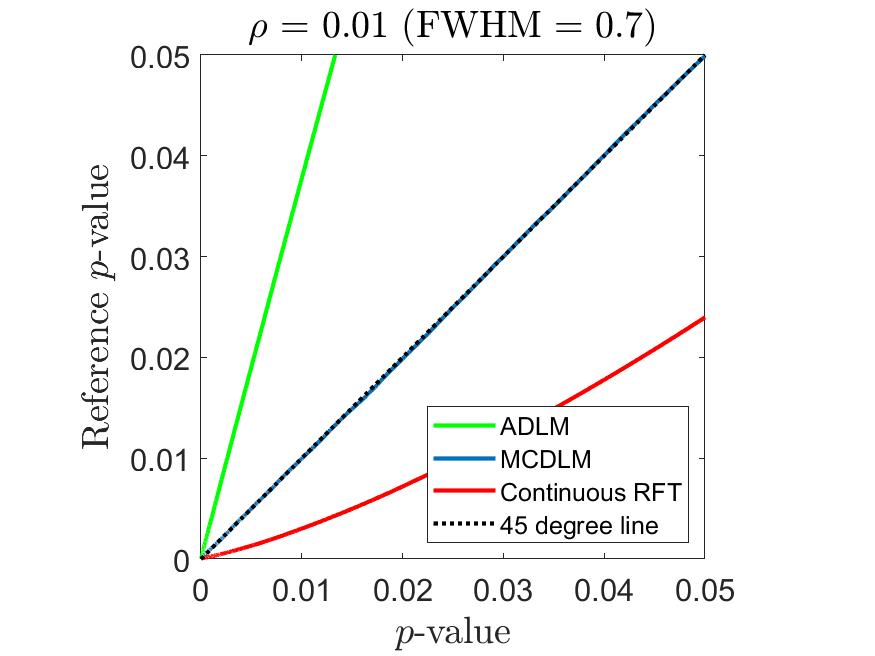}
\includegraphics[trim=80 5 80 5, clip,width=0.3\textwidth]{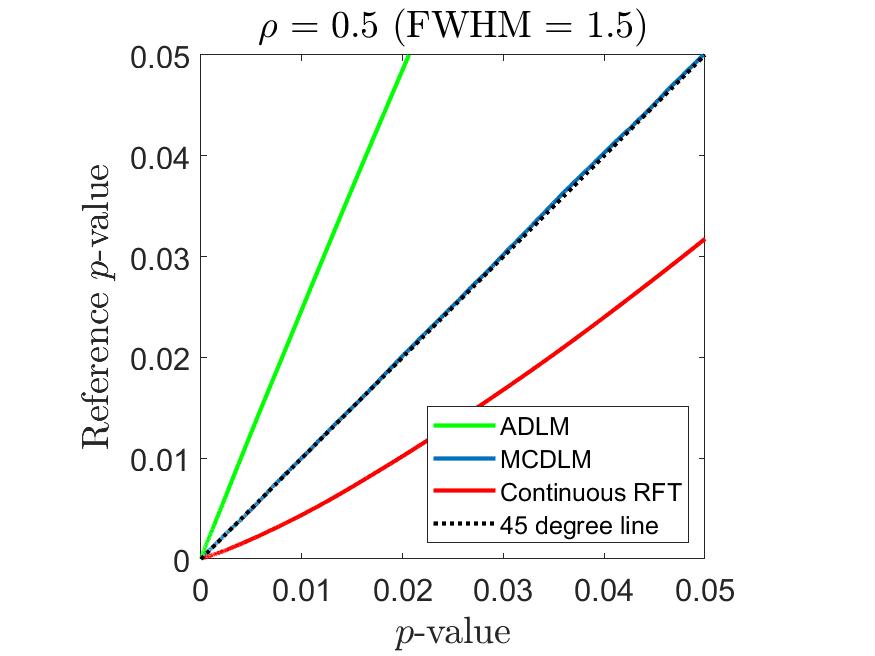}
\includegraphics[trim=100 5 100 5, clip,width=0.3\textwidth]{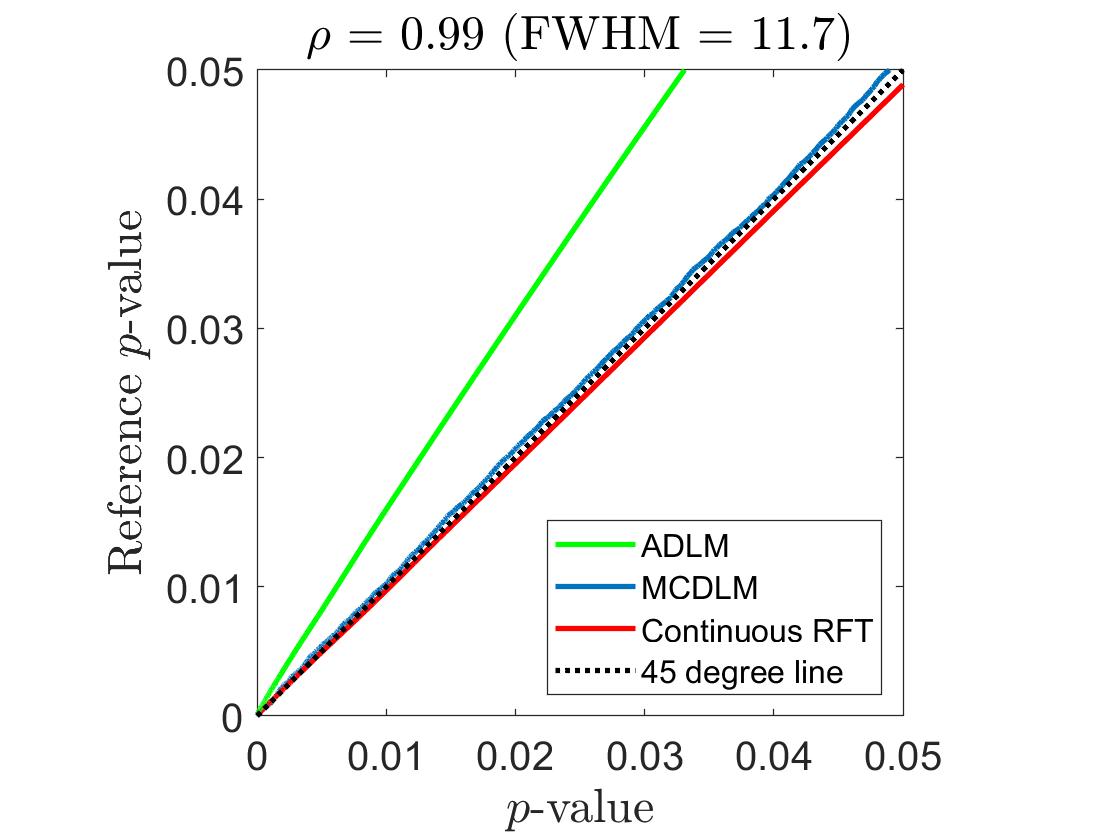}
\caption{$pp$ plots which compare the different methods of computing peak height $p$-values in the isotropic Gaussian random field scenario. 2D and 3D results are displayed in the first and second rows respectively. The correlations between adjacent voxels are $\rho = 0.01, 0.5, 0.99$. The plots compare the performance of ADLM, MCDLM and the continuous RFT approach.\label{fig7}}
\end{figure}

Results comparing all three methods in the isotropic Gaussian random fields setting (described in Section \ref{sec3.1}) are presented in Figure \ref{fig7} for the fully connected neighborhood. From this figure we see that the MCDLM method obtains $p$-values which are uniformly distributed and thus provides accurate and valid inference
at all smoothness levels. The continuous RFT method is valid but conservative unless the data is very smooth, i.e., with a large smoothing FWHM. The ADLM method gives liberal $p$-values at all \nt{considered} smoothness levels though the severity of this reduces when the smoothness is very large. In 3D the results are similar though at the highest smoothness level the curve corresponding to the MCDLM method is slightly rough, this could be made sharper if desired by increasing the number of Monte Carlo runs used. In these results we use the covariance function (\ref{eqn2.7}) for both MCDLM and ADLM distributions since it is the actual covariance of the data, as discussed in Section \ref{SScorrfn}. The results are slightly worse if \eqref{eqn2.32} is used instead (as was done in \cite{worsley2005improved, taylor2007maxima}), see Appendix \ref{appendix.d2}. \nt{Results for partially connected neighborhoods, in which ADLM works comparably to MCDLM are presented in Appendix \ref{appendix.d1}.}

To quantify the difference of $p$-values from all three approaches, we use the mean ratio between the $p$-value curves from each method and the identity line to measure their difference. Our comparison \ntb{focuses} on the region of $p$-value $\in(0.001, 0.05]$ because of research interest and computation precision. The mean ratio results calculated from 2D isotropic Gaussian random fields with different $\rho$ are shown in Table \ref{tab0}. From the table, MCDLM outperforms ADLM and continuous RFT in all low smoothness cases (FWHM $< 6.7$). However continuous RFT does better when smoothness level is high (FWHM $= 6.7, 8.3$ and $11.7$). We also perform a similar analysis using root mean squared error (RMSE) between the $p$-value curves from each method and the identity line, which shows the same conclusion. The results are shown in Appendix \ref{sec.rmse}.

\begin{table}[!htp]
\caption{Mean ratio results from 2D isotropic Gaussian random fields for comparing the $p$-values from MCDLM, ADLM and continuous RFT approaches. The smallest value in each row is highlighted in red color. \label{tab0}}
\centering
\begin{tabular}{llll}
\hline
 & MCDLM & ADLM & Continuous RFT \\ \hline
$\rho = 0.01$ (FWHM $ = 0.7$) & \color{red}\num{0.98} & \num{0.55} & \num{2.05} \\ 
$\rho = 0.1\ $ \ (FWHM $ = 1$) & \color{red}\num{0.96} & \num{0.57} & \num{1.98} \\ 
$\rho = 0.3\ $ \ (FWHM $ = 1.2$) & \color{red}\num{0.99} & \num{0.62} & \num{1.82} \\ 
$\rho = 0.5\ $ \ (FWHM $ = 1.5$) & \color{red}\num{1.01} & \num{0.68} & \num{1.62} \\ 
$\rho = 0.7\ $ \ (FWHM $ = 2$) & \color{red}\num{1.01} & \num{0.72} & \num{1.36} \\ 
$\rho = 0.9\ $ \ (FWHM $ = 3.6$) & \color{red}\num{0.99} & \num{0.78} & \num{1.10} \\ 
$\rho = 0.95$ (FWHM $ = 5.2$) & \color{red}\num{0.99} & \num{0.80} & \num{1.05} \\ 
$\rho = 0.96$ (FWHM $ = 5.8$) & \color{red}\num{0.98} & \num{0.80} & \num{1.03} \\ 
$\rho = 0.97$ (FWHM $ = 6.7$) & \num{0.97} & \num{0.78} & \color{red}\num{0.99} \\
$\rho = 0.98$ (FWHM $ = 8.3$) & \num{0.97} & \num{0.80} & \color{red}\num{0.10} \\ 
$\rho = 0.99$ (FWHM $ = 11.7$) & \num{0.96} & \num{0.80} & \color{red}\num{0.98} \\ \hline
\end{tabular}
\end{table}

\subsection{Isotropic $t$-fields}
\label{sec4.3}

The $pp$ plots comparing the different methods in the setting of isotropic $t$-fields (described in Section \ref{sec3.2}) are presented in Figure \ref{fig11} (2D) and Figure \ref{fig12} (3D). As shown in these two figures, MCDLM obtains $p$-values which are uniformly distributed at all smoothness levels and different degrees of freedom for 2D and lower smoothness levels for 3D. In the 3D case at high smoothness levels, the MCDLM approach becomes a rough approximation because the number of peaks generated is not sufficient. As such the height distribution computed is slightly inaccurate, as shown in the noisy subfigure in the bottom right of Figure \ref{fig12}. \nt{The continuous method is designed for Gaussian fields rather than $t$-fields, so it is liberal when $\nu$ is small. Because it targets continuous peaks, it is conservative when $\nu$ is large and $\rho$ is small. MCDLM outperforms the continuous method in all considered scenarios as long as sufficeintly many local maxima are used in the Monte Carlo simulations.}

As discussed in Section \ref{sec3.2}, a Gaussianization approach \nt{can be used} to save computation time (see Appendix \ref{appendix.g} for details of computation time). The results for the Gaussianization approach are shown in Appendix \ref{SS:gauss}. They show that the MCDLM method performs well when the number of degrees of freedom is sufficiently large \nt{at each smoothness level} whereas continuous method requires both the degrees of freedom and FWHM to be large.

\begin{figure}[!htp]
\centering
\begin{sideways}
\phantom{------------------}$\nu = 20$
\end{sideways}
\includegraphics[trim=70 5 100 5, clip,width=0.3\textwidth]{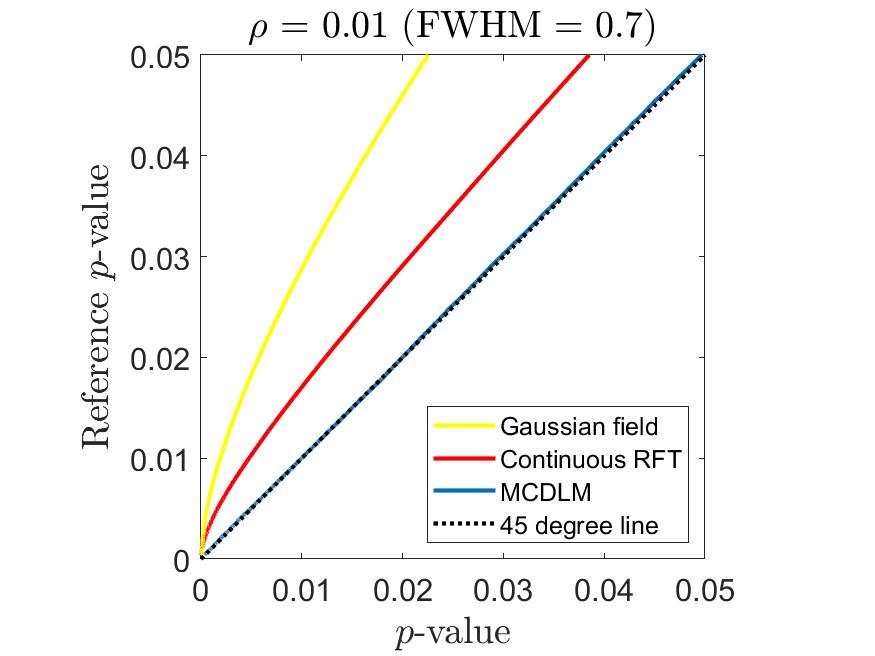}
\includegraphics[trim=70 5 100 5, clip,width=0.3\textwidth]{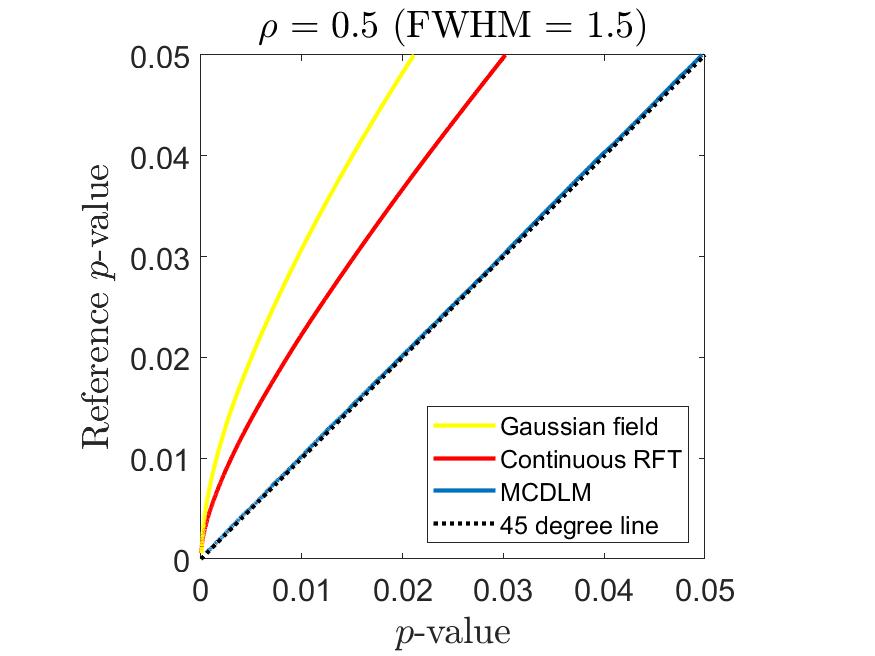}
\includegraphics[trim=70 5 100 5, clip,width=0.3\textwidth]{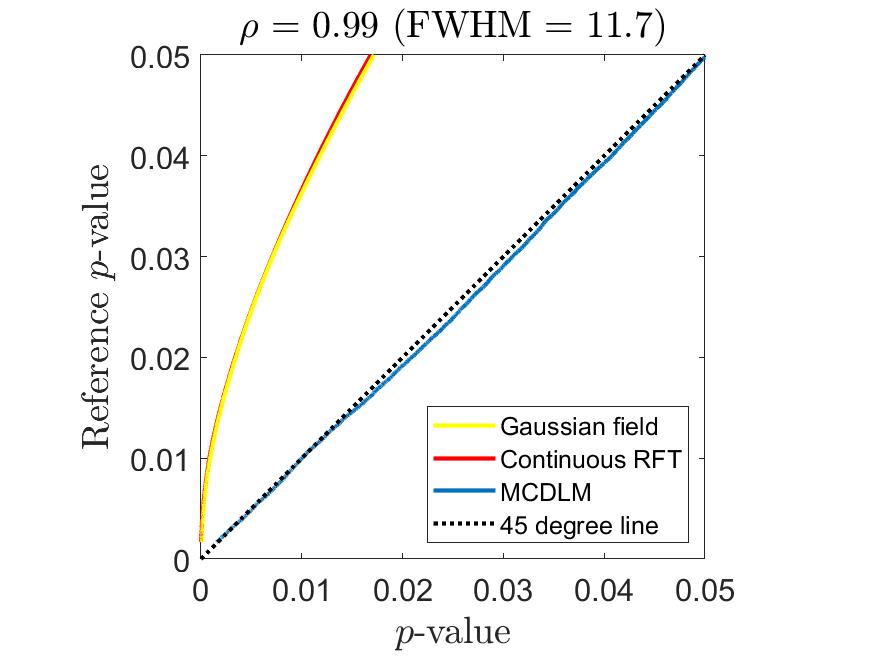}

\begin{sideways}
\phantom{------------------}$\nu = 50$
\end{sideways}
\includegraphics[trim=70 5 100 5, clip,width=0.3\textwidth]{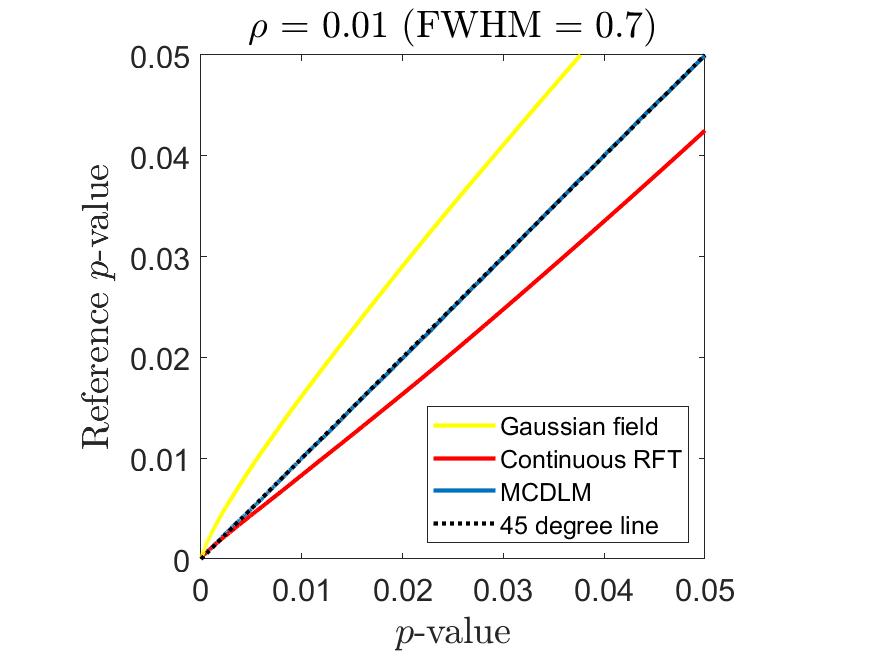}
\includegraphics[trim=70 5 100 5, clip,width=0.3\textwidth]{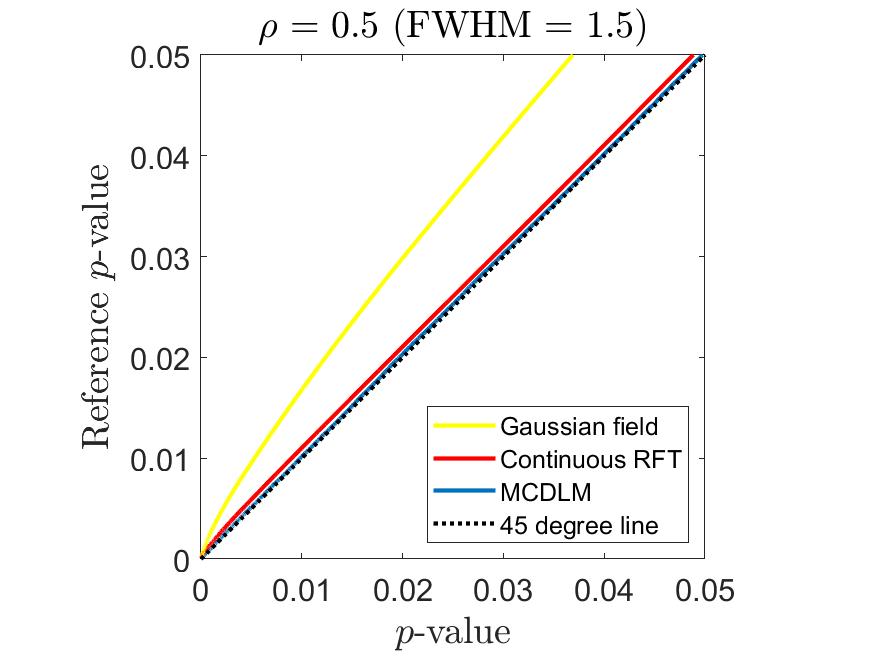}
\includegraphics[trim=70 5 100 5, clip,width=0.3\textwidth]{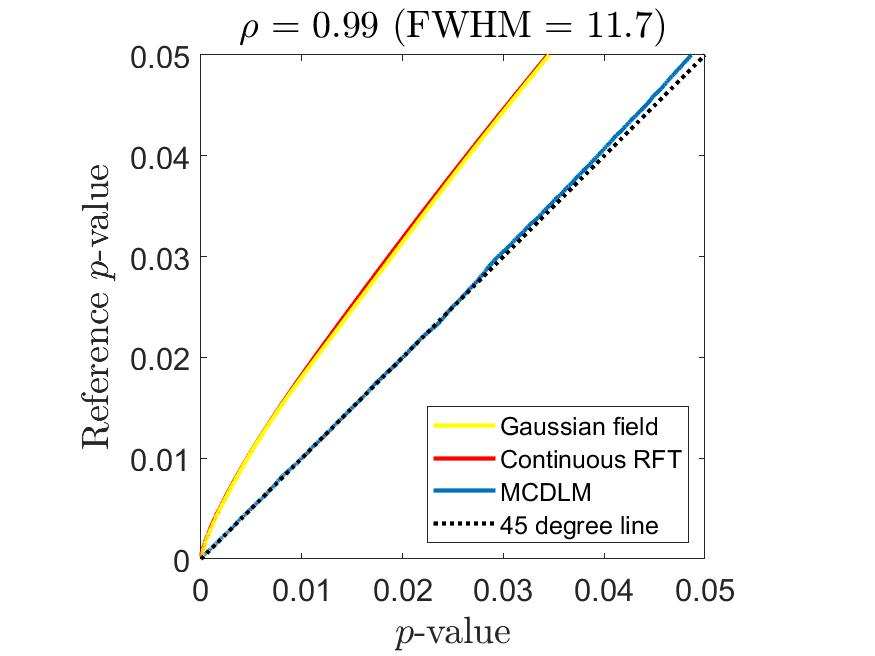}

\begin{sideways}
\phantom{------------------}$\nu = 200$
\end{sideways}
\includegraphics[trim=70 5 100 5, clip,width=0.3\textwidth]{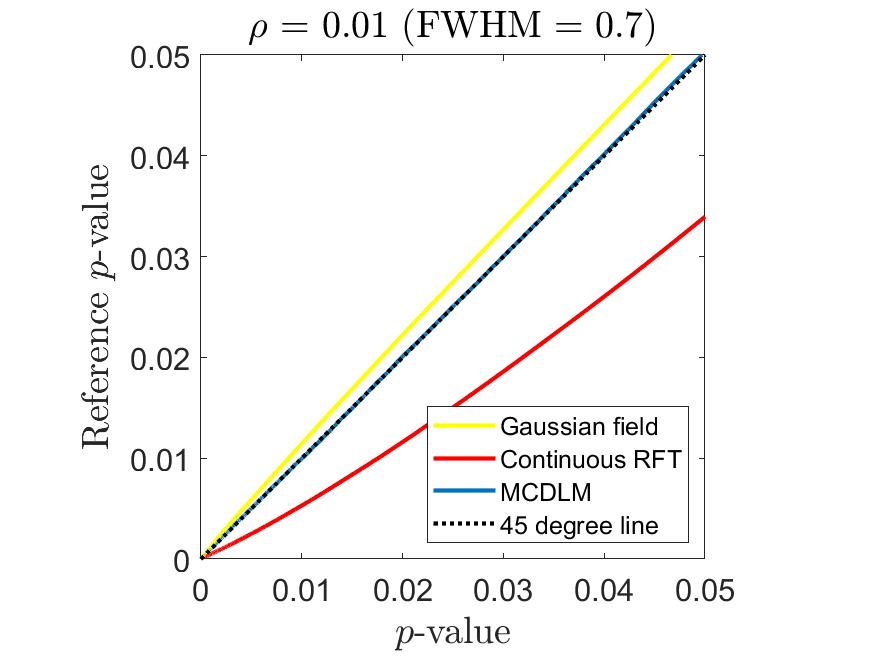}
\includegraphics[trim=70 5 100 5, clip,width=0.3\textwidth]{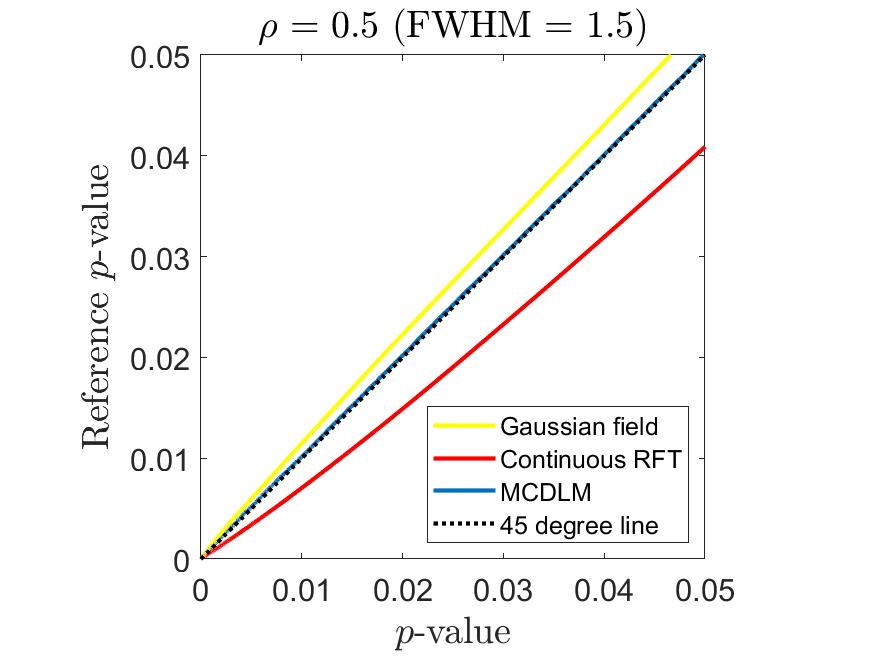}
\includegraphics[trim=70 5 100 5, clip,width=0.3\textwidth]{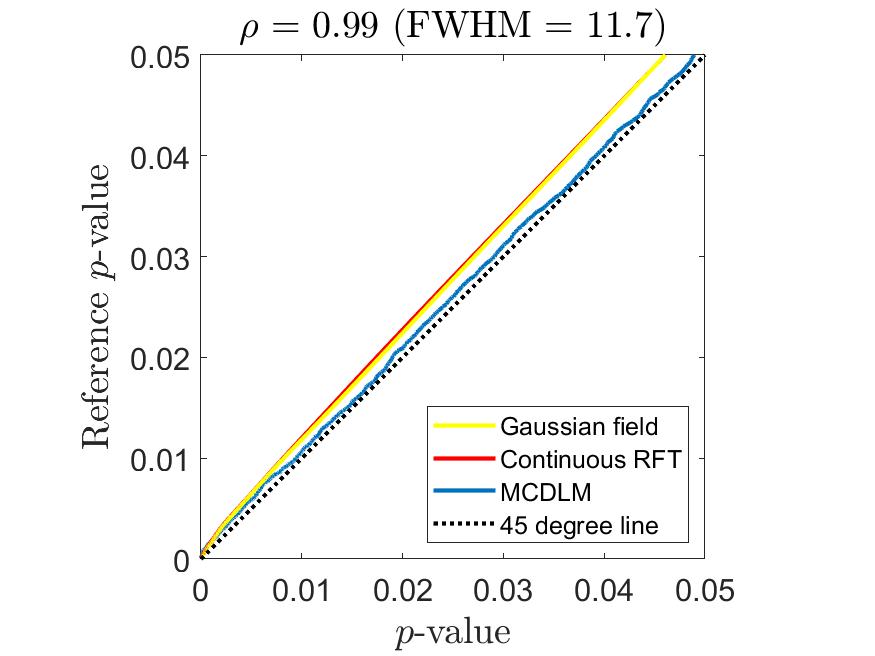}
\caption{Comparing methods for calculating the peak height distribution of a 2D $t$-field with $\nu$ degrees of freedom. From left to right: spatial correlation $\rho = 0.01, 0.5, 0.99$. From top to bottom: $\nu = 20, 50, 200$. The figure is generated based on the comparison of the $p$-values calculated using continuous RFT, MCDLM and Gaussian random fields with the same distribution as the Gaussian fields used to generate the $t$-fields. The reference for computing $p$-values is the true peak height distribution generated from the $t$-fields.\label{fig11}}
\end{figure}

\begin{figure}[!htp]
\centering
\begin{sideways}
\phantom{------------------}$\nu = 20$
\end{sideways}
\includegraphics[trim=70 5 100 5, clip,width=0.3\textwidth]{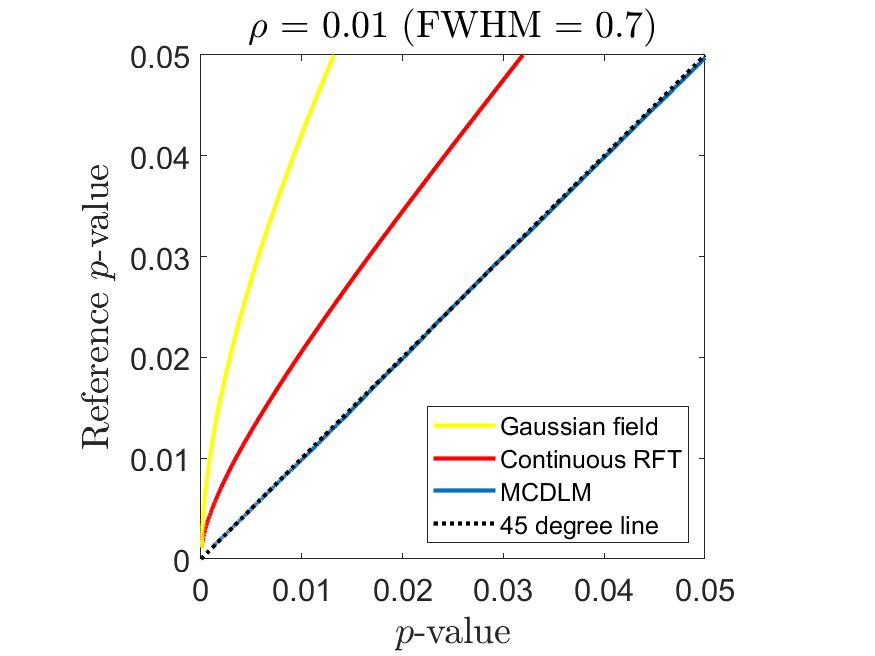}
\includegraphics[trim=70 5 100 5, clip,width=0.3\textwidth]{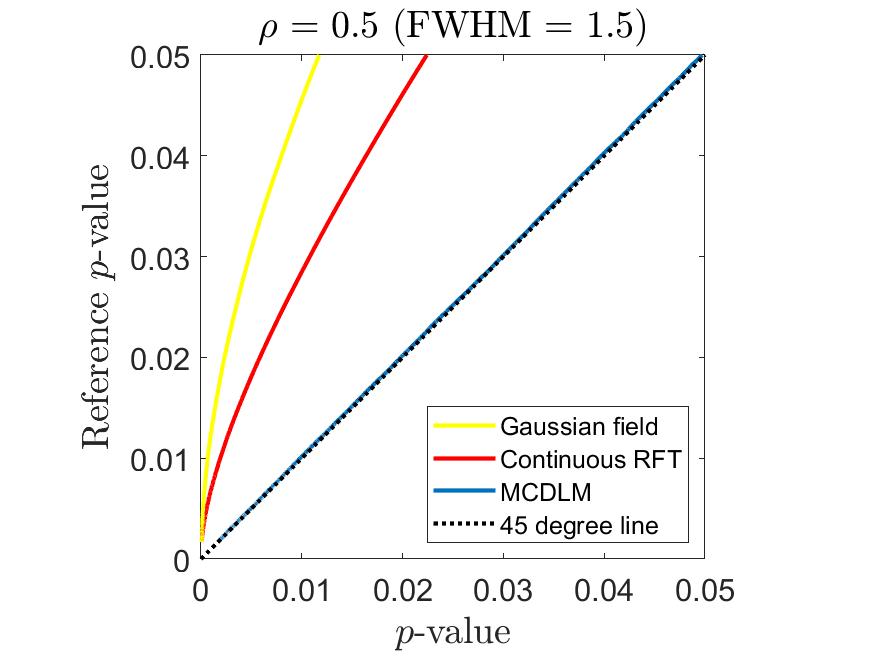}
\includegraphics[trim=70 5 100 5, clip,width=0.3\textwidth]{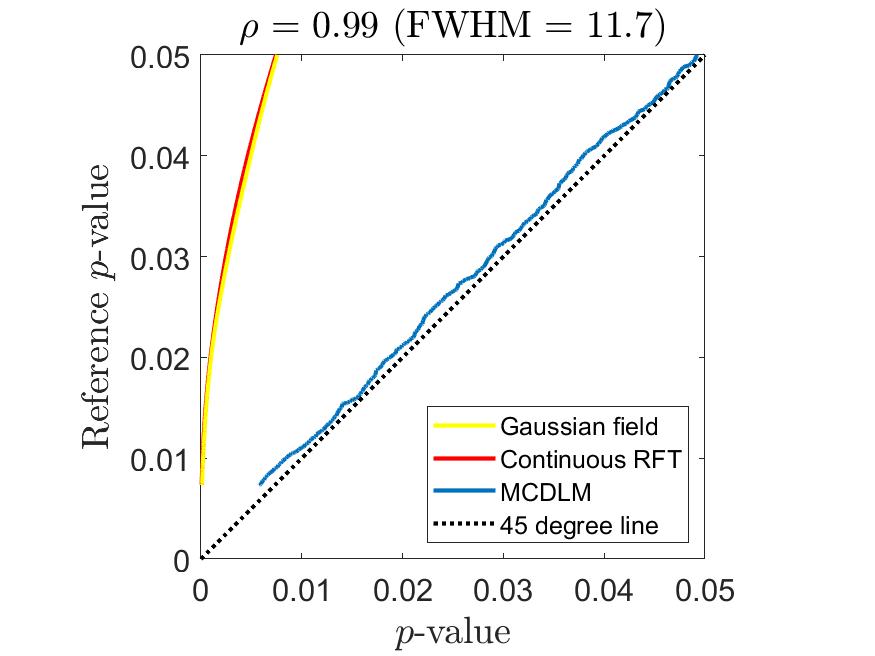}

\begin{sideways}
\phantom{------------------}$\nu = 50$
\end{sideways}
\includegraphics[trim=70 5 100 5, clip,width=0.3\textwidth]{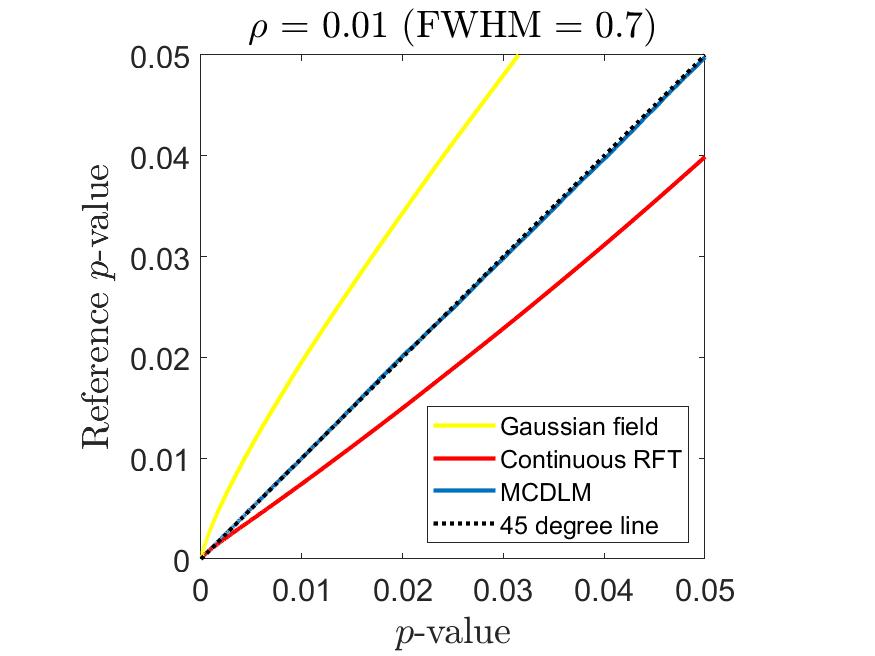}
\includegraphics[trim=70 5 100 5, clip,width=0.3\textwidth]{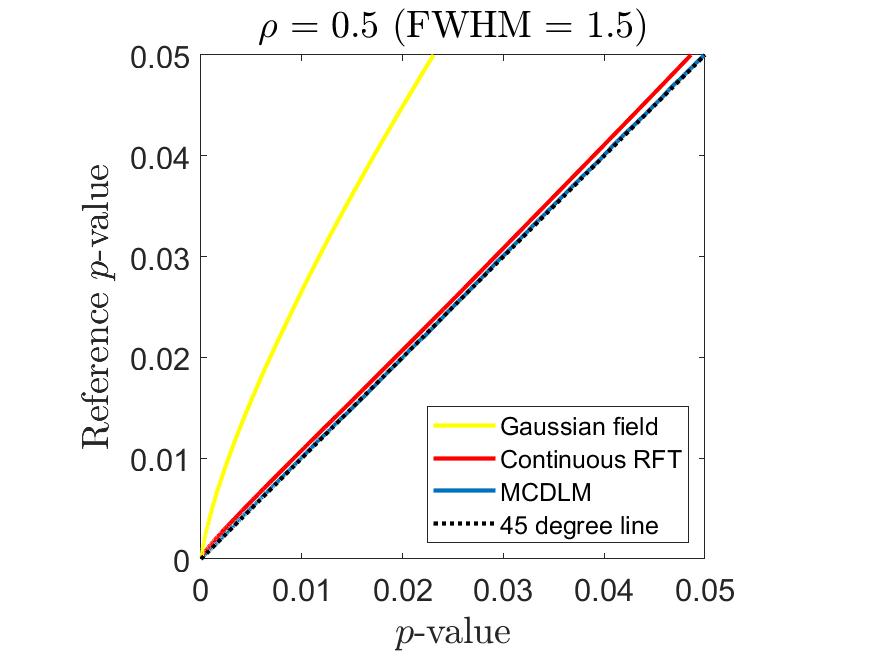}
\includegraphics[trim=70 5 100 5, clip,width=0.3\textwidth]{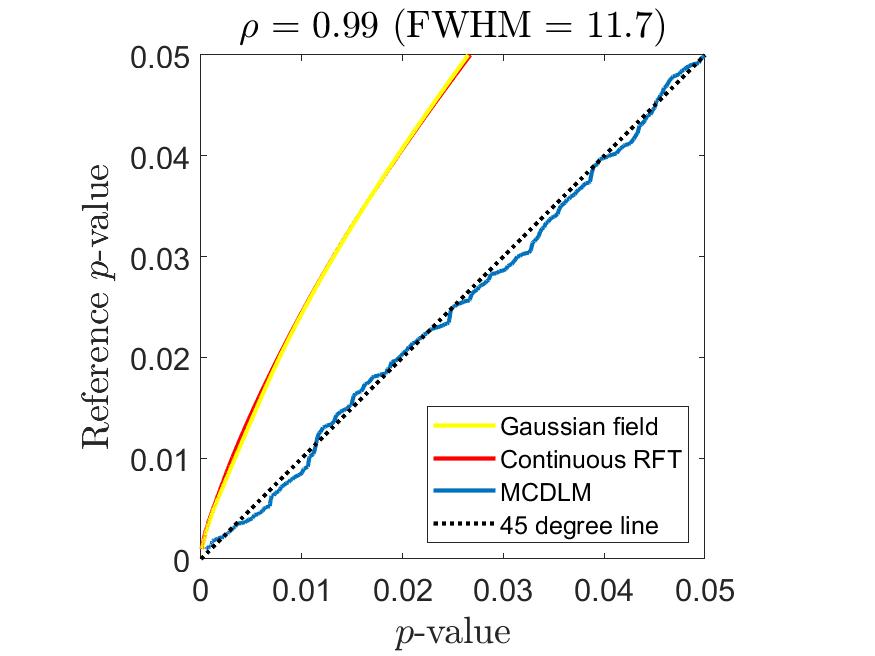}

\begin{sideways}
\phantom{------------------}$\nu = 200$
\end{sideways}
\includegraphics[trim=70 5 100 5, clip,width=0.3\textwidth]{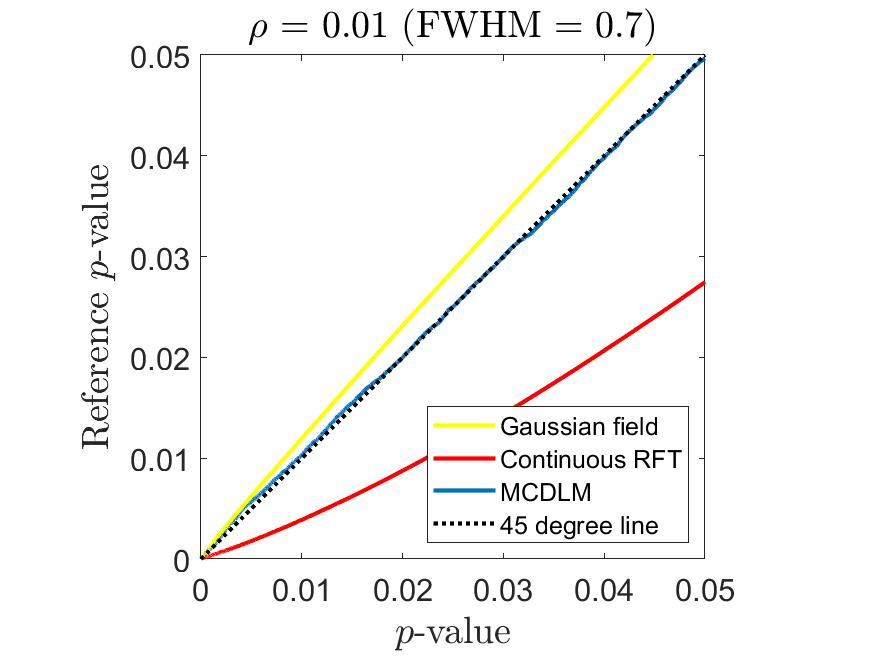}
\includegraphics[trim=70 5 100 5, clip,width=0.3\textwidth]{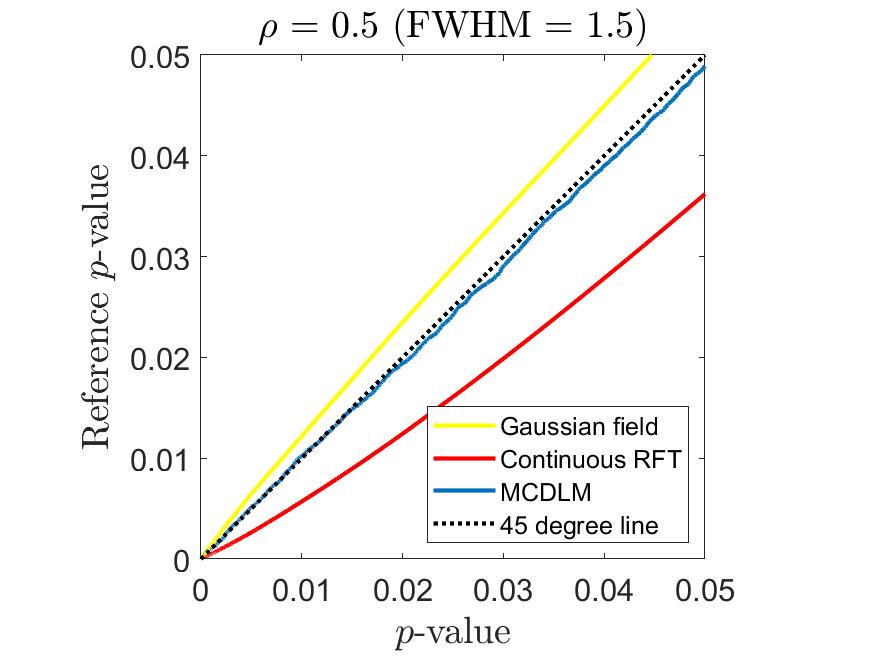}
\includegraphics[trim=70 5 100 5, clip,width=0.3\textwidth]{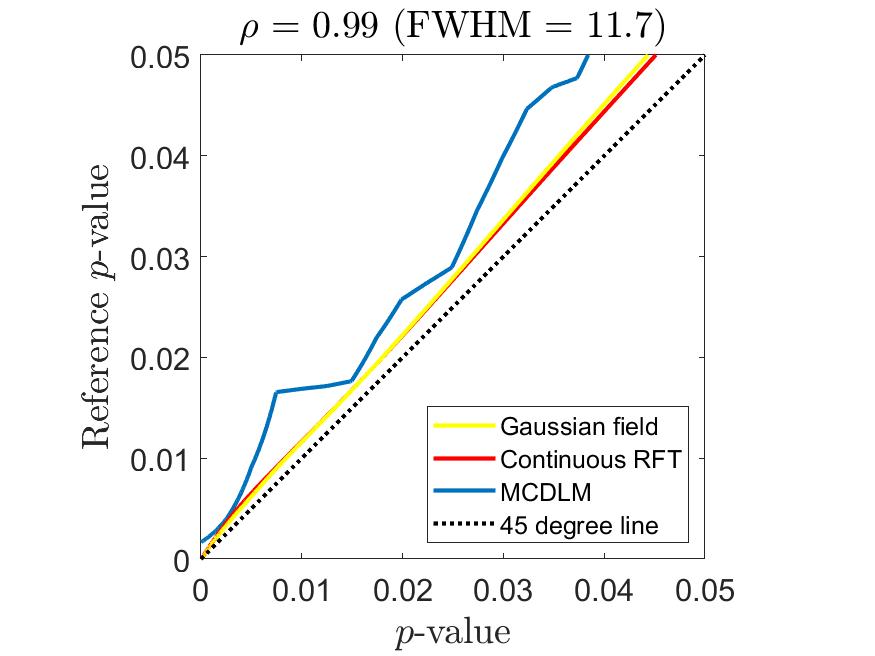}
\caption{Comparing methods for calculating the peak height distribution of a 3D $t$-field with $\nu$ degrees of freedom. From left to right: spatial correlation $\rho = 0.01, 0.5, 0.99$. From top to bottom: $\nu = 20, 50, 200$. The figure is generated based on the comparison of the $p$-values calculated using continuous RFT, MCDLM and Gaussian random fields with the same distribution as the Gaussian fields used to generate the $t$-fields. The reference for computing $p$-values is the true peak height distribution generated from the $t$-fields. \label{fig12}}
\end{figure}

\subsection{Stationary Gaussian fields with known nonseparable covariance} \label{sec.nonseparable}

Figure \ref{figure.nonsep} shows the results of MCDLM method with the theoretical neighborhood covariance when applied to stationary Gaussian fields with nonseparable covariance. In both simulations the MCDLM performs very well. These results validate that the MCDLM can work for a general mean-zero stationary field, even when the field is non-isotropic and the covariance structure is nonseparable. Under such conditions, the continuous RFT method fails, as it is only limited to isotropic fields. 

\begin{figure}[!htp]
\centering
\includegraphics[trim=100 5 100 5, clip,width=0.35\textwidth]{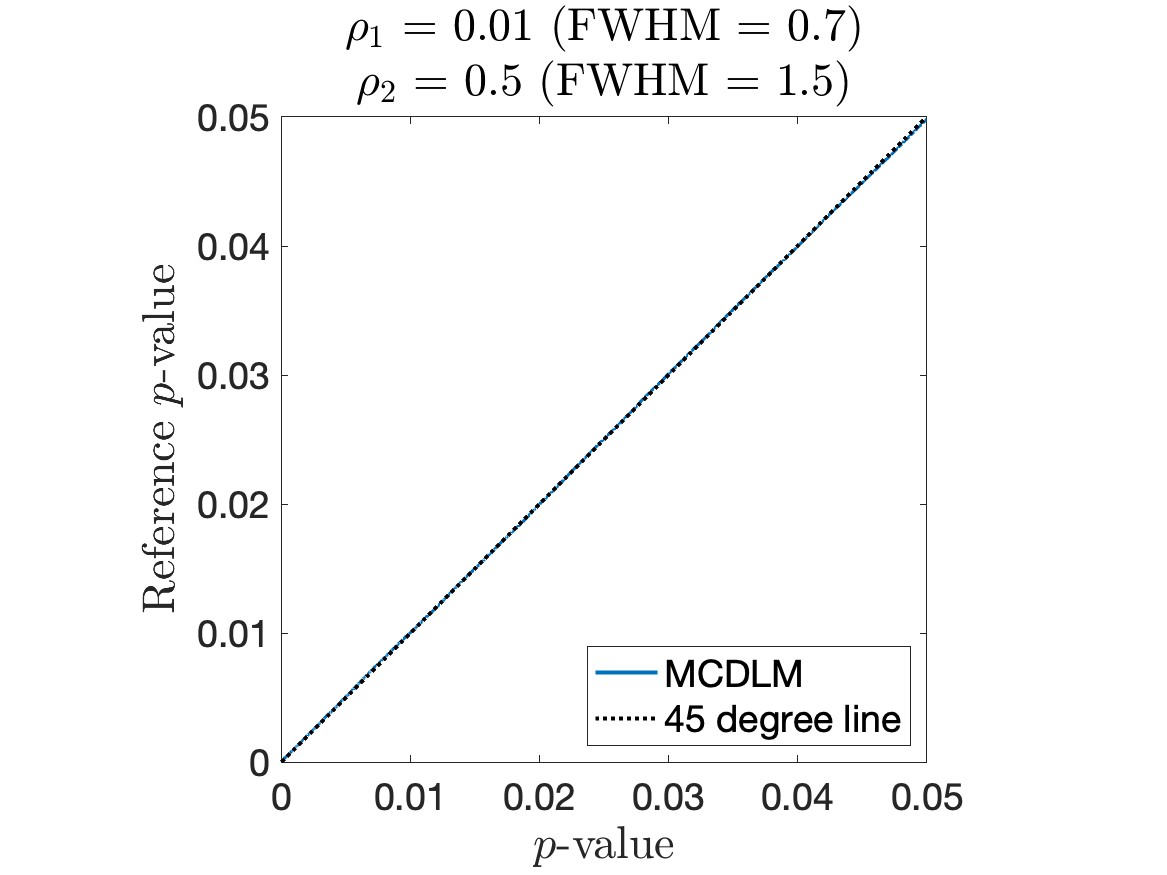}
\includegraphics[trim=100 5 100 5, clip,width=0.35\textwidth]{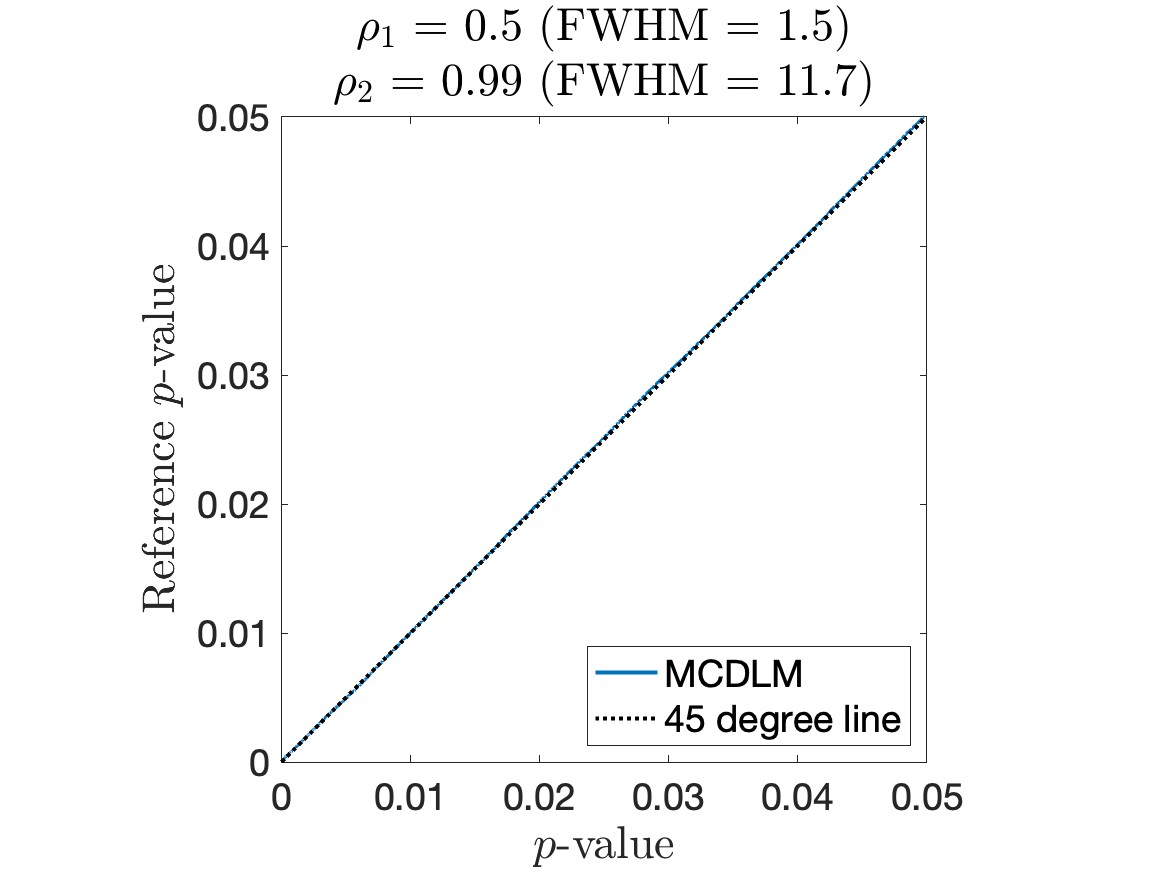}
\caption{Peak height distribution calculated from using MCDLM method for 2D stationary Gaussian fields with known nonseparable covariance. The left: $\rho_1 = 0.01$ and $\rho_2 = 0.5$ and the right: $\rho_1 = 0.5$ and $\rho_2 = 0.99$}
\label{figure.nonsep}
\end{figure}

\subsection{Stationary Gaussian fields with unknown covariance}
\label{sec4.4}

In Figure \ref{fig13}, we compare the peak height distribution calculated from the MCDLM method using both \ntb{true} neighborhood covariance (\ref{eqn2.7}) and estimated neighborhood covariance (\ref{eqn5.1}). Figure \ref{fig13} shows that MCDLM with the estimated neighborhood covariance essentially performs as well as MCDLM with the theoretical neighborhood covariance when $\rho = 0.01$ and $0.5$. When $\rho$ increases to 0.99, the MCDLM method with estimated covariance function requires a large number of simulated peaks before it \ntb{gives an accurate approximation}. This number decreases with the number of voxels in the image (which is why the performance of the 3D simulations are substantially better than the 2D ones), even when a very large sample size is used to estimate it. Since with $\rho = 0.99$, even with $1000$ instances to estimate the neighborhood covariance the MCDLM method still performs poorly, we investigate additional scenarios in which $\rho = 0.9, 0.93, 0.95$. The detailed results are included in Appendix \ref{appendix.d5}. Based on the results, we recommend using the MCDLM method with estimated covariance function when $\rho < 0.97$, or FWHM $< 6.7$ in practice.

The results for the second set of (non-isotropic) simulations discussed in Section \ref{sec3.3} are shown in Figure \ref{fig15}. From this figure we see that the estimated version works well when $\rho_1 = 0.01$ and $\rho_2 = 0.5$, and requires a larger number of realizations to converge when $\rho_1 = 0.5$ and $\rho_2 = 0.99$.

\begin{figure}[!htp]
\centering
\begin{sideways}
\phantom{------------------}2D
\end{sideways}
\includegraphics[trim=80 5 80 5, clip,width=0.3\textwidth]{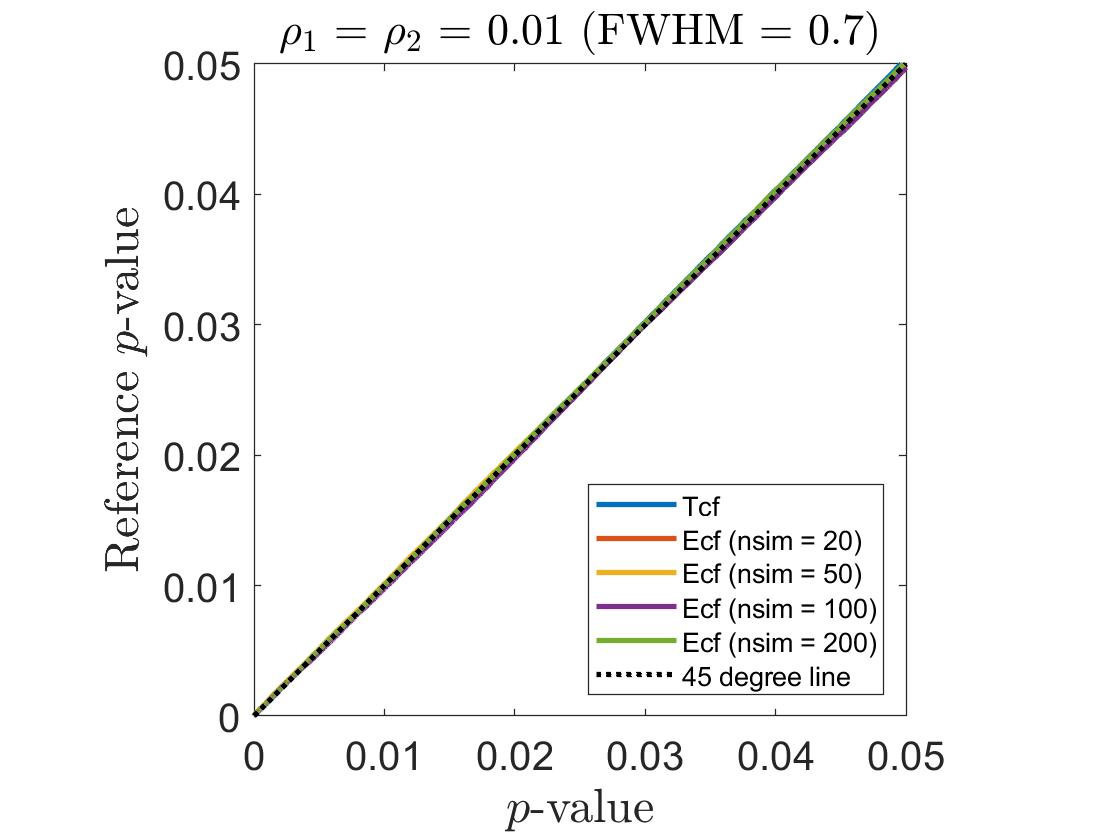}
\includegraphics[trim=80 5 80 5, clip,width=0.3\textwidth]{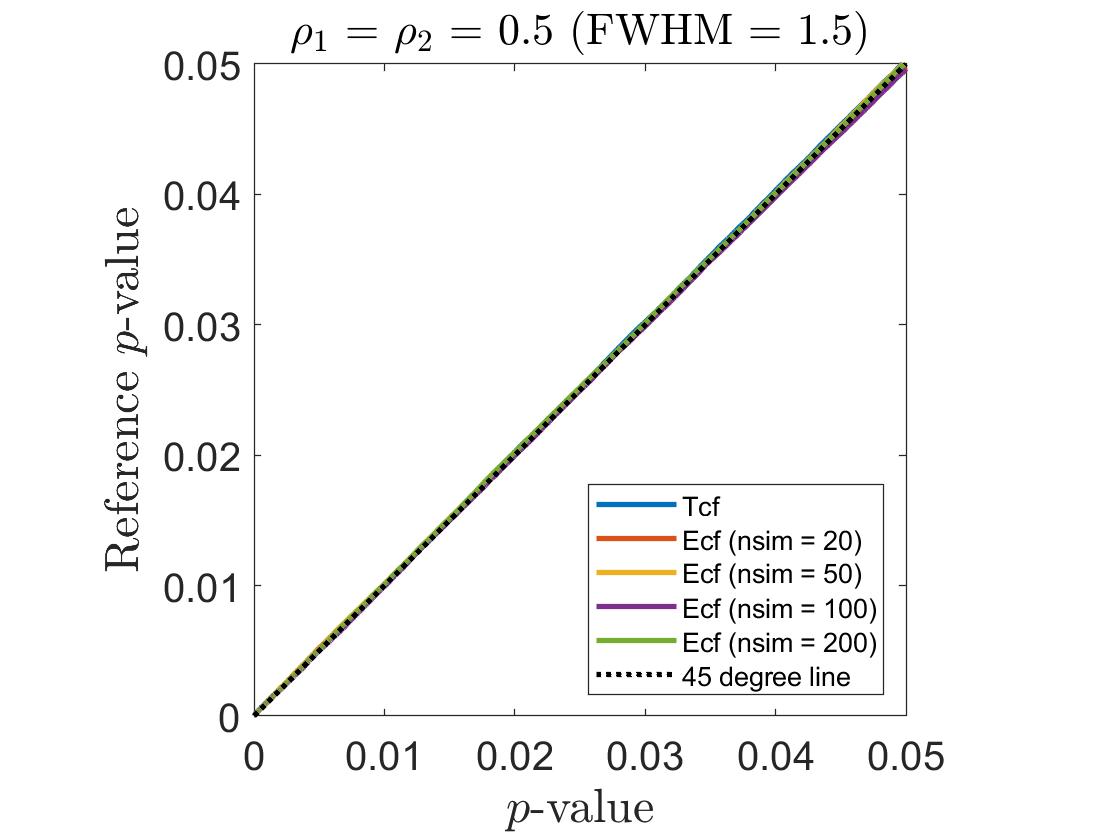}
\includegraphics[trim=80 5 80 5, clip,width=0.3\textwidth]{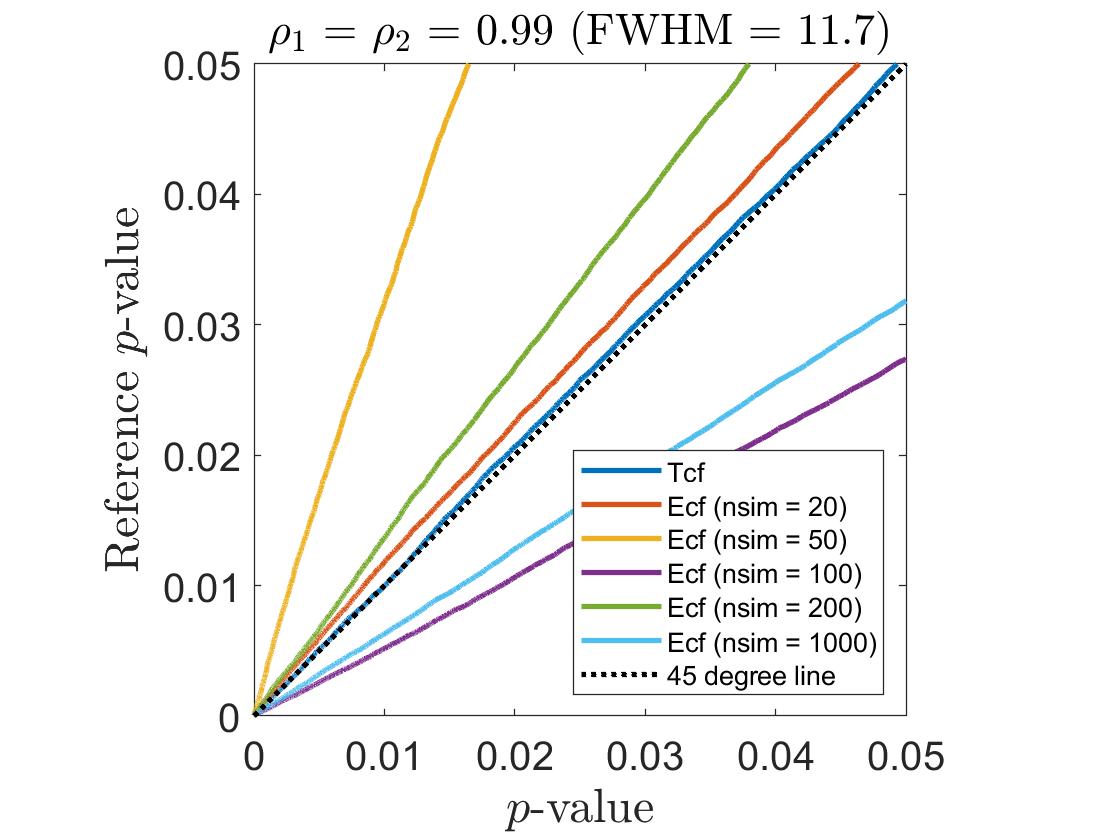}

\begin{sideways}
\phantom{------------------}3D
\end{sideways}
\includegraphics[trim=80 5 80 5, clip,width=0.3\textwidth]{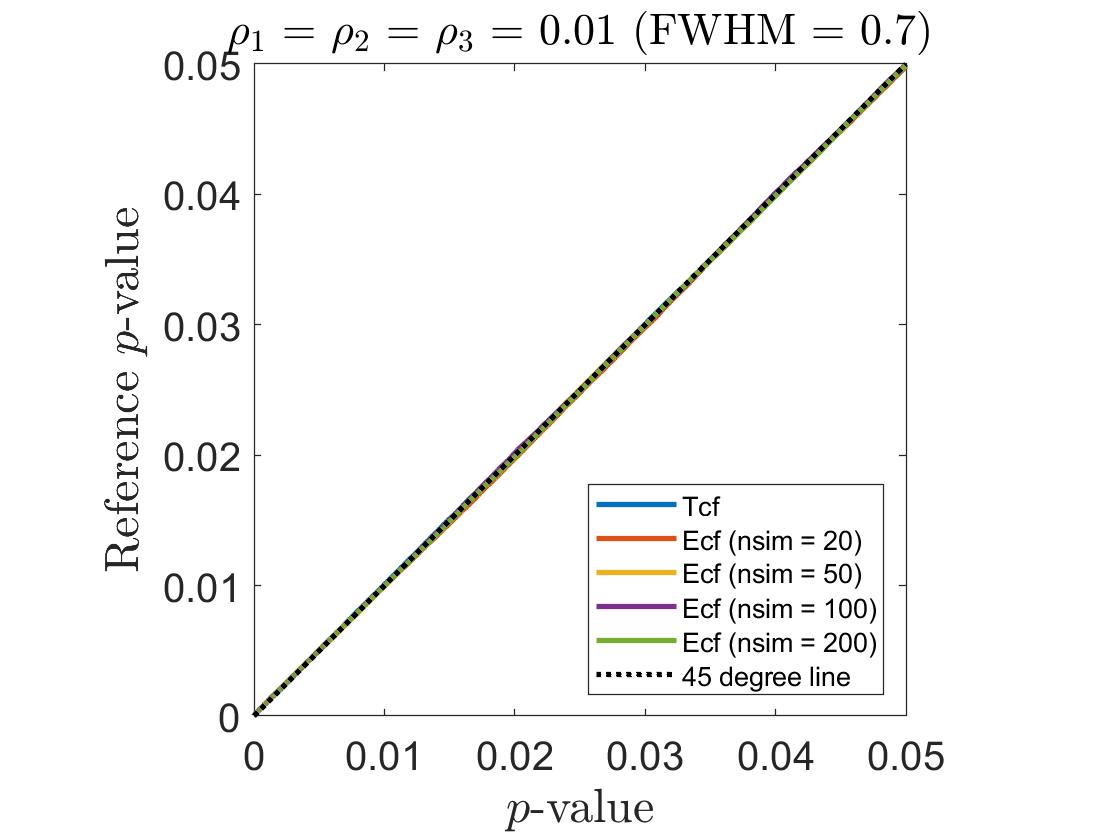}
\includegraphics[trim=80 5 80 5, clip,width=0.3\textwidth]{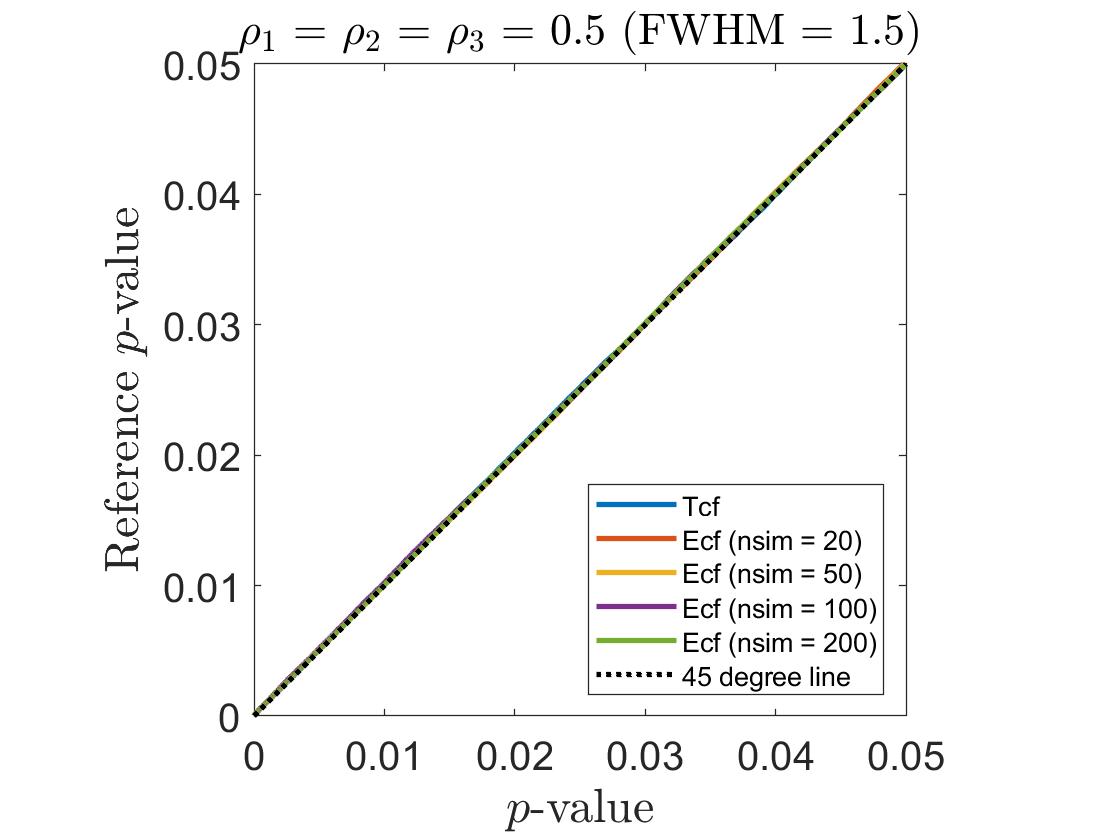}
\includegraphics[trim=80 5 80 5, clip,width=0.3\textwidth]{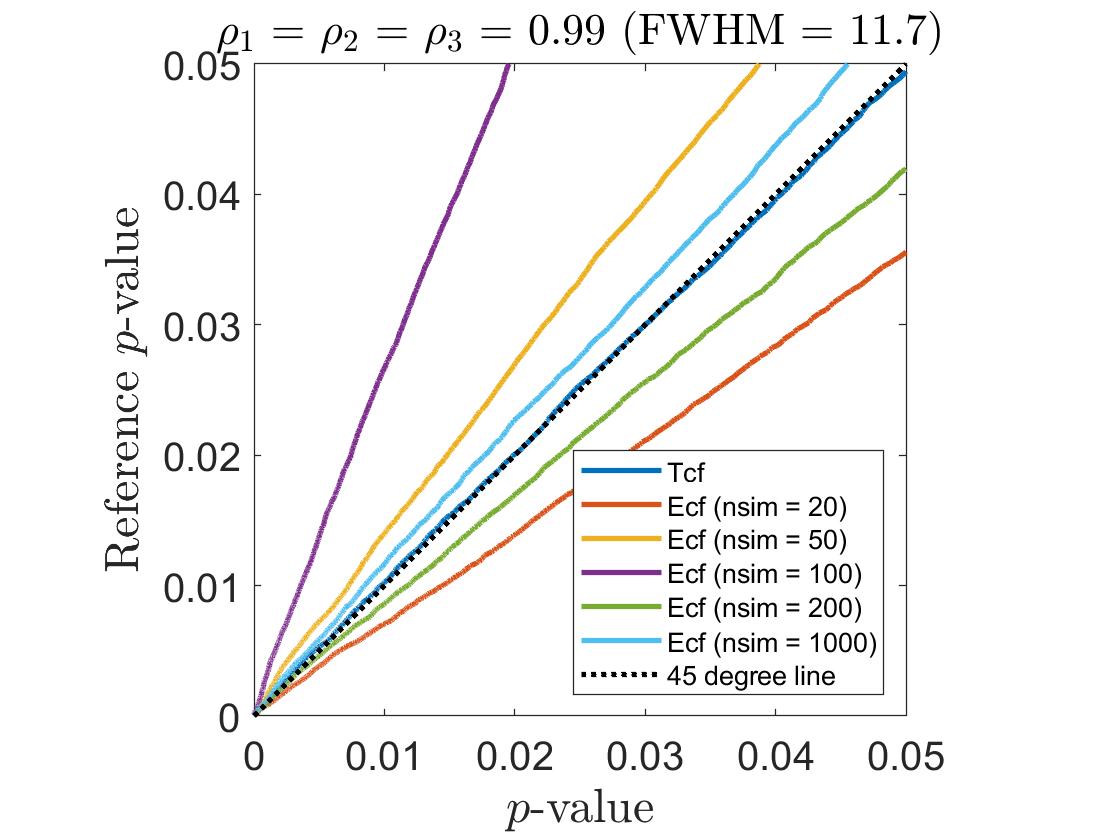}
\caption{Comparison of the peak height distribution calculated from using MCDLM with different neighborhood covariance for 2D and 3D isotropic Gaussian fields. The covariance functions used here are \ntb{true} covariance function (Tcf) and empirically estimated covariance function (Ecf). The number of random fields used to estimate the covariance function is denoted "{\sffamily nsim}". From left to right: $\rho = 0.01, 0.5, 0.99$. \label{fig13}}
\end{figure}

\begin{figure}[!htp]
\centering
\includegraphics[trim=100 5 100 5, clip,width=0.35\textwidth]{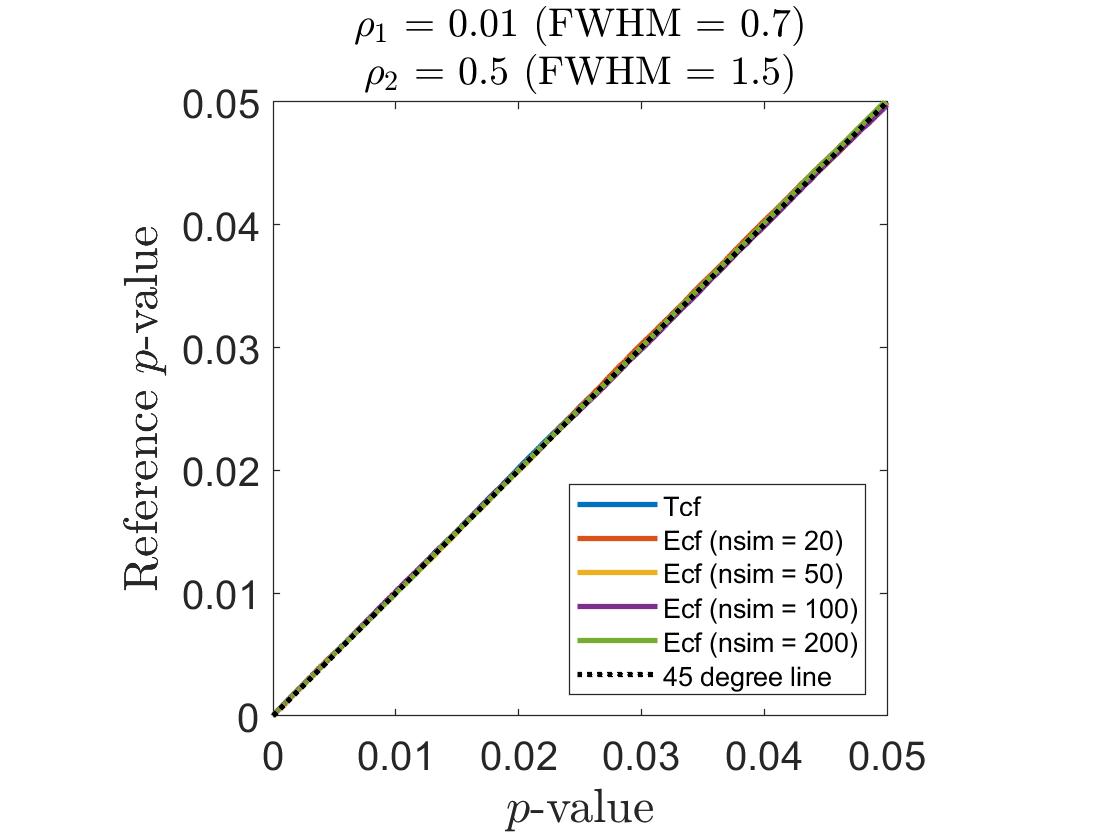}
\includegraphics[trim=100 5 100 5, clip,width=0.35\textwidth]{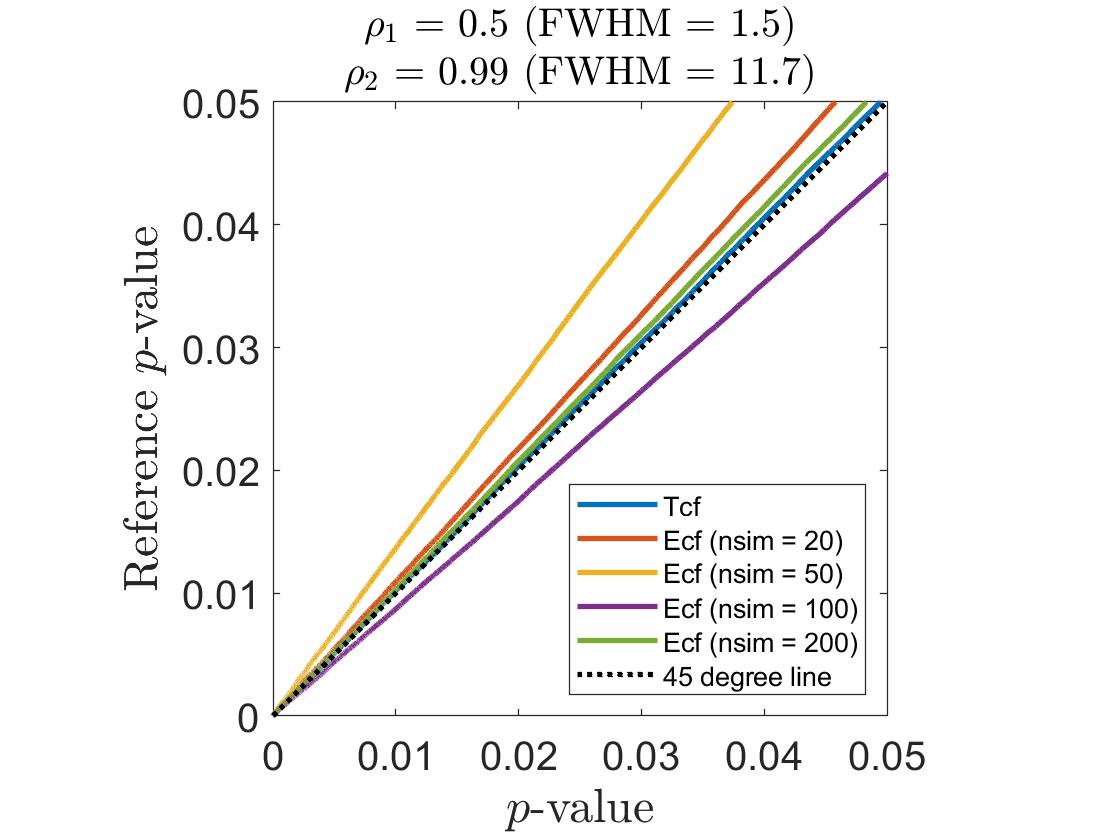}
\caption{Comparison of the peak height distribution calculated from using MCDLM method with different covariance functions for 2D anisotropic stationary Gaussian fields. The left: $\rho_1 = 0.01$ and $\rho_2 = 0.5$ and the right: $\rho_1 = 0.5$ and $\rho_2 = 0.99$. \label{fig15}}
\end{figure}

\section{HCP Task fMRI Results}
\label{sec5}
We evaluate the methods on 80 unrelated subjects of the Human Connectome Project (HCP) by testing one of the working memory contrasts. Specifically, we focus on the average contrasts of 2-back and 0-back tasks \citep{davenport2020selective}. Since N-back task is commonly used to measure working memory, and N can be adjusted for task difficulty, such contrasts help to identify the brain regions supporting working memory \citep{kirchner1958age}.  \nt{To ensure that the assumption of stationarity is reasonable we restrict to a mask of the gray matter on which the smoothness can be assumed to be constant (personal communication with Thomas E. Nichols), see e.g. Supplementary Figure 19 of \cite{eklund2016cluster}}. 

We perform the peak inference by computing the one-sample $t$-statistics at voxel level and applying MCDLM for $t$-fields to calculate $p$-values for 1101 peaks that we found in image. We use both MCDLM based on multivariate $t$-distribution and Gaussianization transformation method described in Section \ref{sec2.4} for the calculation. The isotropy assumption is too strict as we observe three different correlations, 0.86 (FWHM = 3.06), 0.88 (FWHM = 3.27) and 0.85 (FWHM = 2.87) in the horizontal, vertical and longitudinal directions, respectively. \nt{Based on our simulations we would expect continuous RFT to be conservative in this setting due to the low smoothness of this data and since it does not work well for $t$-statistics. Instead this setting is ideal for the use of MCDLM with the estimated neighborhood covariance matrix. We compare the results of both approaches applied to this dataset.}

The largest peak has a $t$-statistic value of 13.59. \nt{Figure \ref{fig.hcp} shows slices in different directions through one-sample $t$-statistic corresponding to this location.} This peak is located within the medial prefrontal cortex, which is an area commonly associated with working memory \citep{euston2012role,perlstein2002dissociation}. \nt{Using the one-sample $t$ version of MCDLM defined in Section \ref{sec2.4} we calculate a $p$-value of \num{4.57e-7}. Instead using the faster but slightly less accurate Gaussianized version of MCDLM we obtain a $p$-value of \num{9.42e-8}. } It is worth noting that, as MCDLM is a computation-based method, there is no closed-form function to calculate $p$-values across the entire real domain. When the $t$-statistic is extremely large, determining the exact $p$-value requires a large number of samples, making the process computationally intensive. In such case, we conclude that the $p$-value is smaller than the minimum value obtained from the empirical distribution in our Monte Carlo experiments. 


\begin{figure}[!htp]
	\centering
	\includegraphics[width = \textwidth]{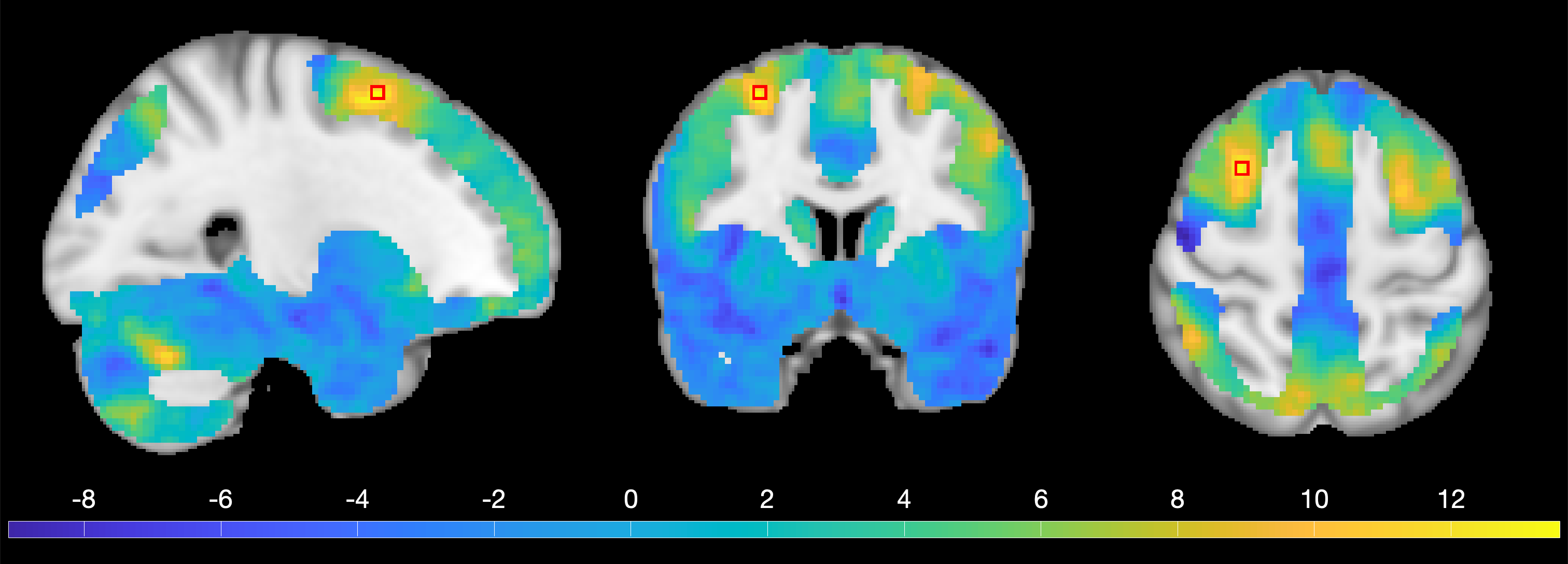}
	\caption{\nt{Slices in different directions through the one-sample $t$-statistic - masked to the gray matter - at the location of the largest peak from the 2-back minus 0-back contrast from the Human Connectome Project. The colorbar indicates the magnitude of the test-statistic. A red box indicates the location of the largest peak.}\label{fig.hcp}}
\end{figure}

The number of significant peaks ($p$-value < 0.05) out of 1101 peaks from MCDLM based on multivariate $t$-distribution is 359, the MCDLM based on Gaussianization transformation method is 341, and the continuous RFT method is 325. We also implement the Benjamini-Hochberg procedure \citep{benjamini1995controlling} to control the false discovery rate (FDR) to a level $\alpha = 0.05$. After multiple comparisons adjustment, there are 293 significant peaks from MCDLM, 276 from MCDLM with Gaussianization, 270 from continuous RFT method. These results demonstrate that continuous RFT is overly conservative in this setting as expected from our simulations. In Table \ref{tab.hcp} we illustrate the results for 10 peaks which have $p$-values close to 0.05. MCDLM and continuous RFT method give different conclusions for all these 10 peaks, i.e., while all these 10 peaks are found to be signfiicant using MCDLM, they are not significant when continuous RFT is used, illustrating the increased sensitivity of our approach. 

\begin{table}[!htp]
\caption{The $t$-statistics and $p$-values from MCDLM based on multivariate $t$-distribution, MCDLM with Gaussianization and Continuous RFT method for 10 selected peaks that have $p$-values close to 0.05. The location indices are from the natural indexing within the 1 to 91, 1 to 109 and 1 to 91 image. \label{tab.hcp}}
\centering
\begin{tabular}{lllll}
\hline
$t$-statistics & Location & MCDLM & MCDLM w/ Gaussianization & Continuous RFT\\ \hline
3.151 & (49,70,49) & 0.033 & 0.045 & 0.060 \\ 
3.146 & (51,62,37) & 0.033 & 0.045 & 0.061 \\ 
3.144 & (56,48,22) & 0.033 & 0.045 & 0.061 \\ 
3.130 & (55,24,15) & 0.034 & 0.047 & 0.063 \\ 
3.124 & (41,37,22) & 0.035 & 0.047 & 0.063 \\ 
3.108 & (44,57,24) & 0.036 & 0.049 & 0.066 \\ 
3.104 & (76,40,54) & 0.036 & 0.049 & 0.066 \\ 
3.088 & (48,21,46) & 0.038 & 0.051 & 0.069 \\ 
3.081 & (47,57,52) & 0.038 & 0.052 & 0.070 \\ 
3.069 & (56,46,22) & 0.039 & 0.053 & 0.071 \\ \hline
\end{tabular}
\end{table}

\section{Discussion}
\label{sec6}
In this paper, we have proposed a new Monte Carlo method to calculate the distribution of the height of a peak of a discrete Gaussian random field which works under minimal assumptions. When inferring on the heights of the peaks of Gaussian fields, MCDLM performed well compared to other approaches. Historically, continuous RFT has been used to calculate the distribution of the height of local maxima in a continuous random field. However, in practice we observe data on a lattice. As shown in \cite{schwartzman2019peak} and the simulation studies in this paper, when the data is sufficiently smooth (FWHM $\ge 7$ in \cite{schwartzman2019peak} and FWHM $\geq 6.7$ in this paper), the continuous formulae provide a good approximation to the height of local maxima. \nt{Nevertheless, in many realistic situations (FWHM $< 6.7$) the data is not smooth enough and using the continuous formulae can lead to conservative inference. Furthermore, the continuous formulae only work for an isotropic field or a field that can be deformed to an isotropic field, which is a highly restrictive assumption and is often not reasonable in practice. When applied to an isotropic $t$-field, continuous RFT also relies heavily on convergence to Gaussianity, which is only valid when the degrees of freedom is large. One further point is that the peak height distribution will be different for points on the boundary of the domain which is difficult to account for using continuous RFT but easy to do so using MCDLM by changing the neighborhood structure. }

\nt{In our simulations we showed that continuous RFT was conservative while ADLM was liberal in a wide variety of settings. The conservativeness of continuous RFT is due to a failure of the good lattice assumption and is the same phenomenon observed in \cite{telschow2023riding} and \cite{davenport2023robust}. Instead the anti-conservativeness of ADLM is due to the failure of its rather restrictive assumptions on the covariance and neighbourhood structure. MCDLM instead provided exact inference in most of the settings considered. We thus recommend using MCDLM to infer on peak height at low to medium levels of smoothness. However, when the data is very smooth and it is reasonable to assume the data is isotropic, there may not be much gain relative to continuous RFT. Since the latter is very efficient we recommend it in that setting as it provides a precise formula for the peak height distribution which can be quickly and accurately calculated. At high levels of smoothness the covariance matrix is nearly singular making it difficult to estimate, which causes problems when applying MCDLM.} A detailed running-time table for all three approaches, under different scenarios, is provided in Appendix \ref{appendix.g}.


The proposed MCDLM method also works for $t$-fields, but it takes a long time to implement this approach when the number of degrees of freedom is large. To improve the computational efficiency, we recommend using a Gaussianization transformation and then applying MCDLM to the Gaussianized field. In this setting continuous RFT works well when both smoothness and degrees of freedom are high, but even when degrees of freedom increases to $200$, it is still outperformed by MCDLM. Moreover MCDLM can be easily extended to obtain the peak height distribution of two-sample $t$-statistic and $F$-fields observed on a lattice \nt{by changing line 3 of Algorithm \ref{alg1} appropriately.} The proposed method is limited to stationary Gaussian or Gaussian-derived random fields. However extensions to locally stationary and non-stationary fields are possible and are an interesting avenue for future research.

\nt{\section{Acknowledgements}
SD and AS were partially supported by NIH grant R01EB026859.  Data were provided in part by the Human Connectome Project, WU-Minn Consortium(Principal Investigators: David Van Essen and Kamil Ugurbil;1U54MH091657) funded by the 16 NIH Institutes and Centers that support the NIHBlueprint for Neuroscience Research; and by the McDonnell Center for Systems Neuroscienceat Washington University.}

\newpage
\bibliographystyle{apa}
\bibliography{references}

\newpage
\appendix
\label{sec:headings}

\section{Theory for ADLM}
\label{appendix.a}
\subsection{Proof of Proposition \ref{prop1}\label{appendix.a1}}
\begin{proof}
Without loss of generality, we assume that $Z$ is mean zero and unit variance. For $s \in \mathcal{L}$, taking $\mathcal{N}(s)$ to be the partially connected neighborhood of $s$, for $s_a, s_b \in \mathcal{N}(s)$, we have
\begin{align*}
    \begin{pmatrix} Z(s) \\ Z(s_a) \\ Z(s_b) \end{pmatrix} \sim N\left(\mathbf{0}, \begin{pmatrix} \Sigma_{AA} & \Sigma_{AB} \\ \Sigma_{BA} & \Sigma_{BB} \end{pmatrix}\right),
\end{align*}
where $\Sigma_{AA} = \Var(Z(s)) = 1$, $\Sigma_{AB} = \Sigma_{BA}^\top = \left[\Cov(Z(s), Z(s_a))\  \Cov(Z(s), Z(s_b))\right] = \left[\rho(s,s_a)\ \rho(s, s_b)\right]$ and $\Sigma_{BB} = \Var\left[\begin{pmatrix} Z(s_a) \\ Z(s_b) \end{pmatrix}\right] = \begin{pmatrix}
    1 & \rho(s_a, s_b) \\ \rho(s_a, s_b) & 1
    \end{pmatrix} $.

For $z \in \mathbb{R}$, the covariance of $Z(s_a)$ and $Z(s_b)$ conditional on $Z(s)=z$ is 
\begin{align*}
    \Var\left[\begin{pmatrix} Z(s_a) \\ Z(s_b) \end{pmatrix}\Bigg|Z(s) = z\right] &= \Sigma_{BB} - \Sigma_{BA}\Sigma_{AA}^{-1}\Sigma_{AB}\\ &= \Sigma_{BB} - \Sigma_{BA}\Sigma_{AB} \\ &= \begin{pmatrix}
    1 & \rho(s_a, s_b) \\ \rho(s_a, s_b) & 1
    \end{pmatrix} - \begin{pmatrix}
    \rho(s, s_a)^2 & \rho(s, s_a)\rho(s, s_b) \\ \rho(s, s_a)\rho(s, s_b) & \rho(s, s_b)^2
    \end{pmatrix}.
\end{align*}

Due to the form of the covariance function, the off-diagonal entries are 
\begin{align*}
    \rho(s_a, s_b) - \rho(s, s_a)\rho(s, s_b) = \rho^{||s_a-s_b||^2} - \rho^{||s_a-s||^2}\rho^{||s_b-s||^2} = \rho^{||s_a-s_b||^2} - \rho^{||s_a-s||^2 + ||s_b-s||^2},
\end{align*}
where $\rho$ is the correlation between two adjacent voxels and $||\cdot||$ denotes the Euclidean norm. Thus, if $||s_a-s_b||^2 = ||s_a-s||^2 + ||s_b-s||^2$, then $\Cov(Z(s_a), Z(s_b)|Z(s)) = 0$. Now $(Z(s_a), Z(s_b))$ follows a bivariate Gaussian distribution so it follows that $Z(s_a)$ and $Z(s_b)$ is independent conditional on $Z(s_5)$. Taking $s_a = s\pm e_{d_1}$ and $s_b = s\pm e_{d_2}$, where $d_1$ and $d_2$ denote different lattice directions, $||s_a-s_b||^2 = ||e_{d_1} \pm e_{d_2}||^2$ which equals $||s_a-s||^2 + ||s_b-s||^2 = ||e_{d_1}||^2 + ||e_{d_2}||^2$ and so the result follows. 
\end{proof}

\subsection{Theoretical derivation of the probability density function of ADLM method\label{appendix.a2}}
\nt{In the result below we derive a closed form for peak height distribution for the ADLM approach, recalling that $\mathcal{N}_{PC}$ is the partially connected neighbourhood and defining $Q(\rho_d,z)$ as in \eqref{sdfsdf}.}
\begin{Proposition}\label{prop2}
\nt{Under the conditions of Proposition \ref{prop1}, given $s \in \mathcal{L}$ such that $\mathcal{N}_{PC}(s) \subseteq \mathcal{L}$, for all $u \in \mathbb{R}$ we have}
\begin{align*}
    P [Z(s) > u|Z(t) < Z(s), \forall t \in \mathcal{N}_{PC}(s)] = \frac{\int_u^\infty\left(\prod_{d=1}^DQ(\rho_d,z)(z)\right)\phi(z)dz}{\int_{-\infty}^\infty\left(\prod_{d=1}^DQ(\rho_d,z)(z)\right)\phi(z)dz}
\end{align*}
\nt{and in particular
\begin{align*}
	f_{\text{DLM}}(z) = \frac{\prod_{d=1}^DQ(\rho_d,z)\phi(z)}{\int_{-\infty}^\infty\bigg(\prod_{d=1}^DQ(\rho_d,z)\bigg)\phi(z)dz}.
\end{align*}}
\end{Proposition}

\begin{proof}
First, by the law of iterated expectations,
\begin{align*}
P\left[\{Z(s) > u\}\cap_{t\in\mathcal{N}_{PC}(s)}\{Z(t) < Z(s)\}\right]
&=\int_{-\infty}^{\infty}P\left[\{z>u\}\cap\{Z(t)<z, \forall t\in\mathcal{N}_{PC}(s)|Z(s)=z\}\right]\phi(z)dz\\
&=\int_{-\infty}^{\infty}\mathbbm{1}_{\{z>u\}}\cdot P\left[\{Z(t)<z, \forall t\in\mathcal{N}_{PC}(s)|Z(s)=z\}\right]\phi(z)dz\\
&=\int_{u}^{\infty}P\left[\{Z(t)<z, \forall t\in\mathcal{N}_{PC}(s)|Z(s)=z\}\right]\phi(z)dz.
\end{align*}

Next, by applying Proposition \ref{prop1}, 
\begin{align*}
P\left[\{Z(t)<z, \forall t\in\mathcal{N}_{PC}(s)|Z(s)=z\}\right] &= P\bigg[\bigcap_{d=1}^{D}(Z(t)<z,\ \forall t=s\pm v_de_d|Z(s)=z) \bigg]\nonumber = \prod_{d=1}^{D}Q(\rho_d,z),
\end{align*}
where
\begin{align*}
Q(\rho_d,z)=P\{Z(t)<z,\ \forall t=s\pm v_de_d |Z(s)=z\}.
\end{align*}

Thus we have
\begin{align*}
P\left[\{Z(s) > u\}\cap_{t\in\mathcal{N}_{PC}(s)}\{Z(t) < Z(s)\}\right]= \int_u^\infty\bigg(\prod_{d=1}^DQ(\rho_d,z)\bigg)\phi(z)dz.
\end{align*}

By Bayes Rule,
\begin{align*}
P [Z(s) > u|Z(t) < Z(s), \forall t \in \mathcal{N}_{PC}(s)] &= \frac{P[\{Z(s)> u\} \cap \{Z(t) < Z(s), \forall t \in \mathcal{N}_{PC}(s)\}]}{P[Z(t) < Z(s), \forall t \in \mathcal{N}_{PC}(s)]} \\&= \frac{\int_{u}^{\infty}P\big[(Z(t)<z, \forall t \in\mathcal{N}_{PC}(s)|Z(x)=z)\big]\phi(z)dz}{\int_{-\infty}^{\infty}P\big[(Z(t)<z, \forall t\in\mathcal{N}_{PC}(s)|Z(x)=z)\big]\phi(z)dz} \\&=\frac{\int_u^\infty\left(\prod_{d=1}^DQ(\rho_d,z)\right)\phi(z)dz}{\int_{-\infty}^\infty\left(\prod_{d=1}^DQ(\rho_d,z)\right)\phi(z)dz}.
\end{align*}
\nt{The result for the peak height density follows accordingly.}
\end{proof}

\section{The neighbourhood covariance matrix in the fully connected setting\label{appendix.dd}}
\subsection{Theoretical derivation of the neighbourhood covariance for the integral convolution field}\label{appendix.B.1}
\nt{Before proving the result we need to first establish some notation. Given $r \in \mathcal{L}$, we first define an indexing of $\mathcal{N}_{FC}(r)$. We index the elements of the neighbourhood with vectors taking values in $\lbrace0, 1, 2\rbrace^D$. Under this indexing $r$ is assigned to $1_D$ (the $D$-dimensional vector of 1s) and $r + \sum_{i = 1}^D a_i e_i$ (for $a_i \in \lbrace-1,0,1\rbrace$) is assigned to the vector $1_D + \sum_{i = 1}^D a_i e_i$. To help understand this indexing in practice, a 2D example of this is as follows.
\begin{align*}
	\begin{matrix}
		(0,0) & (0,1) & (0,2)\\
		(1,0) & (1,1) & (1,2)\\
		(2,0) & (2,1) & (2,2).
	\end{matrix}
\end{align*}
We are now ready to state and prove our result.}
\begin{theorem}\label{thm:thm}
	\nt{Let $Z$ be a $D$-dimensional isotropic Gaussian random field on a discrete lattice $\mathcal{L}$ with mean zero and unit variance. Assume that $Z$ is derived by integrating continuous white noise against a Gaussian kernel. Given $r \in \mathcal{L}$, suppose that $s, t$ are two elements of $\mathcal{N}_{FC}(r)$. And write these elements as $s = (s_1, \dots, s_D)$ and $t = (t_1, \dots, t_D)$ under the indexing system defined above.} Then
	\begin{align}\label{eqn234}
		\mathrm{Cov}(Z(s),Z(t)) = \left[A\otimes A \otimes... \otimes A\right]_{m,n},
	\end{align}
	where
	\begin{align*}
		A = \begin{pmatrix}
			1 & \rho & \rho^4\\
			\rho & 1 & \rho\\
			\rho^4 & \rho & 1
		\end{pmatrix},
	\end{align*}
	$\rho$ is the correlation between adjacent voxels, and $m = \sum_{i=1}^D 3^{i-1}s_{i}+1$ and $n = \sum_{i=1}^D 3^{i-1}t_{i}+1$.
\end{theorem}
\begin{proof} Because of the form of $Z$, $\Cov(Z(s), Z(t)) = \rho^{||s-t||^2}$. We have that $\left[A\otimes A \otimes... \otimes A\right]_{m,n} = \prod_{i=1}^D A_{s_it_i}$. If $||s_i-t_i|| = 0$, $A_{s_it_i} = 1 = \rho^0$; if $||s_i-t_i|| = 1$, $A_{s_it_i} = \rho$ and if $||s_i-t_i|| = 2$, $A_{s_it_i} = \rho^4$. Thus, 
\begin{align*}
    A_{s_it_i} = \rho^{||s_i-t_i||^2} ,
\end{align*}
and hence
\begin{align*}
\prod_{i=1}^D A_{s_it_i} = \rho^{\sum_{i=1}^D||s_i-t_i||^2} = \rho^{||s-t||^2}.
\end{align*}
As such both sides of (\ref{eqn234}) match and the result follows.
\end{proof}
This proof relies on a special indexing of $\mathbf{Z}$. However, this is in practice not a restriction as a re-indexing of $\mathbf{Z}$ can be treated as a linear transformation, i.e. $\mathbf{Z} = T \mathbf{Z}_o$, where $\mathbf{Z}_o$ is the vector with the special indexing. In that case $\Sigma$ then can be calculated as
\begin{align*}
	\Sigma = T \mathrm{Cov}(\mathbf{Z}_o) T^\top.
\end{align*}

\subsection{Neighbourhood covariance matrix for a 3D stationary Gaussian random field}\label{SS:cov}
Figure \ref{fig17} shows the theoretical covariance function for a 3D stationary field. The construction and indexing of this $27 \times 27$ matrix follow the logic from Section \ref{appendix.B.1}. The numbers $r0,r1,...,r61$ denote all $62$ distinct values of covariance between two voxels. We use the same color to denote all $3\times 3$ matrices with the same values. One useful conclusion from this figure is this covariance matrix is a block Toeplitz matrix with $9$ blocks, and each block is still a block Toeplitz matrix with $9$ sub-blocks. In addition, each sub-block is a Toeplitz matrix. 

\begin{figure}[!htp]
\centering
\includegraphics[width = 8cm]{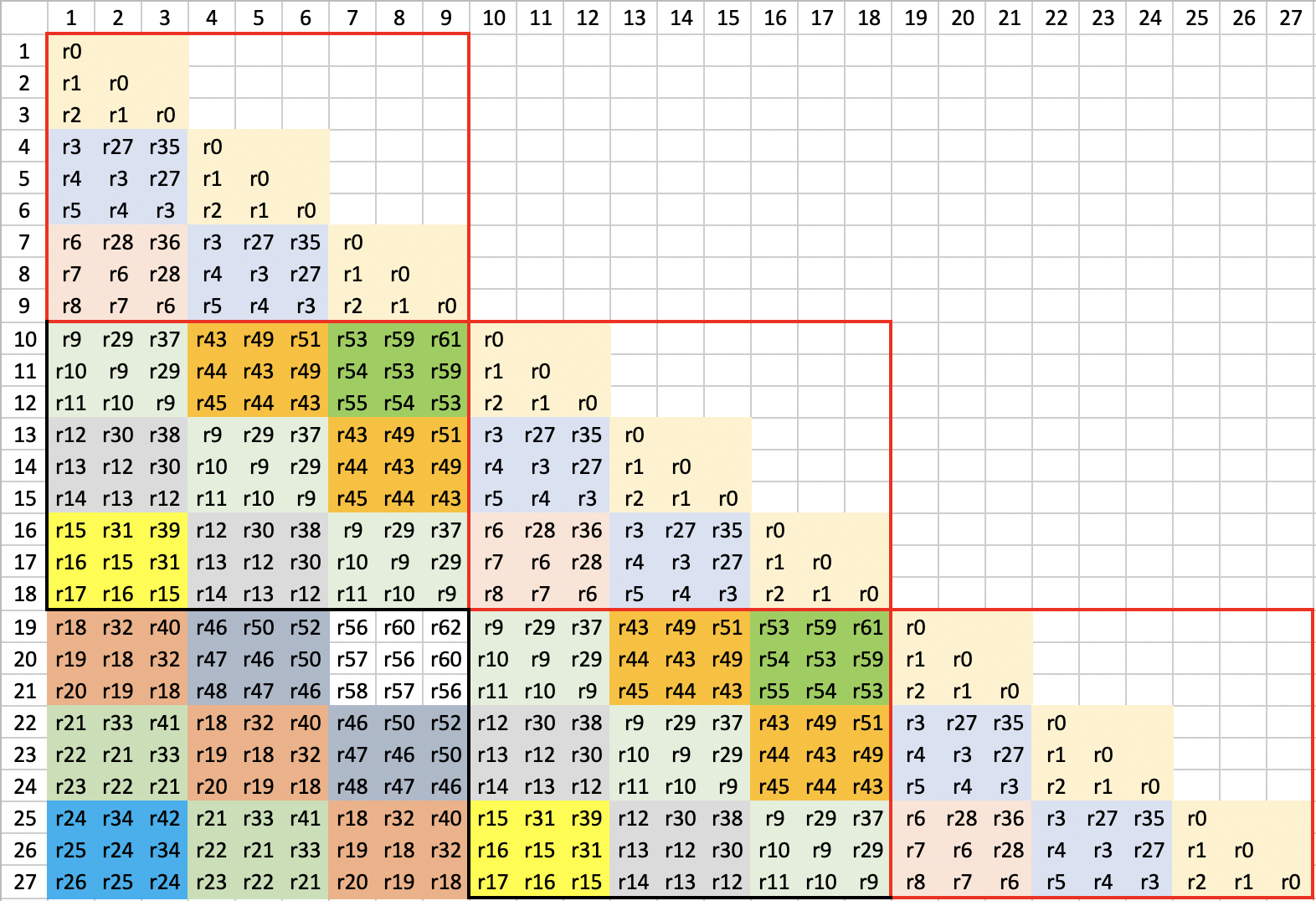}
\caption{The theoretical \nt{neighbourhood covariance matrix} for a 3D stationary field.\label{fig17}}
\end{figure}

\section{Formally defining the $pp$ plot\label{appendix.d}}

The $pp$ plots are used throughout Section \ref{sec3} and \ref{sec4}. \nt{In this section, we will define these plots formally. The first step is to define the reference $p$-value, which was obtain via simulation. To do so, we generate $N$ i.i.d. random fields as described in each setting in Section \ref{sec3}.} Let $n$ be the number of obtained local maxima across all fields, where the local maxima are selected based on the criteria that their height values are larger than all their neighbours in the specified neighbourhood (we will consider both the fully connected and partially connected neighbourhoods).  Let $g_1, \dots, g_n$ be the heights of the recorded local maxima. For each peak this allows us to compute a reference $p$-value as
\begin{align}
    p_i = \frac{1}{n} \sum_{j = 1}^n \mathbbm{1} [ g_j > g_i],\quad 1 \leq i \leq n,  \label{eqn3.2}
\end{align}
where $\mathbbm{1}[\cdot]$ denotes the indicator function, $p_i$ is the $p$-value when observed value is $g_i$. As $n\rightarrow \infty$, $p_i$ calculated by (\ref{eqn3.2}) converges to the true tail probability. Moreover for each peak, we calculate a $p$-value for each of the three approaches. Next we plot the reference $p$-values against the $p$-values obtained using each method. Since the reference distribution converges to the true peak height distribution as the number of instances converges to infinity, the closer these plots are to the identity function, the closer the approximation to the true distribution. We use these $pp$ plots to compare the performance of the three approaches in all of our simulation studies. 

The idea of this $pp$ plot is similar to the one used in \cite{schwartzman2019peak}. Although the two plotting mechanisms look different, the logic behind is the same, as we justify below. 

Our two plotting mechanisms are
\begin{itemize}
    \item Plot $p(i)$ vs. $i/n$ and $q(i)$ vs. $i/n$ (\cite{schwartzman2019peak}),
    \item Plot $p(i)$ vs. $p(i)$ and $q(i)$ vs. $p(i)$ ($pp$ plot).
\end{itemize}

Let $F(z)$ be true cdf of $z$, $G(z)$ be one of the other cdfs of $z$ used for comparison purpose. Next, we define
\begin{align*}
    p &= 1-F(z)\\
    q &= 1-G(z)
\end{align*}
Suppose that we now generate $n$ $p$-values for both the true distribution and the distribution for comparison purpose, i.e., we generate $p_1,...,p_n$ and $q_1,...,q_n$ as in (\ref{eqn3.2}). Denote $p(i)$ and $q(i)$ as the order statistics of $p_1,...,p_n$ and $q_1,...,q_n$. Then under monotonicity, 
\begin{align*}
    \text{ecdf}(p(i)) &= \frac{1}{n}\sum_{i = 1}^n I[p_i \leq p(i)] = \frac{i}{n}\\
    \text{ecdf}(q(i)) &= \frac{1}{n}\sum_{i = 1}^n I[q_i \leq q(i)] = \frac{i}{n}
\end{align*}

If $z$ is distributed according to $F$, $p$-values are uniformly distributed on $[0,1]$. Thus, by the LLN,
\begin{align*}
    \text{ecdf}(p) = \frac{1}{n}\sum_{i=1}^nI(p_i\leq p)\rightarrow_p P(p_i\leq p) = p
\end{align*}

Thus, when $p$ is from the true distribution and $n$ is large enough, we expect the two plotting mechanisms provide similar plots. The comparison of two method is shown below in Figure \ref{fig18}.

\begin{figure}[!htp]
\centering
\includegraphics[trim=80 5 80 5, clip,width=0.3\textwidth]{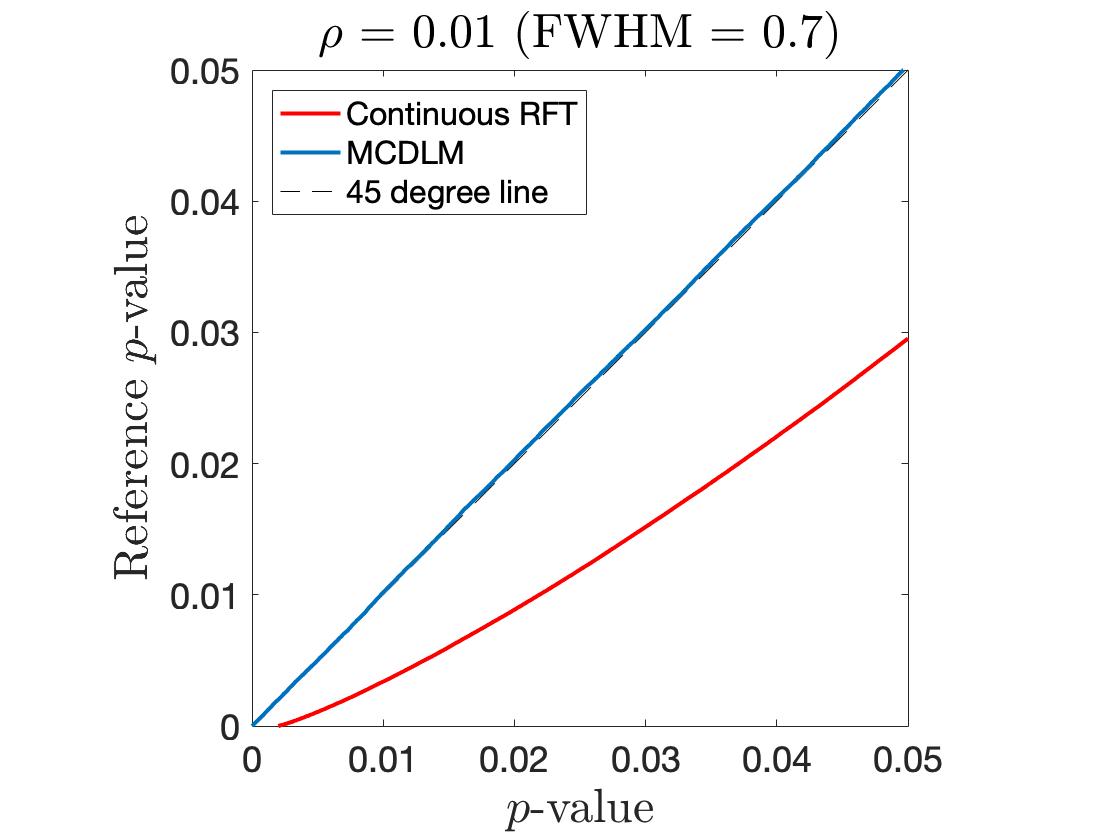}
\includegraphics[trim=80 5 80 5, clip,width=0.3\textwidth]{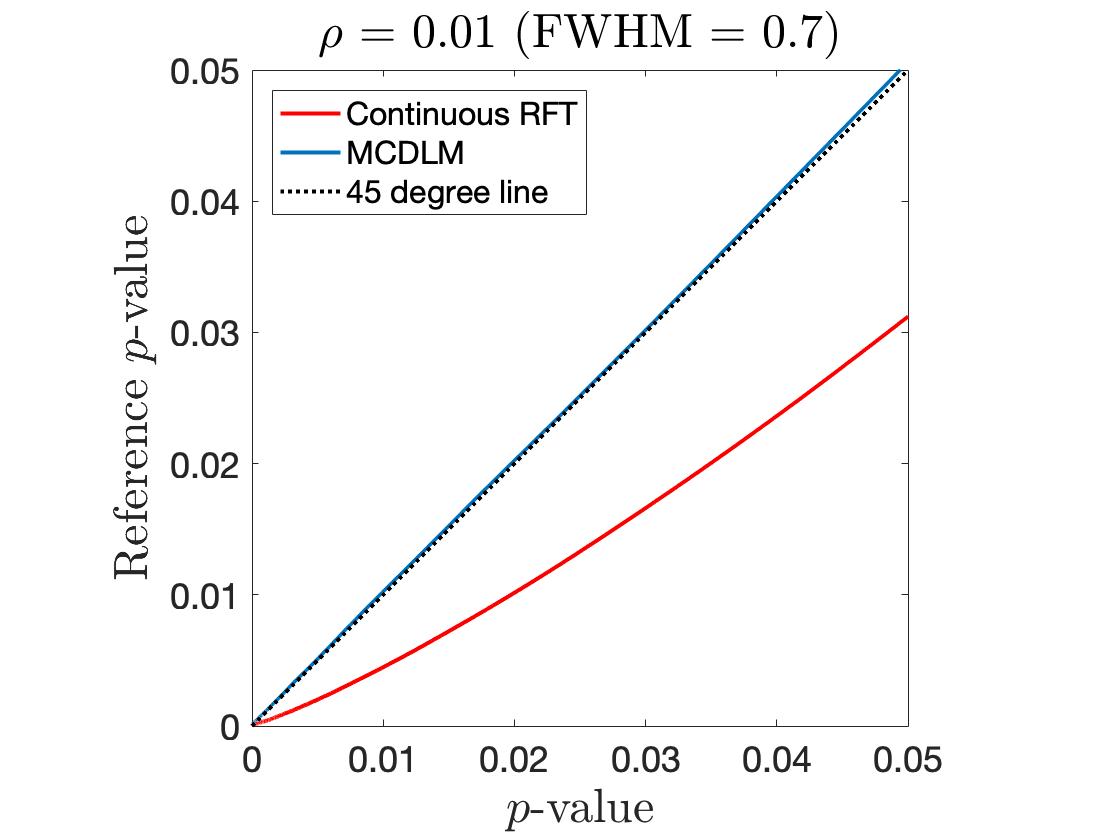}
\caption{Comparison of the two considered plotting mechanisms obtained using the height distribution generated from 2D isotropic Gaussian field with spatial correlation $\rho = 0.01$. Left is from the plotting mechanism in \cite{schwartzman2019peak} and right is from the plotting mechanism used in this paper. \nt{As can be seen from the plots they are essentially identical.}\label{fig18}}
\end{figure}

\section{Additional results}
\label{appendix.e}
\subsection{Partial connectivity case}
\label{appendix.d1}
\nt{In Section \ref{sec2} we discussed some of the disadvantages of ADLM. In particular Proposition \ref{prop1} does not apply to diagonal neighbours and so ADLM cannot be used for peaks of the fully connected neighbourhood which is why it performed badly in Figure \ref{fig7}. However (when its assumptions hold) it can be used to provide inference  for the partially connected neighbourhood defined in (\eqref{eqn2.22}). We illustrate this in Figure \ref{fig8} within the same simulation setting of Section \ref{sec3.1}.}

\begin{figure}[!htp]
\centering
\begin{sideways}
\phantom{------------------}2D
\end{sideways}
\includegraphics[trim=80 5 80 5, clip,width=0.3\textwidth]{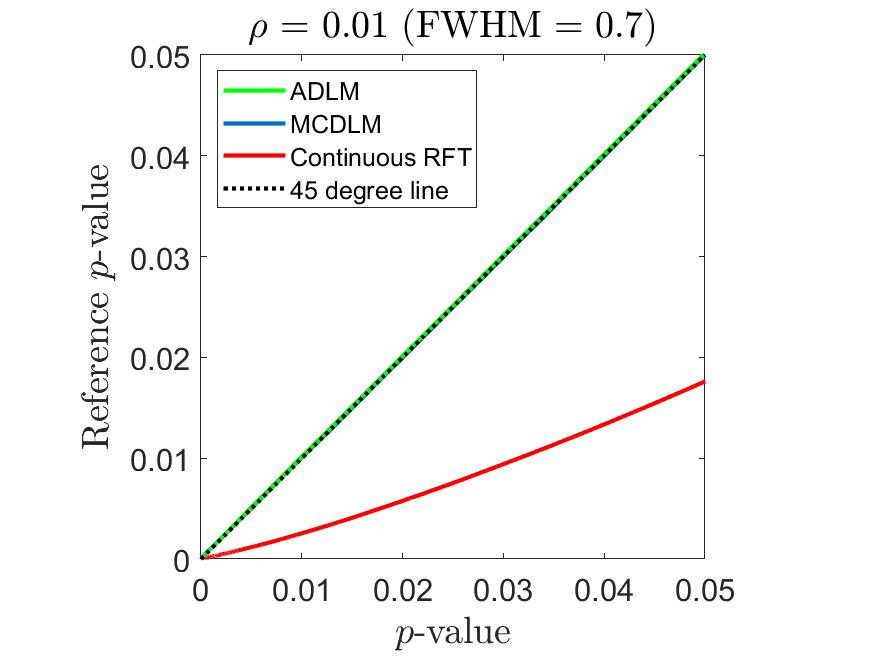}
\includegraphics[trim=80 5 80 5, clip,width=0.3\textwidth]{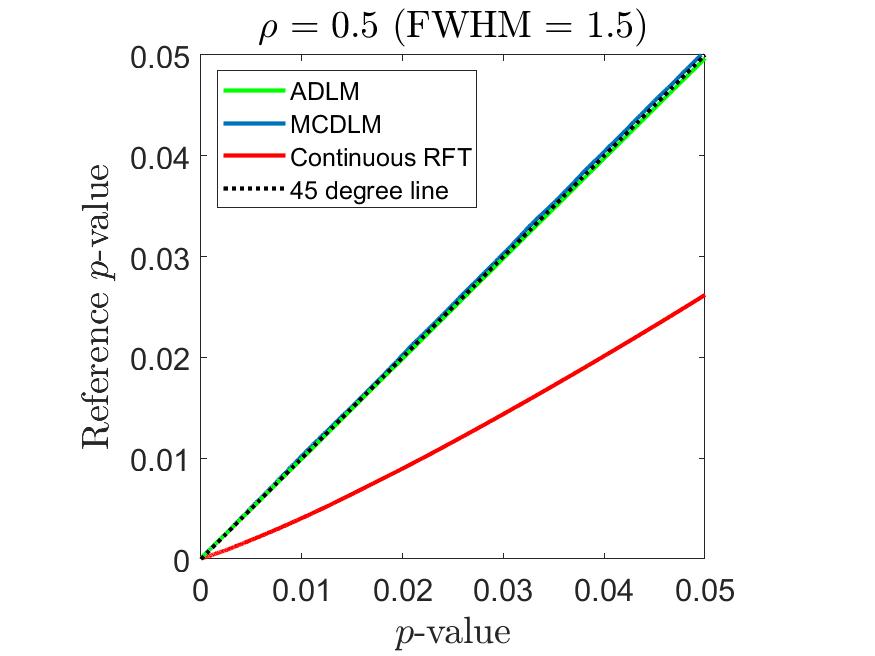}
\includegraphics[trim=80 5 80 5, clip,width=0.3\textwidth]{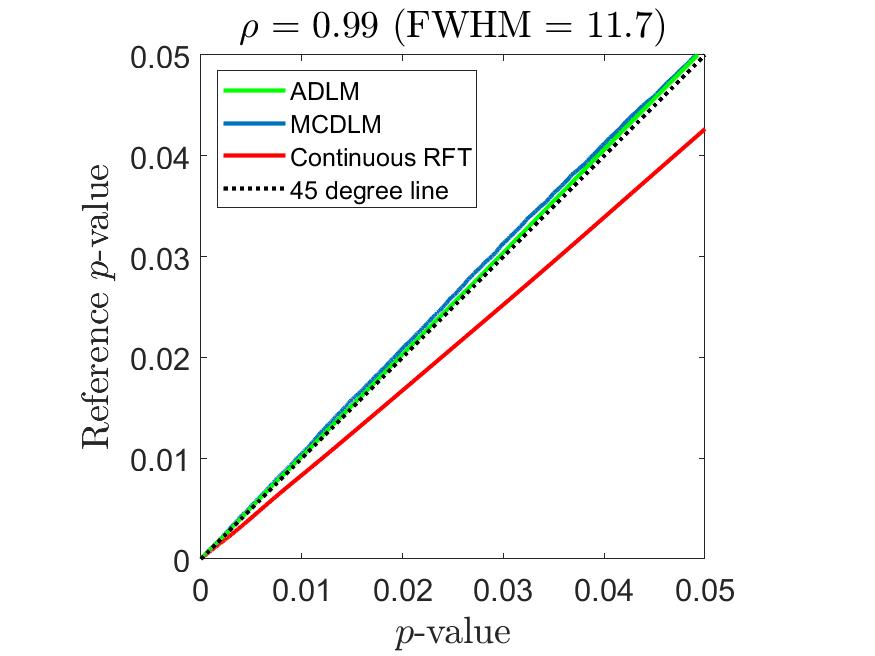}

\begin{sideways}
\phantom{------------------}3D
\end{sideways}
\includegraphics[trim=80 5 80 5, clip,width=0.3\textwidth]{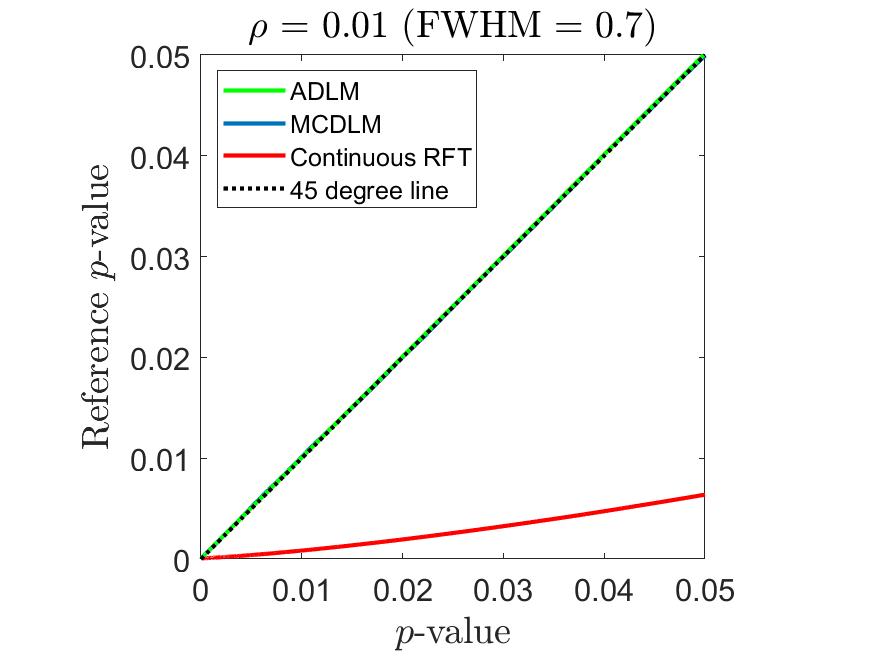}
\includegraphics[trim=80 5 80 5, clip,width=0.3\textwidth]{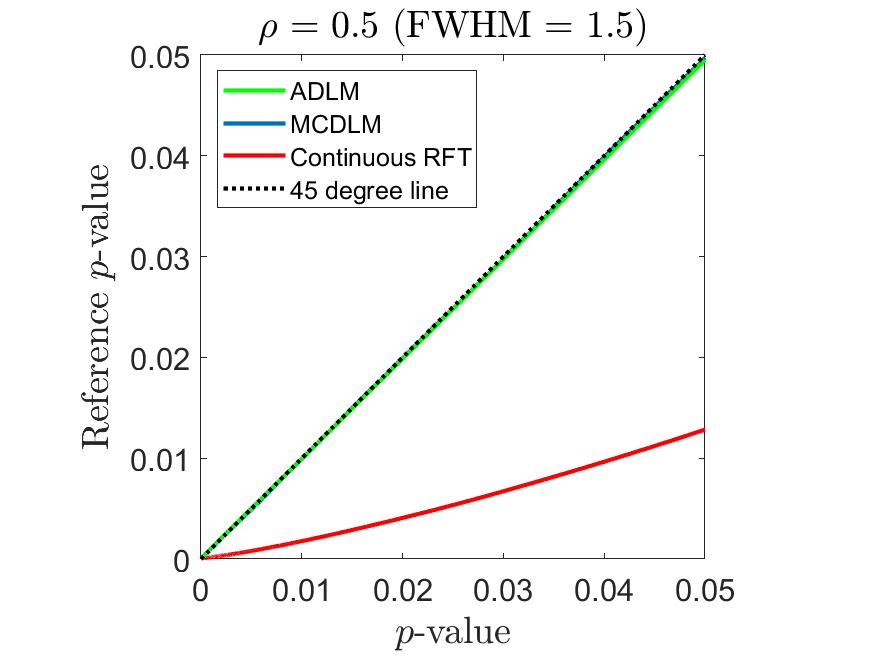}
\includegraphics[trim=80 5 80 5, clip,width=0.3\textwidth]{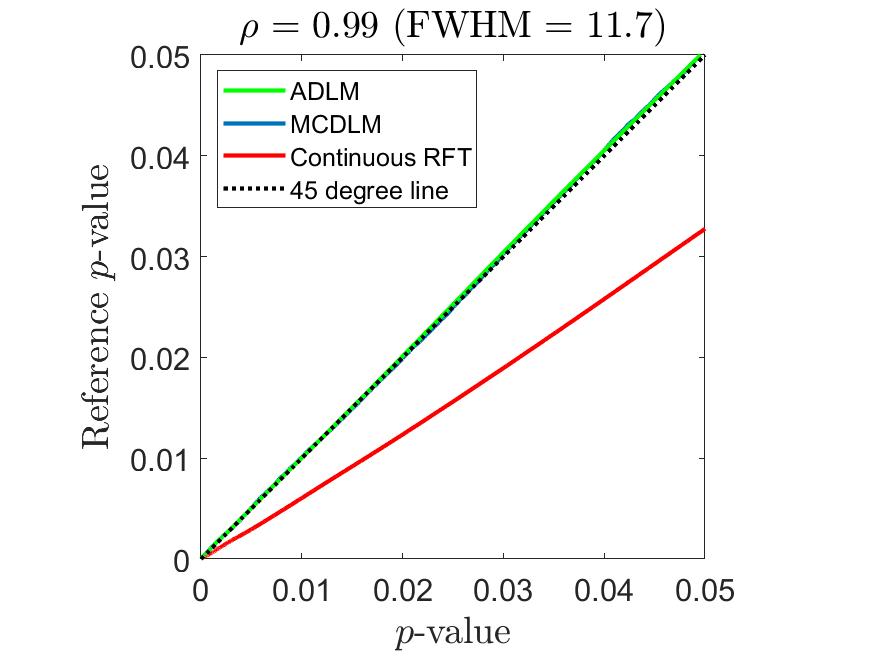}
\caption{Comparison of the peak height distribution for peaks in a partially connected neighbourhood calculated via the different methods for 2D and 3D isotropic Gaussian fields. \label{fig8}}
\end{figure}

Figure \ref{fig8} shows the comparison between MCDLM, ADLM and the continuous RFT method in 2D and 3D. In all three scenarios, the $p$-value distribution of ADLM and MCDLM match and are close to a uniform distribution. As in the main text the continuous RFT approach is conservative.

\subsection{Applying the neighborhood covariance function in \eqref{eqn2.32}}
\label{appendix.d2}
In this section we perform the same simulations as in Section \ref{sec3.1} but with the neighbourhood covariance of \cite{worsley2005improved} (given in \eqref{eqn2.32}) instead of the actual neighbourhood covariance (which we derived in \eqref{eqn2.8}. 

The results are shown in Figure \ref{fig9}. They are similar to those of Figure \ref{fig7}. However, as exemplified in the $\rho = 0.5$ case the MCDLM approach is incorrect when this covariance function is used. This is because it is not in fact the correct neighbourhood covariance. 
\begin{figure}[!htp]
\centering
\begin{sideways}
\phantom{------------------}2D
\end{sideways}
\includegraphics[trim=80 5 80 5, clip,width=0.3\textwidth]{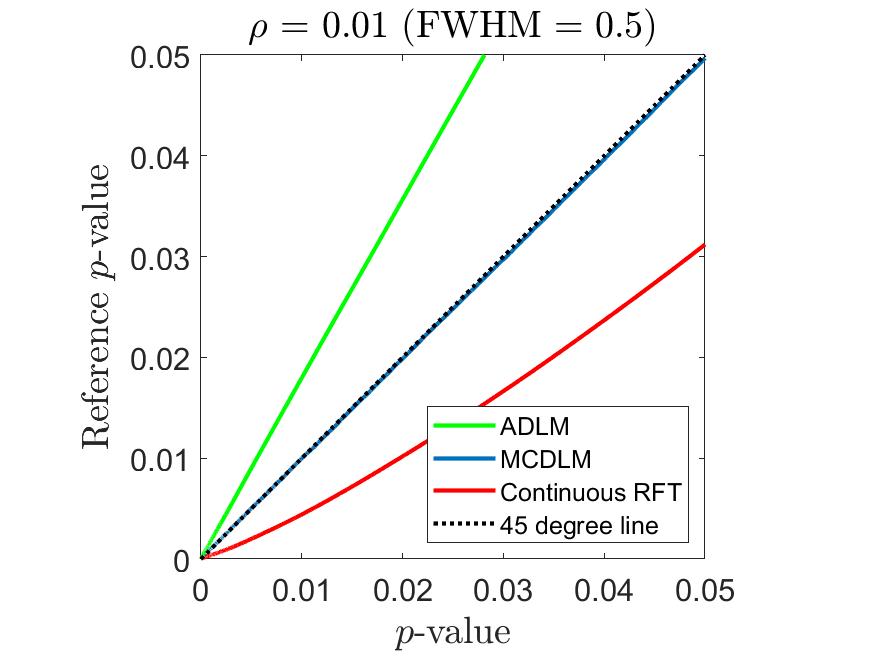}
\includegraphics[trim=80 5 80 5, clip,width=0.3\textwidth]{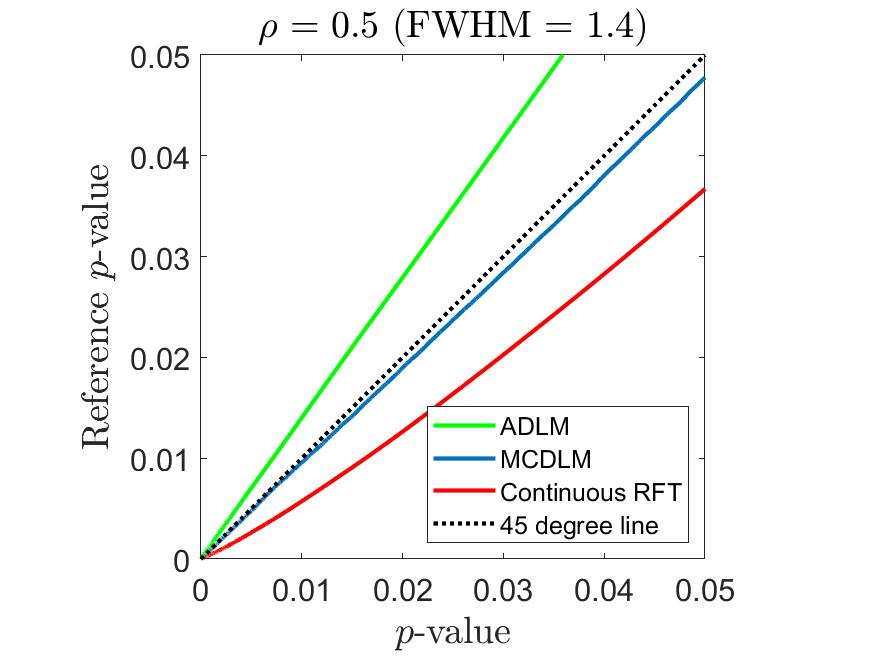}
\includegraphics[trim=80 5 80 5, clip,width=0.3\textwidth]{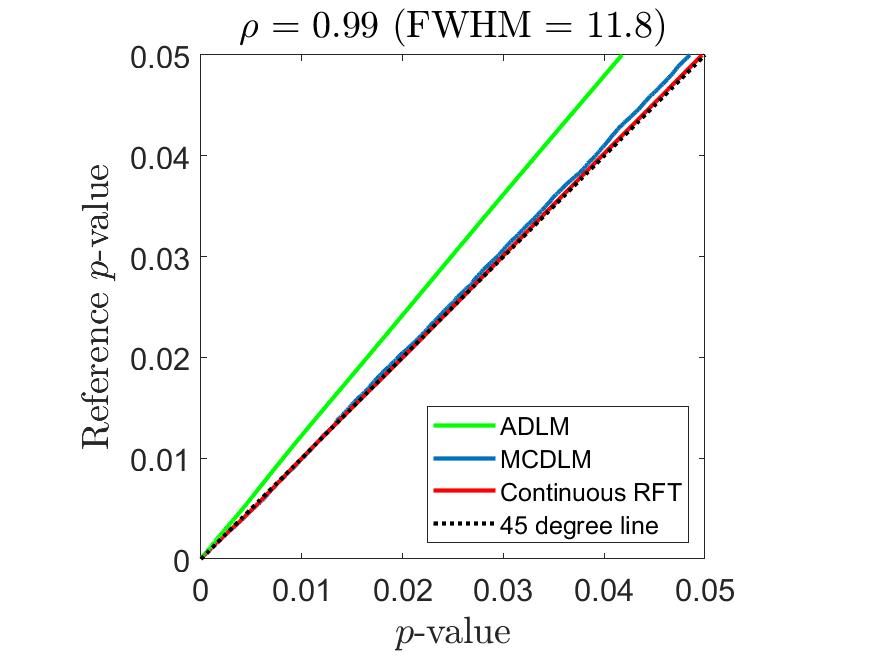}

\begin{sideways}
\phantom{------------------}3D
\end{sideways}
\includegraphics[trim=80 5 80 5, clip,width=0.3\textwidth]{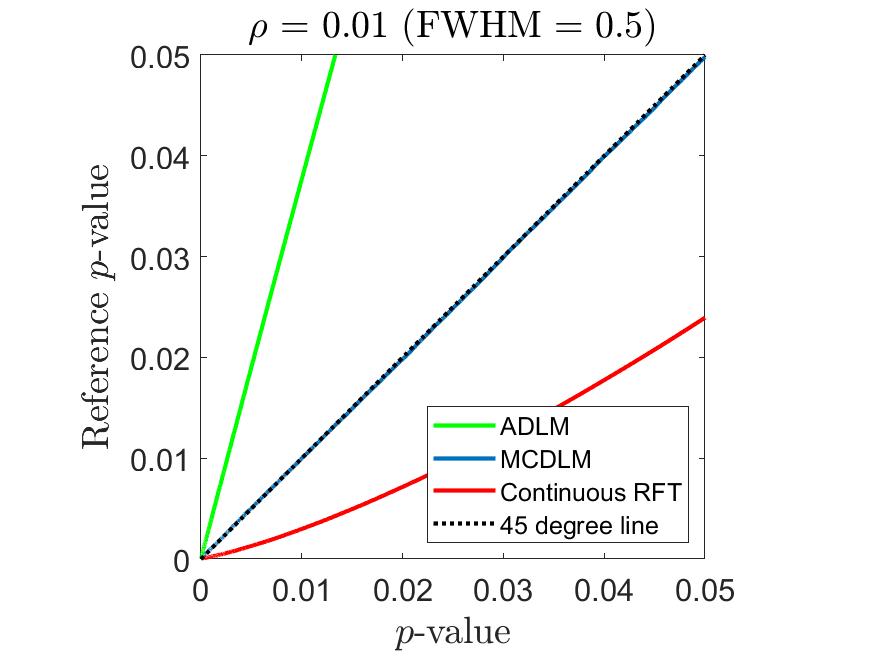}
\includegraphics[trim=80 5 80 5, clip,width=0.3\textwidth]{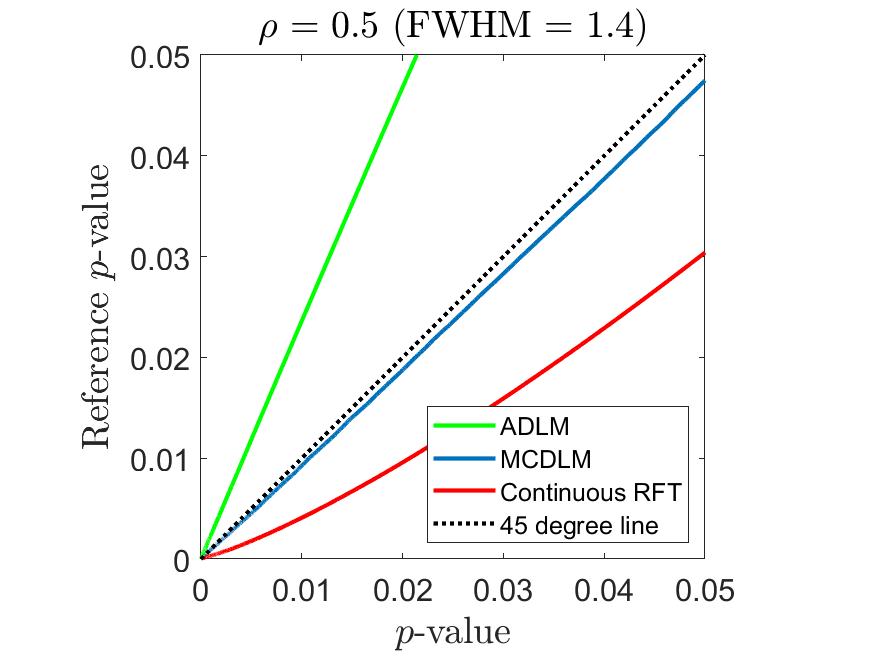}
\includegraphics[trim=80 5 80 5, clip,width=0.3\textwidth]{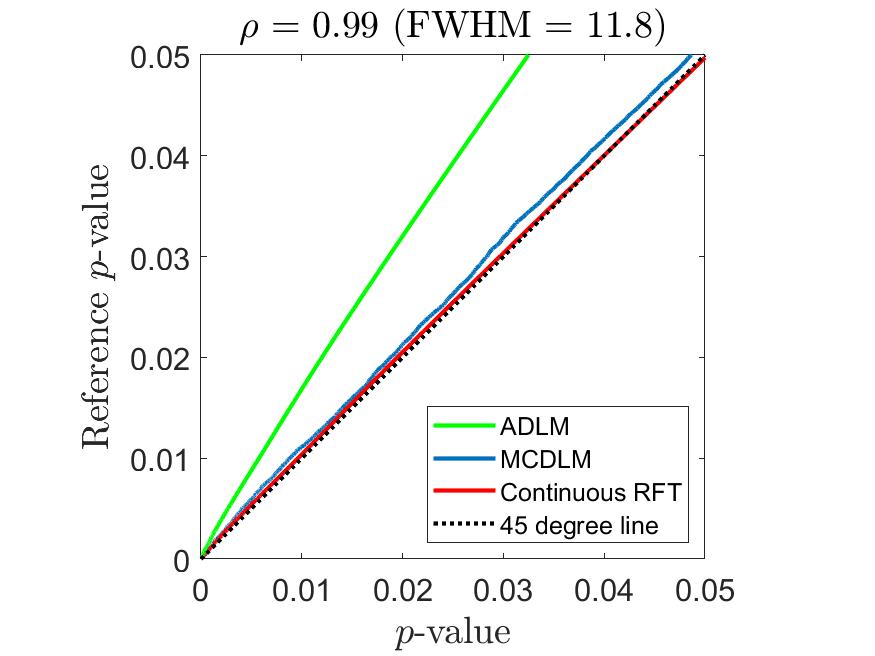}
\caption{Comparison of the peak height distribution calculated using different methods for 2D and 3D isotropic Gaussian field using  neighborhood covariance function in \eqref{eqn2.32}.\label{fig9}}
\end{figure}

\subsection{RMSE Results for comparing $p$-value curves} \label{sec.rmse}
In this section, we calculate RMSE between the $p$-value curves from each of the MCDLM, ADLM and continuous RFT method and identity line. From Table \ref{tab01}, we obtain the same conclusion as Table \ref{tab0} (when using mean ratio) that MCDLM outperforms ADLM and continuous RFT when FWHM $< 6.7$, but continuous RFT is better when FWHM $ = 6.7, 8.3 \text{ and } 11.7$. 
\begin{table}[!htp]
\caption{RMSE results from 2D isotropic Gaussian random fields for comparing the $p$-values from MCDLM, ADLM and continuous RFT approaches. The smallest value in each row is highlighted in red color. \label{tab01}}
\centering
\begin{tabular}{llll}
\hline
 & MCDLM & ADLM & Continuous RFT \\ \hline
$\rho = 0.01$ (FWHM $ = 0.7$) & \color{red}\num{1.71e-04} & \num{5.77e-03} & \num{8.18e-03} \\ 
$\rho = 0.1\ $ \ (FWHM $ = 1$) & \color{red}\num{1.91e-04} & \num{5.48e-03} & \num{7.75e-03} \\ 
$\rho = 0.3\ $ \ (FWHM $ = 1.2$) & \color{red}\num{6.16e-05} & \num{4.83e-03} & \num{6.53e-03} \\ 
$\rho = 0.5\ $ \ (FWHM $ = 1.5$) & \color{red}\num{6.34e-05} & \num{4.16e-03} & \num{4.99e-03} \\ 
$\rho = 0.7\ $ \ (FWHM $ = 2$) & \color{red}\num{1.53e-04} & \num{3.58e-03} & \num{3.01e-03} \\ 
$\rho = 0.9\ $ \ (FWHM $ = 3.6$) & \color{red}\num{1.40e-04} & \num{2.91e-03} & \num{8.97e-04} \\ 
$\rho = 0.95$ (FWHM $ = 5.2$) & \color{red}\num{1.37e-04} & \num{2.56e-03} & \num{5.57e-04} \\ 
$\rho = 0.96$ (FWHM $ = 5.8$) & \color{red}\num{8.66e-05} & \num{2.57e-03} & \num{3.38e-04} \\ 
$\rho = 0.97$ (FWHM $ = 6.7$) & \num{3.56e-04} & \num{2.73e-03} & \color{red}\num{7.84e-05} \\ 
$\rho = 0.98$ (FWHM $ = 8.3$) & \num{2.22e-04} & \num{2.61e-03} & \color{red}\num{1.38e-04} \\ 
$\rho = 0.99$ (FWHM $ = 11.7$) & \num{2.87e-04} & \num{2.54e-03} & \color{red}\num{1.29e-04} \\ \hline
\end{tabular}
\end{table}


\subsection{Calculating the look-up table} 
\label{appendix.d3}
\nt{Since our method is Monte Carlo based, it is desirable to reduce the computation time where possible. To do so, under the isotropic setting described in Section \ref{SScorrfn} we pre-record the results of $10^5$ possible local maxima values at different values of $\rho$. We vary $\rho$ form $0.01$ to $0.99$ with increments of $0.01$ and calculate a look-up table as follows.}
\begin{enumerate}
    \item Loop through the array of different values of $\rho$ and obtain $10^5$ local maxima for each $\rho$.
    \item From the union obtained in step 1, sample $10^5$ local maxima, which consist a set of local maxima, $u$, that we want to evaluate the CDF, $F(\cdot)$, at.
    \item Loop through the array of different values of $\rho$ again. For each value of $\rho$, interpolate $F(\cdot)$ at each of $u$ we get in step 2. Record all the $F(u;\rho)$ in a matrix of look-up table with row the $\rho$ and column the $u$.
\end{enumerate}
To evaluate the $p$-value at a given threshold $u$ and correlation $\rho$, we interpolate $F(\cdot)$ at $u$ and $\rho$ through pre-recorded look-up table, and the $p$-value is calculated by $1-F(u;\rho)$.

After generating a look-up table, we apply the cubic spline smoothing to smooth the noisy look-up table. The procedure of our smoothing is as follow:
\begin{enumerate}
\item Use the Cubic spline smoothing to smooth the matrix across $\rho$;
\item Use the Cubic spline smoothing to smooth the matrix we smoothed in step 1 across $u$.
\end{enumerate}

In doing so, we aim to reduce the violation of monotonicity across $u$ but also retain the smoothness. The smoothing parameters in Cubic spline smoothing is selected by 5-fold cross validation in each scenario separately.

\begin{figure}[!htp]
\centering
\includegraphics[width=6cm]{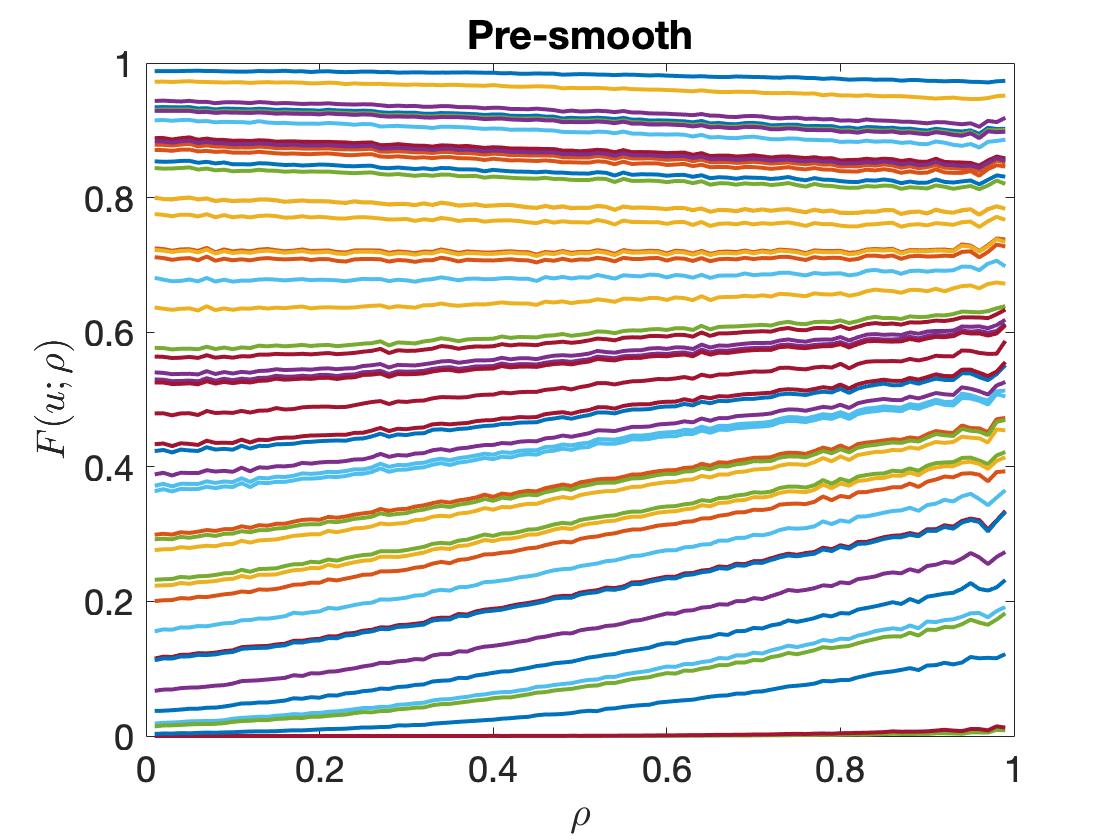}
\includegraphics[width=6cm]{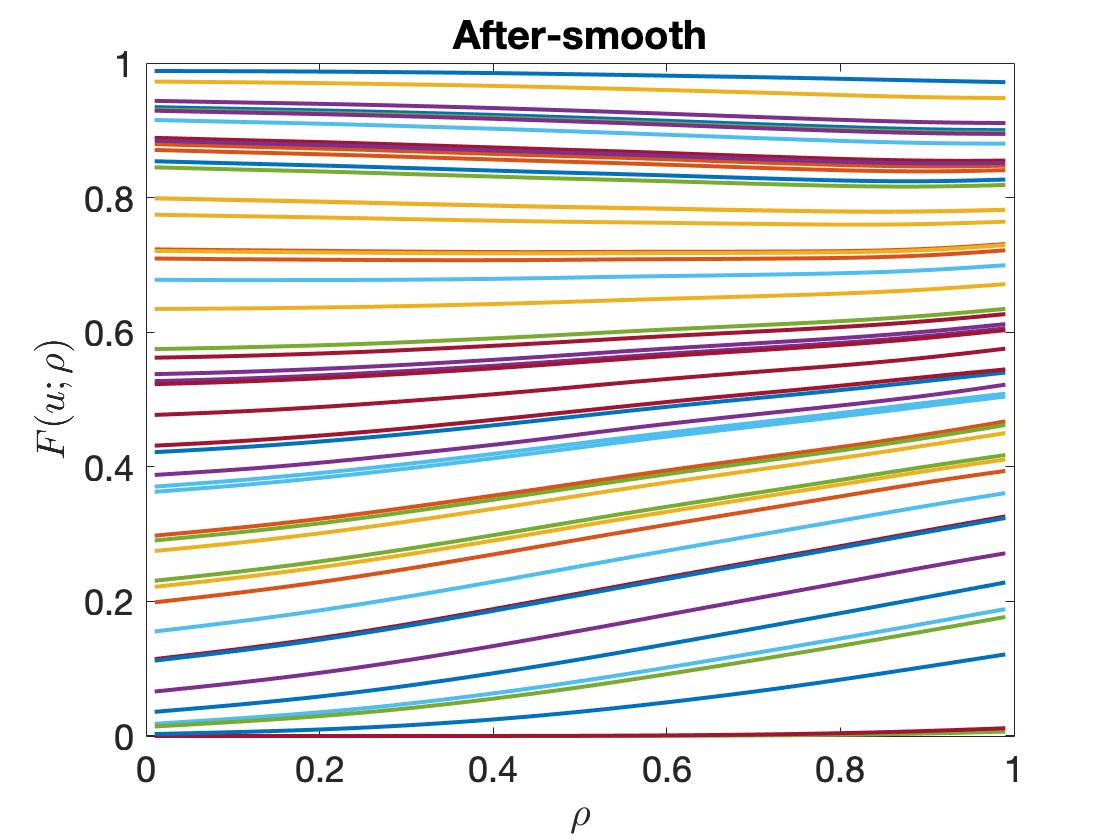}
\caption{This figure shows $F(u;\rho)$ of selected 50 samples across $\rho$ from pre-smooth table (left) and after-smooth table (right). The same color is used for the same sample before and after-smoothing.\label{fig4}}
\end{figure}

\begin{figure}[!htp]
\centering
\includegraphics[width=5cm]{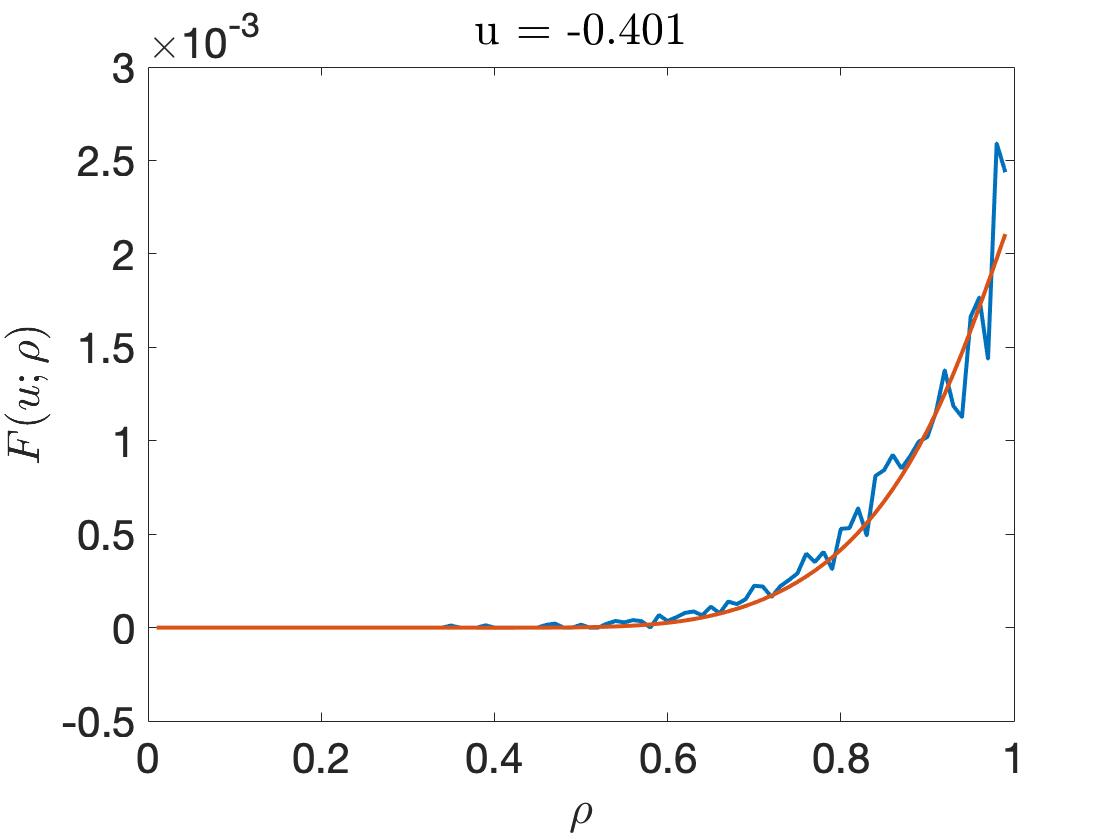}
\includegraphics[width=5cm]{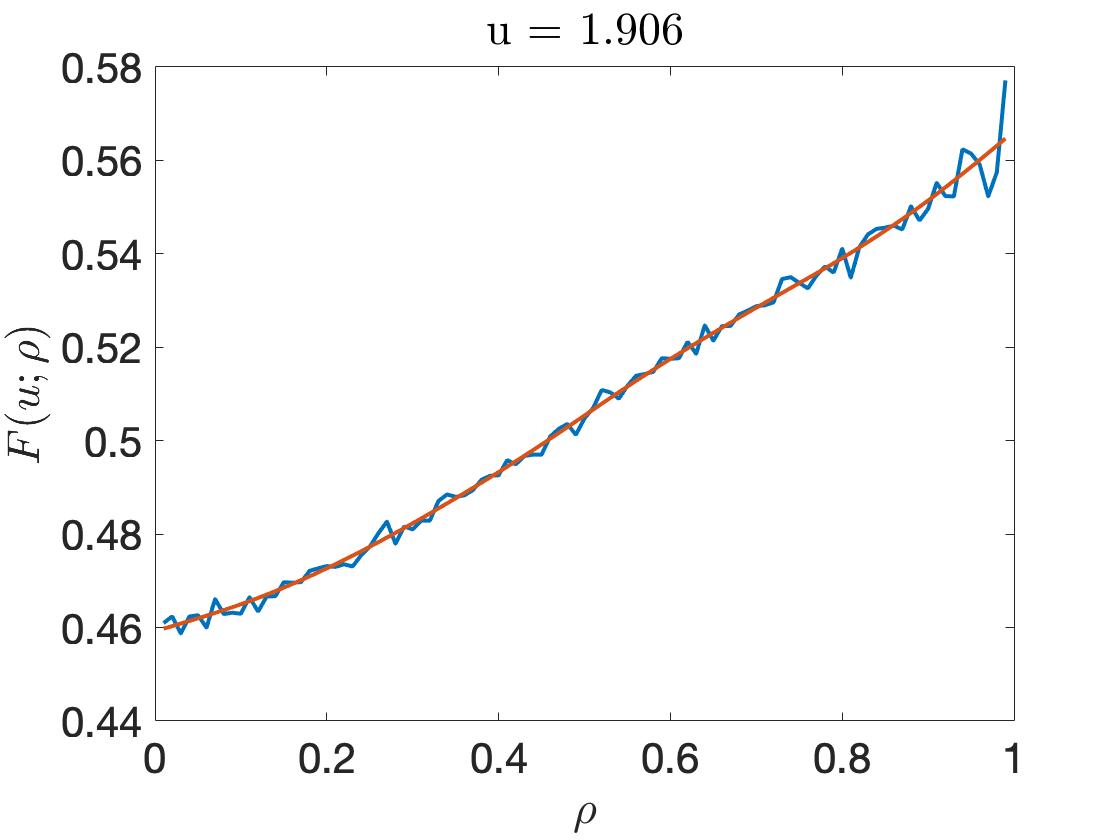}
\includegraphics[width=5cm]{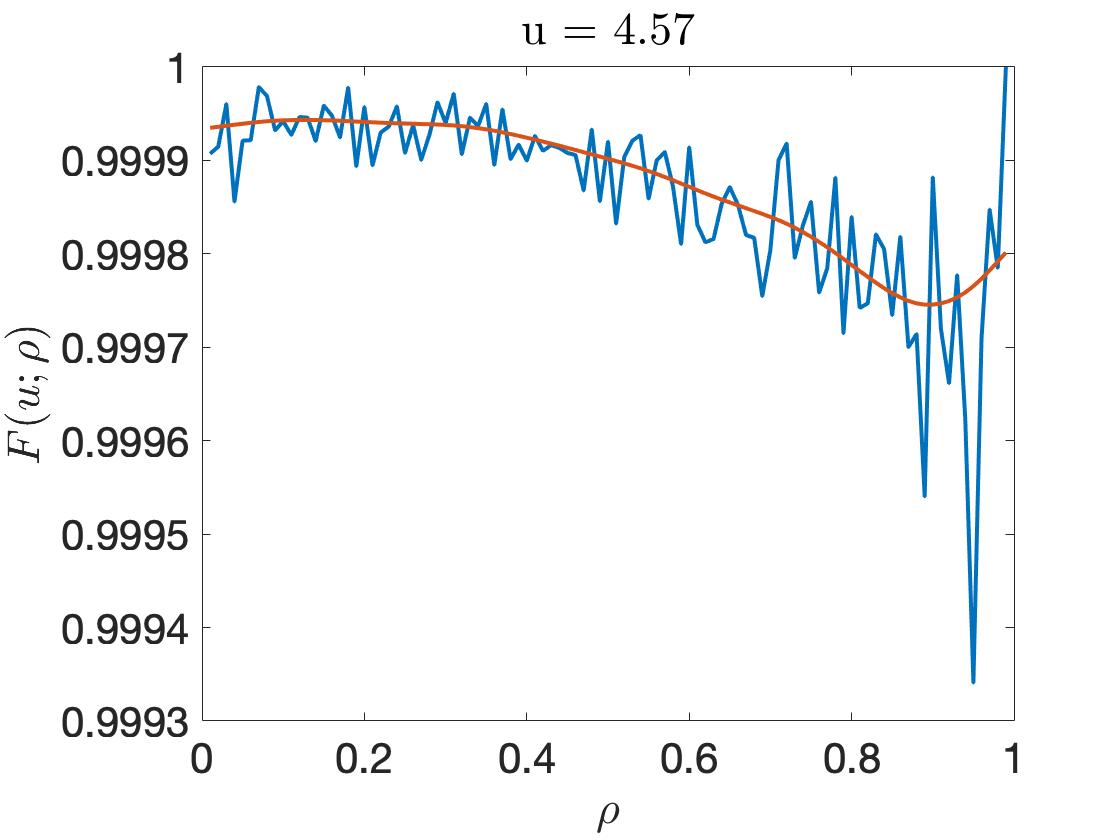}
\caption{From left to right are the 10, 50000 and 99990 columns of the pre-smooth table (red) and after-smooth table (blue). \label{fig5}}
\end{figure}

In Figure \ref{fig4}, we select 50 columns of both pre-smooth and after-smooth look-up table generated from 3D Gaussian random field and plot the CDF across different $\rho$ values. Different samples selected within one look-up table are denoted by different colors, while same color means same sample between two tables. The pre-smooth plot shows that the look-up table we generated is very noisy across $\rho$, especially at some large values of $\rho$. This will cause big problems in computing $p$-values for different values of $\rho$. For example, if two are interested in calculating the $p$-values for a $\rho$ close to 1, a slightly difference in picking the $\rho$ will bring a significant change in $p$-values, and such ambiguity could finally result in an inconsistent interpretation of the scientific findings. From the after-smooth plot, we observe that the look-up table is smooth enough to provide consistent results. 

In addition, we explore to what extent the smoothing works in Figure \ref{fig5}. In this figure, we select three columns of both pre-smooth table and after-smooth table and then compare. The $u$ that we select are -0.401, 1.906 and 4.57, which are in both ends and middle of the support. From these three plots, the noise before the smoothing is obvious, but the cubic spline smoothing fits the curve perfectly, in the sense that it preserves the general shape of the curve yet removes the noise.

\subsection{Applying Gaussianization transformation of the $t$-fields}\label{SS:gauss}
As discussed in Section \ref{sec3.2}, to improve the computation efficiency for the height distribution of peaks of $t$-fields, we consider using Gaussianization transformation of the $t$-fields. In this section we perform the same simulations as in Section \ref{sec3.2} but the simulated $t$-fields were Gaussianized \nt{as in \eqref{eqn.t2Gauss} and peak height distributions were obtained using MCDLM and continuous RFT after the transformation.}

The results are shown in Figure \ref{fig.t2gauss2D} (2D) and Figure \ref{fig.t2gauss3D} (3D). MCDLM works well when degrees of freedom is large and $\rho$ is small. At high smoothness, MCDLM and continuous RFT work similarly. They are only correct when degrees of freedom is large since in this case the $t$-fields can approximate Gaussian fields.

\begin{figure}[!htp]
\centering
\begin{sideways}
\phantom{------------------}$\nu = 20$
\end{sideways}
\includegraphics[trim=70 5 100 5, clip,width=0.3\textwidth]{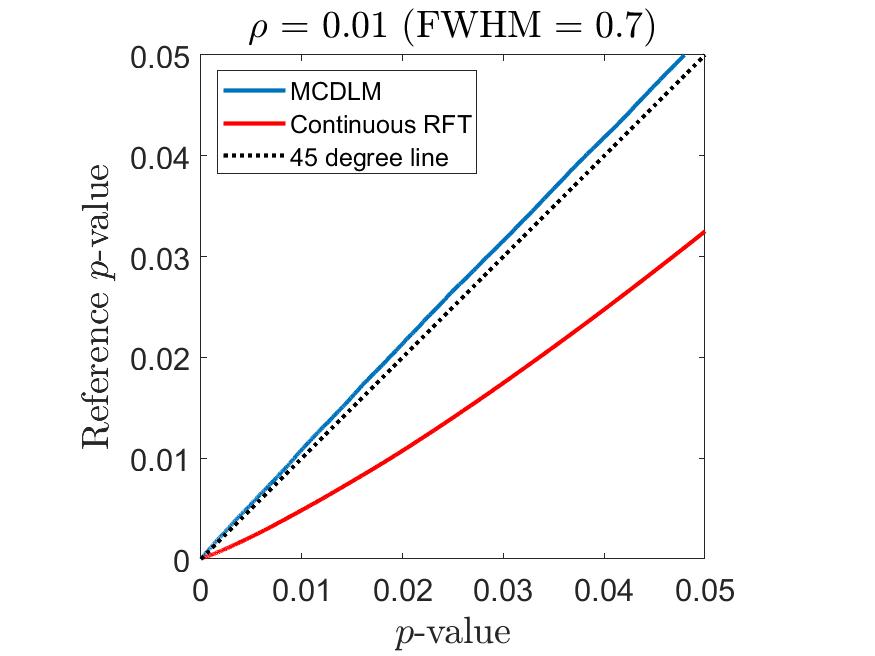}
\includegraphics[trim=70 5 100 5, clip,width=0.3\textwidth]{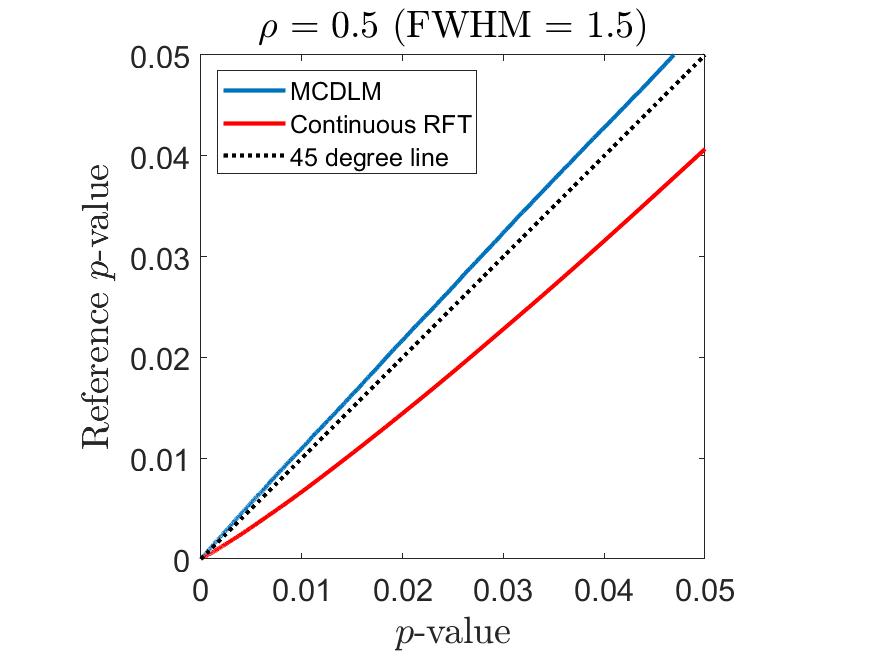}
\includegraphics[trim=70 5 100 5, clip,width=0.3\textwidth]{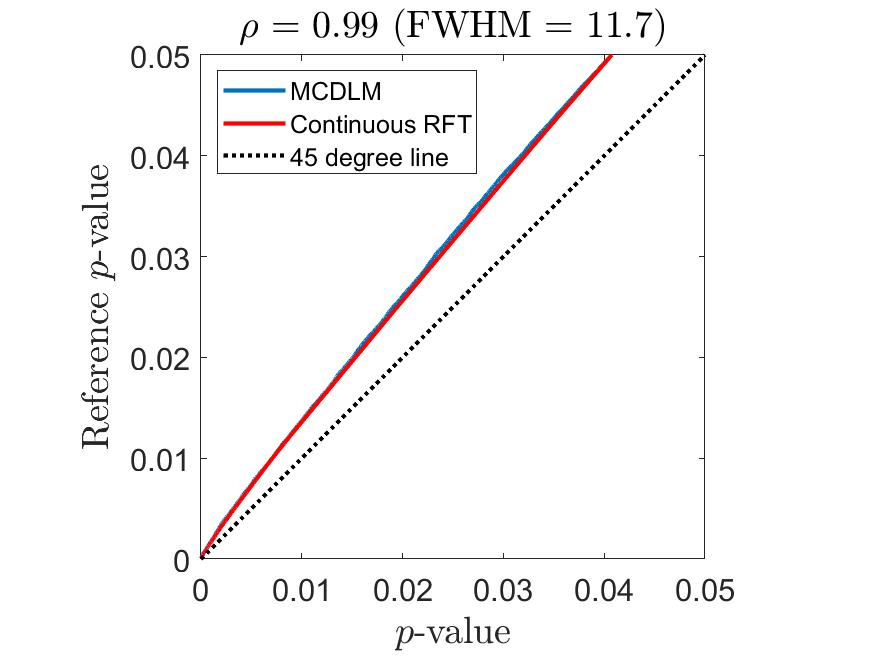}

\begin{sideways}
\phantom{------------------}$\nu = 50$
\end{sideways}
\includegraphics[trim=70 5 100 5, clip,width=0.3\textwidth]{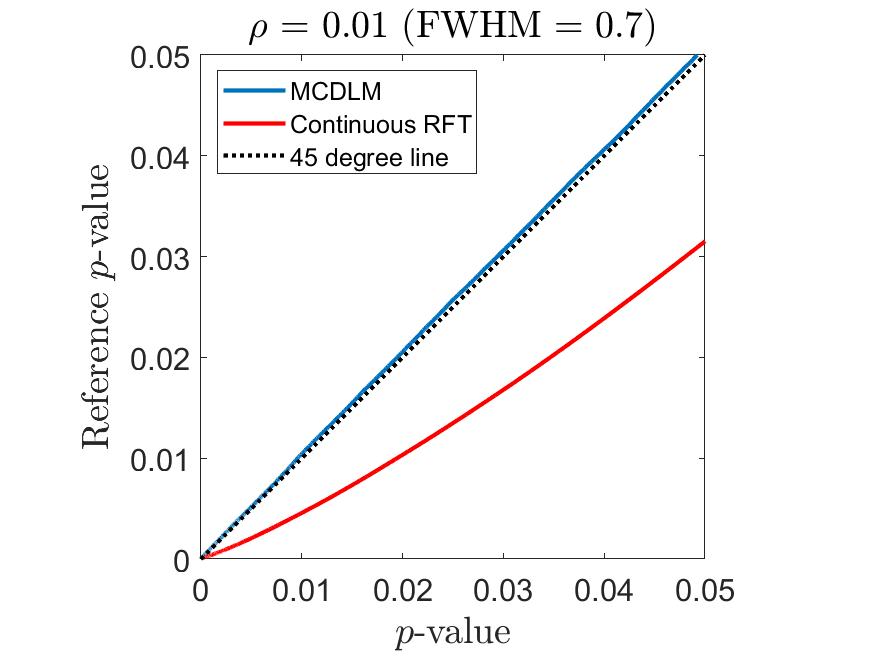}
\includegraphics[trim=70 5 100 5, clip,width=0.3\textwidth]{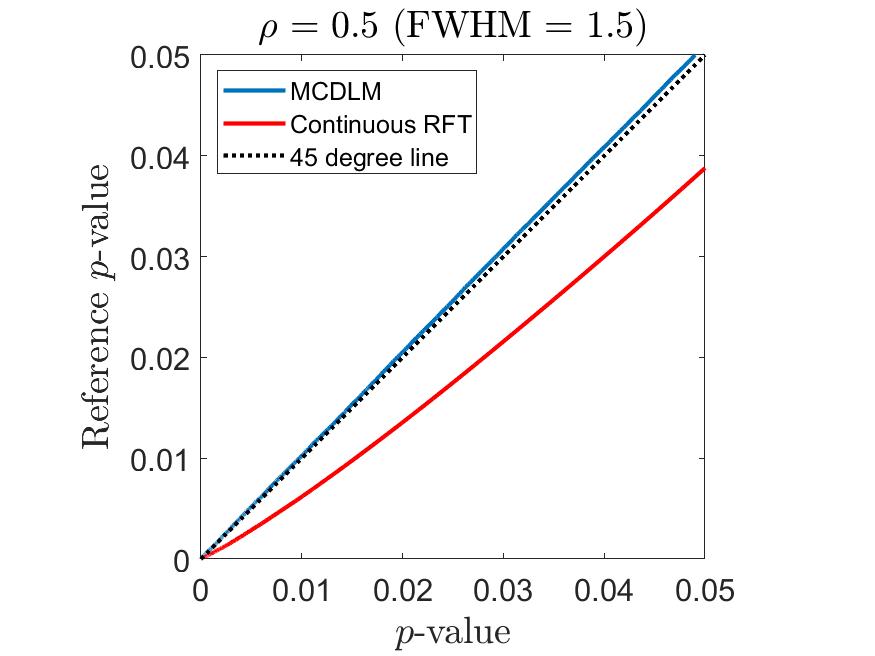}
\includegraphics[trim=70 5 100 5, clip,width=0.3\textwidth]{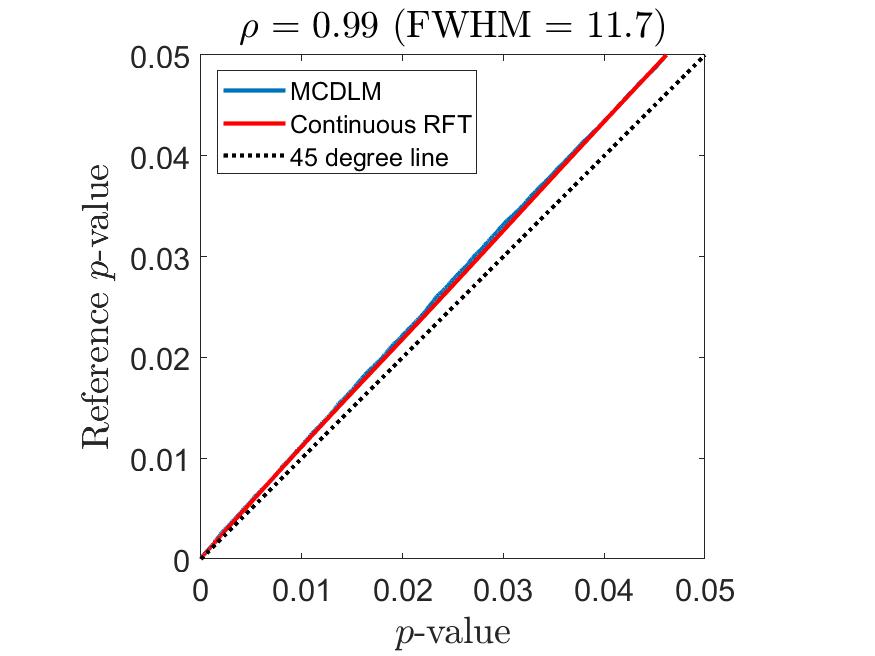}

\begin{sideways}
\phantom{------------------}$\nu = 200$
\end{sideways}
\includegraphics[trim=70 5 100 5, clip,width=0.3\textwidth]{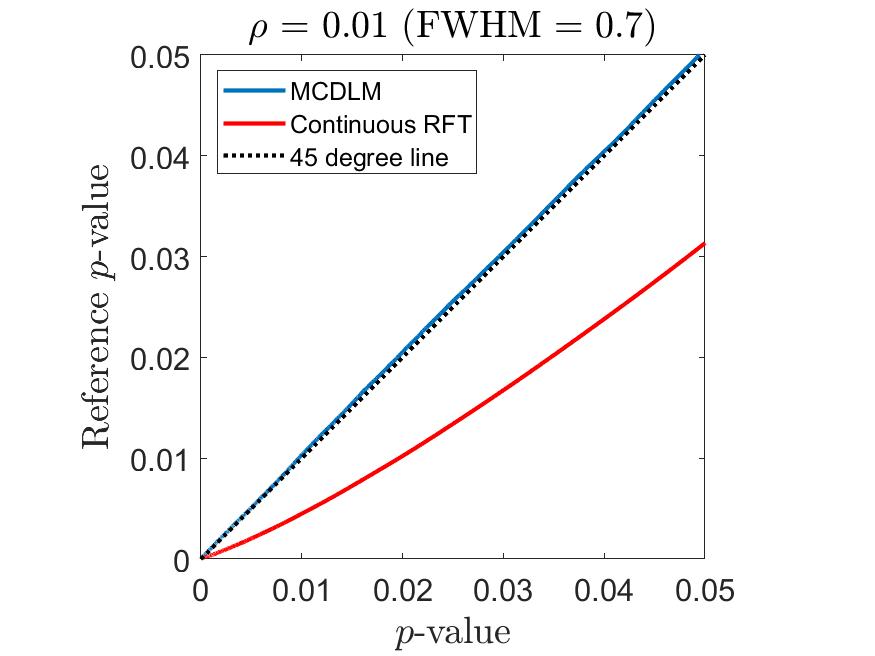}
\includegraphics[trim=70 5 100 5, clip,width=0.3\textwidth]{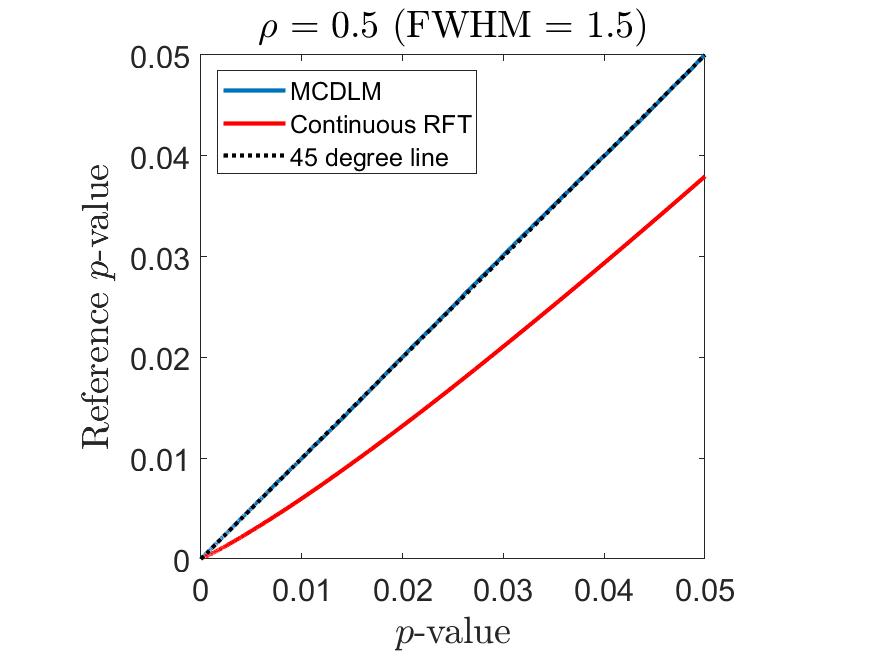}
\includegraphics[trim=70 5 100 5, clip,width=0.3\textwidth]{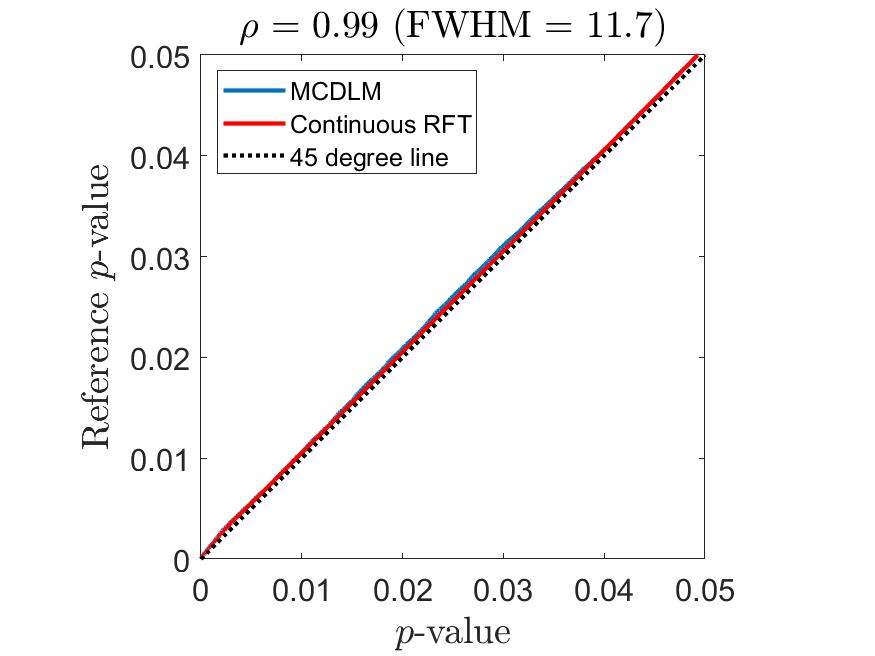}
\caption{Comparison of methods for calculating the peak height distribution of a Gaussianized 2D $t$-field with $\nu$ degrees of freedom. \label{fig.t2gauss2D}}
\end{figure}

\begin{figure}[!htp]
\centering
\begin{sideways}
\phantom{------------------}$\nu = 20$
\end{sideways}
\includegraphics[trim=70 5 100 5, clip,width=0.3\textwidth]{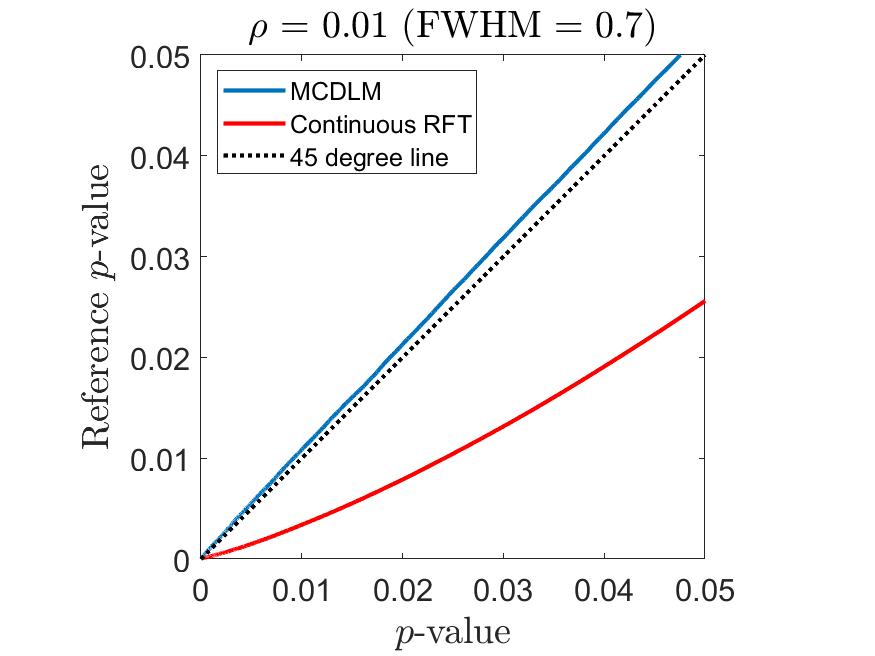}
\includegraphics[trim=70 5 100 5, clip,width=0.3\textwidth]{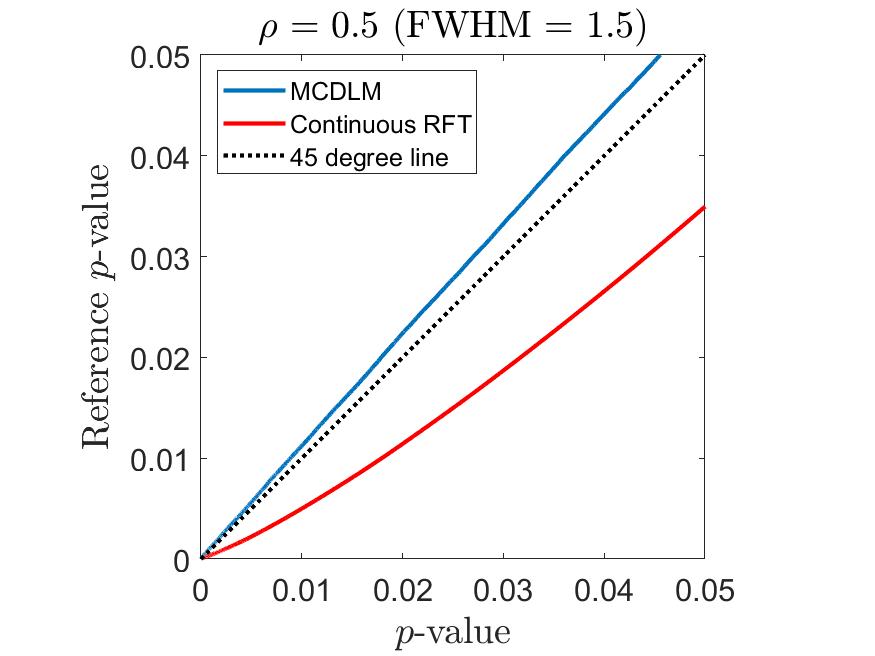}
\includegraphics[trim=70 5 100 5, clip,width=0.3\textwidth]{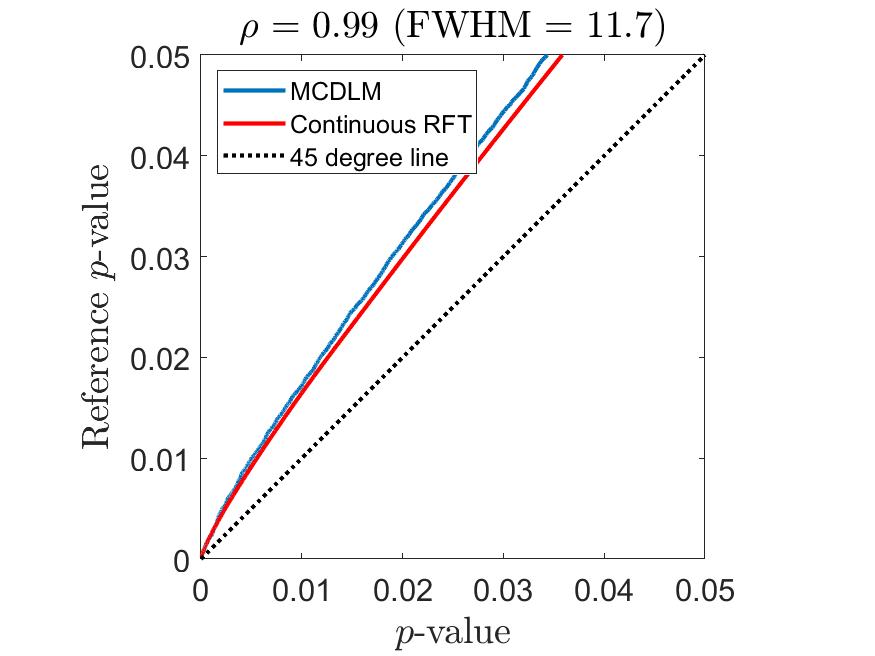}

\begin{sideways}
\phantom{------------------}$\nu = 50$
\end{sideways}
\includegraphics[trim=70 5 100 5, clip,width=0.3\textwidth]{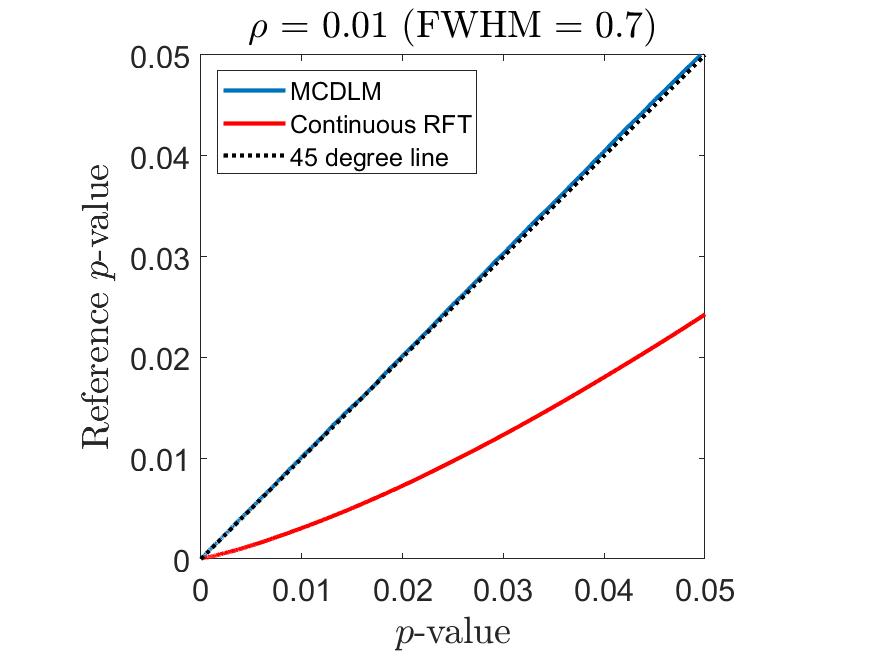}
\includegraphics[trim=70 5 100 5, clip,width=0.3\textwidth]{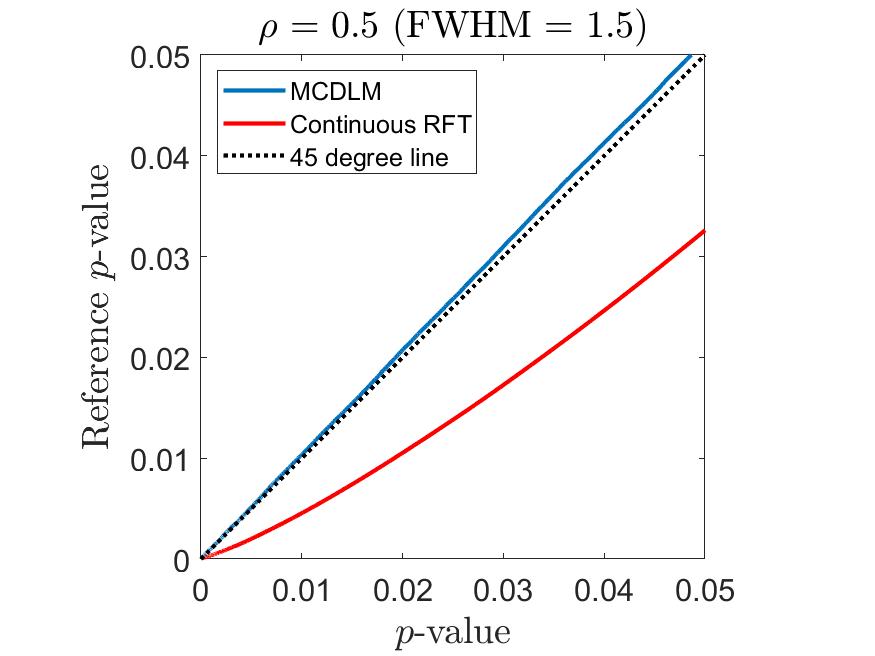}
\includegraphics[trim=70 5 100 5, clip,width=0.3\textwidth]{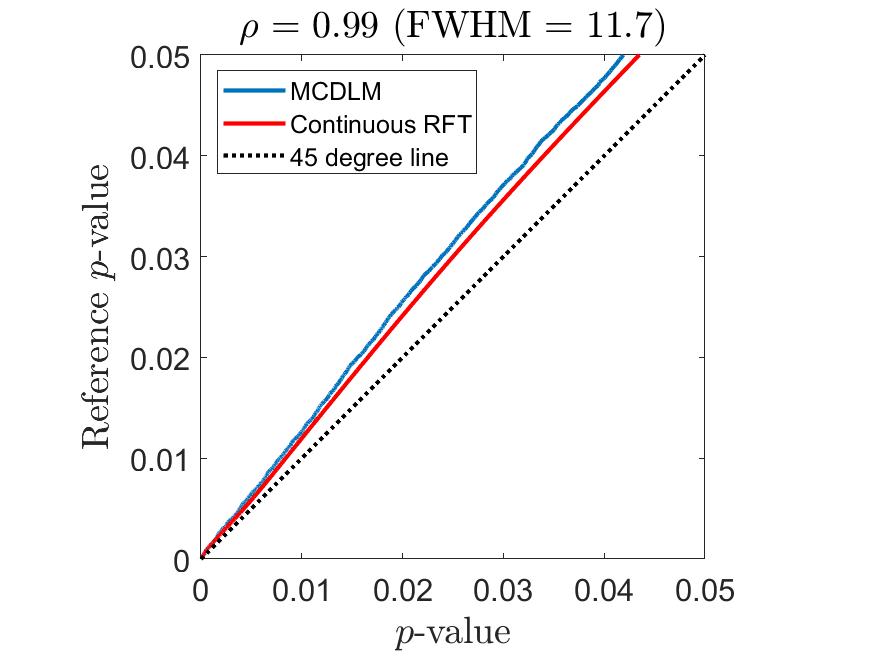}

\begin{sideways}
\phantom{------------------}$\nu = 200$
\end{sideways}
\includegraphics[trim=70 5 100 5, clip,width=0.3\textwidth]{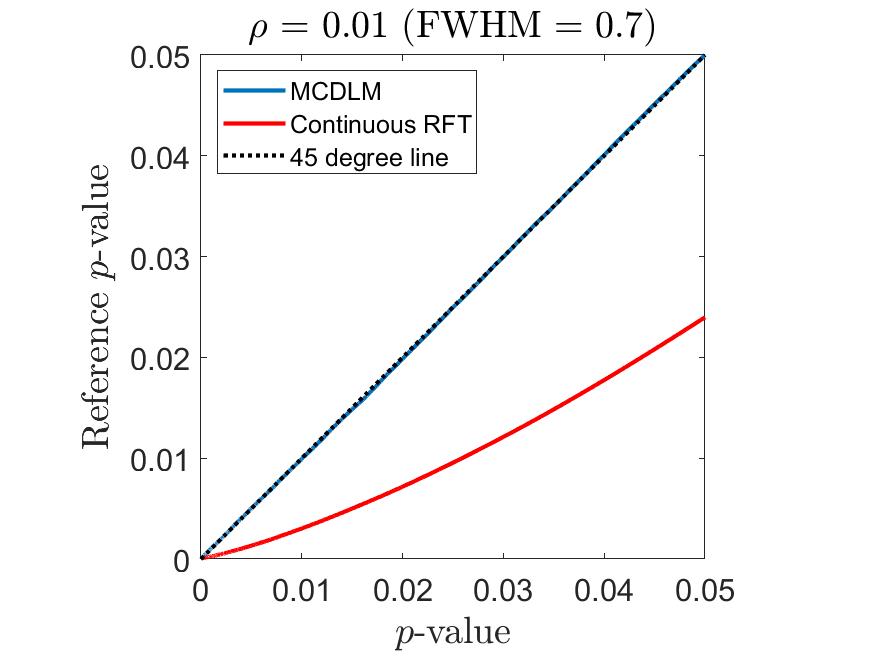}
\includegraphics[trim=70 5 100 5, clip,width=0.3\textwidth]{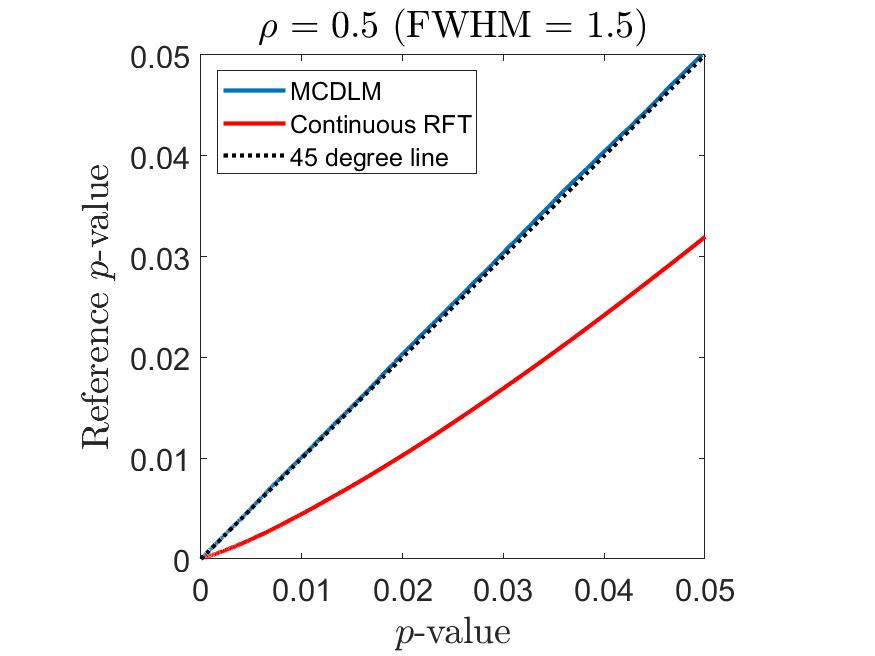}
\includegraphics[trim=70 5 100 5, clip,width=0.3\textwidth]{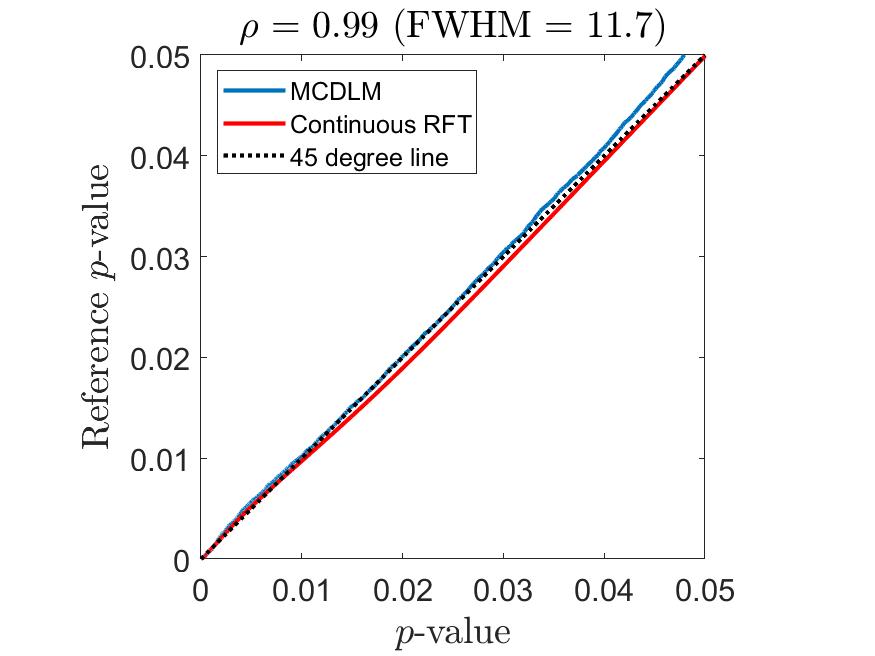}
\caption{Comparison of methods for calculating the peak height distribution of a Gaussianized 3D $t$-field with $\nu$ degrees of freedom.\label{fig.t2gauss3D}}
\end{figure}

\subsection{Additional simulation results for MCDLM based on estimated neighborhood covariance \label{appendix.d5}}
In Section \ref{sec4.4} we observed that the performance of MCDLM with the estimated covariance worsened at very high values of $\rho$. Here we consider additional values in order to better understand the performance in practice in particular we consider $\rho = 0.9, 0.93, 0.95$. In Figure \ref{fig.appendix.d5}, when $\rho < 0.97$, or FWHM $<6.7$, MCDLM with estimated covariance works well for both 2D and 3D cases given sufficiently many random fields available to estimate the covariance function (with nsim $\geq 50$ sufficing in most cases). By comparing the neighborhood covariance matrices \eqref{cov.theory} (theoretical covariance) and \eqref{cov.empirical} (estimated covariance) for $\rho = 0.99$ shown below in 2D, the estimate is very close to the true covariance. However, the covariance matrix in this case is nearly singular, which causes problems when simulating using Algorithm \ref{alg1}. 
\begin{equation}
\label{cov.theory}
\begin{bmatrix}
1.0000 &   0.9900 &   0.9606  &  0.9900   & 0.9801  &  0.9510 &   0.9606  &  0.9510  &  0.9227\\
    0.9900  &  1.0000  &  0.9900  &  0.9801  &  0.9900  &  0.9801 &   0.9510  &  0.9606  &  0.9510\\
    0.9606 &   0.9900  &  1.0000  &  0.9510   & 0.9801 &   0.9900  &  0.9227   & 0.9510  &  0.9606\\
    0.9900  &  0.9801 &   0.9510   & 1.0000   & 0.9900  &  0.9606  &  0.9900  &  0.9801 &   0.9510\\
    0.9801  &  0.9900  &  0.9801 &   0.9900  &  1.0000  &  0.9900   & 0.9801   & 0.9900  &  0.9801\\
    0.9510  &  0.9801 &   0.9900  &  0.9606  &  0.9900  &  1.0000 &   0.9510  &  0.9801 &   0.9900\\
    0.9606  &  0.9510  &  0.9227  &  0.9900  &  0.9801  &  0.9510  &  1.0000  &  0.9900 &   0.9606\\
    0.9510  &  0.9606  &  0.9510  &  0.9801  &  0.9900  &  0.9801  &  0.9900  &  1.0000 &   0.9900\\
0.9227   & 0.9510  &  0.9606  &  0.9510  &  0.9801  &  0.9900 &   0.9606 &   0.9900  & 1.0000

\end{bmatrix}
\end{equation}

\begin{equation}
\label{cov.empirical}
\begin{bmatrix}
1.0000  &  0.9906  &  0.9612  &  0.9903  &  0.9809  &  0.9516  &  0.9606  &  0.9513  &  0.9228\\
    0.9906  &  1.0000 &   0.9906  &  0.9811  &  0.9903  &  0.9809  &  0.9518  &  0.9606  &  0.9513 \\
    0.9612  &   0.9906  &   1.0000& 0.9521 &    0.9811 &    0.9903  &   0.9237  &   0.9518 &    0.9606 \\
    0.9903 &    0.9811 &    0.9521  &   1.0000  &   0.9906  &   0.9612  &   0.9903  &   0.9809  &   0.9516 \\
    0.9809  &   0.9903  &   0.9811  &   0.9906  &   1.0000  &   0.9906 &    0.9811  &   0.9903  &   0.9809 \\
    0.9516  &   0.9809 &    0.9903 &   0.9612  &   0.9906 &    1.0000 &    0.9521  &   0.9811&     0.9903 \\
    0.9606  &   0.9518  &   0.9237  &   0.9903  &   0.9811 &    0.9521 &   1.0000  &   0.9906 &    0.9612 \\
    0.9513   &  0.9606  &   0.9518 &    0.9809 &    0.9903  &   0.9811  &   0.9906  &   1.0000   &  0.9906 \\
    0.9228  &   0.9513 &    0.9606  &   0.9516  &   0.9809  &   0.9903  &   0.9612  &   0.9906  &   1.0000 \\
\end{bmatrix} 
\end{equation}

\begin{figure}[!htp]
\centering
\begin{sideways}
\phantom{------------------}2D
\end{sideways}
\includegraphics[trim=80 5 80 5, clip,width=0.3\textwidth]{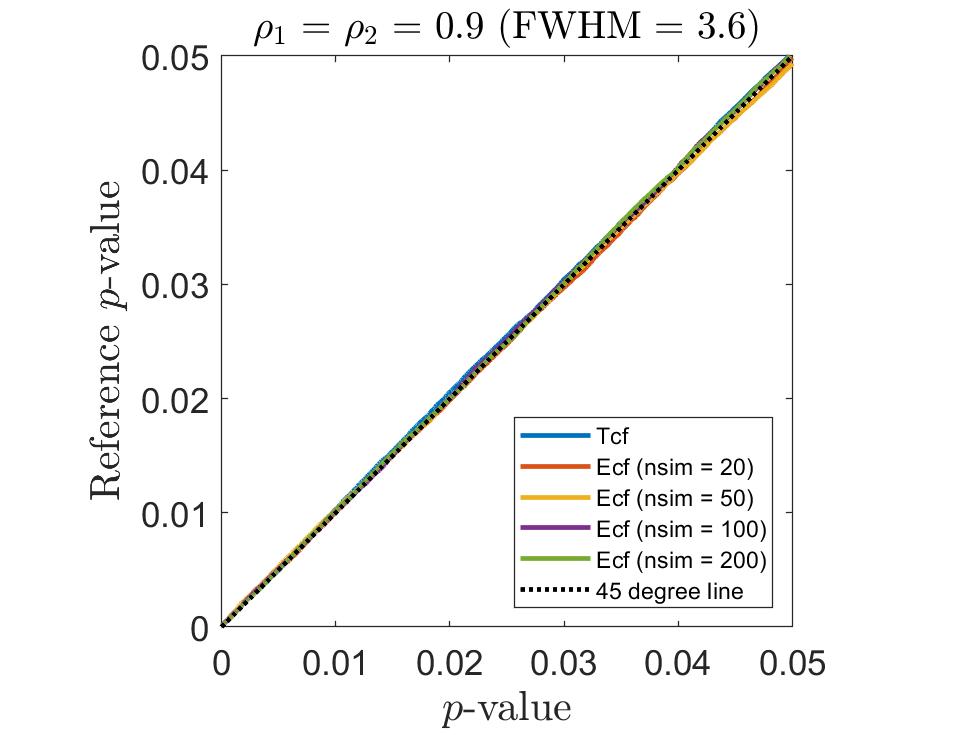}
\includegraphics[trim=80 5 80 5, clip,width=0.3\textwidth]{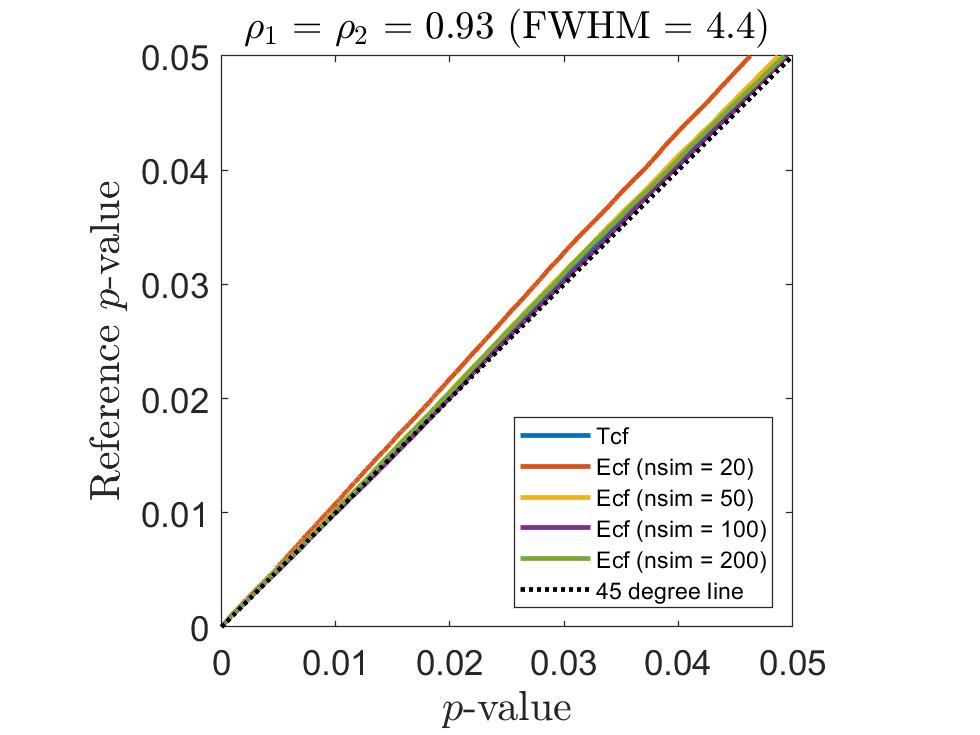}
\includegraphics[trim=80 5 80 5, clip,width=0.3\textwidth]{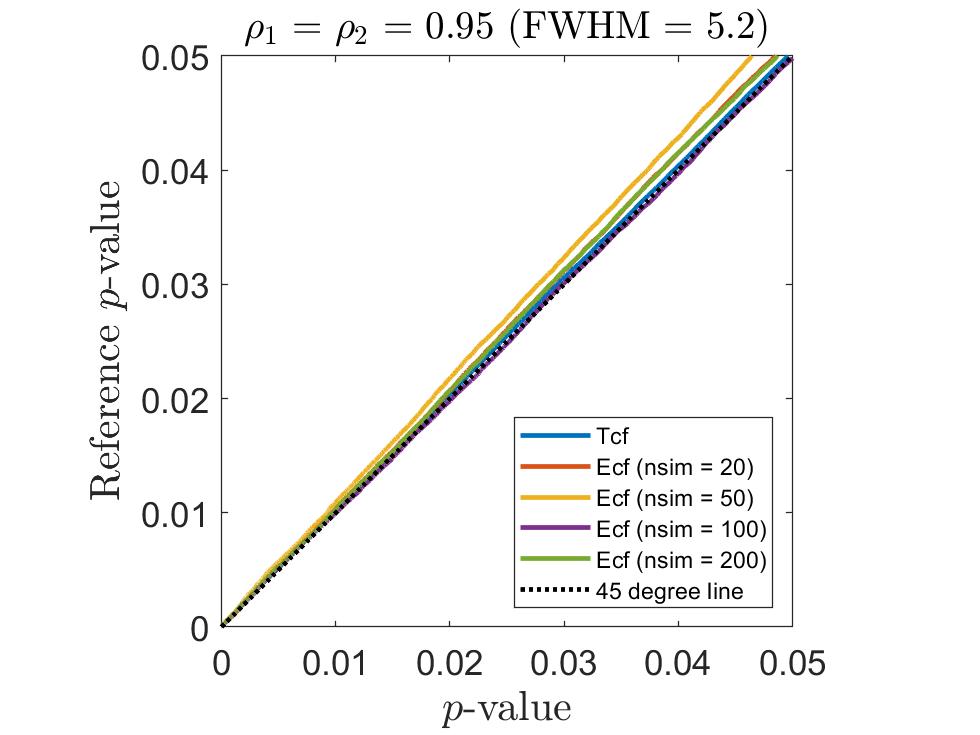}

\begin{sideways}
\phantom{------------------}3D
\end{sideways}
\includegraphics[trim=80 5 80 5, clip,width=0.3\textwidth]{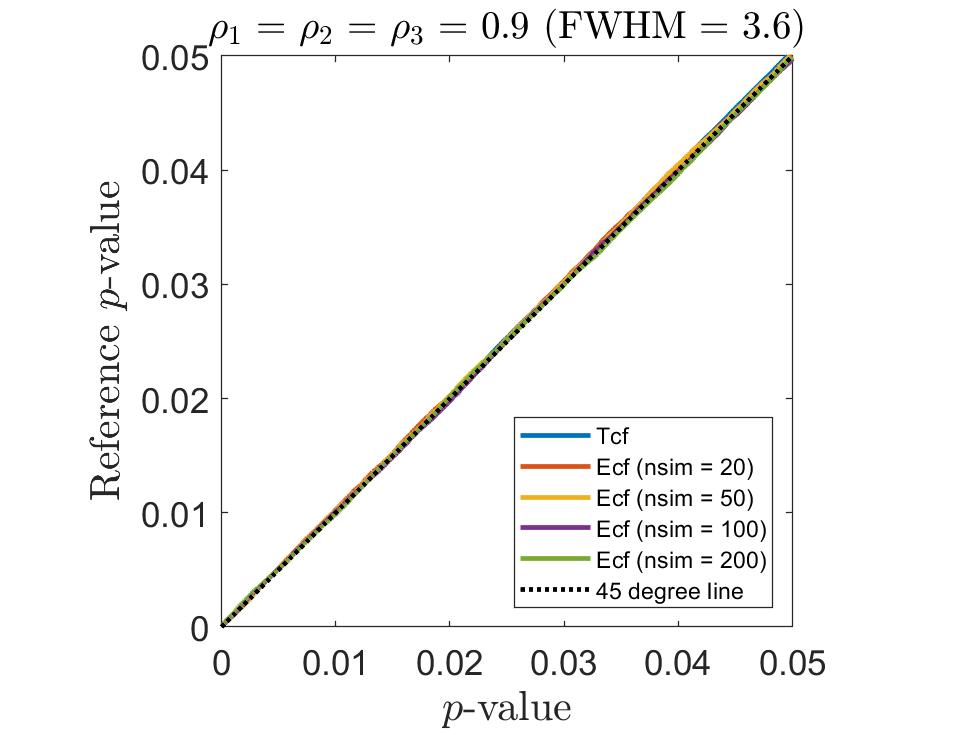}
\includegraphics[trim=80 5 80 5, clip,width=0.3\textwidth]{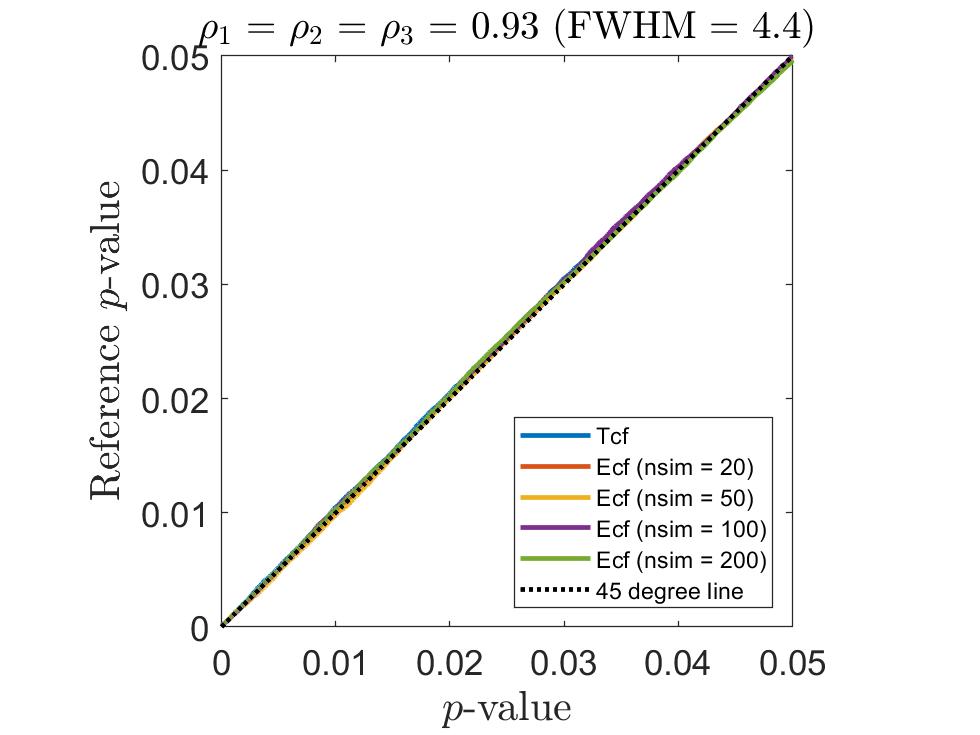}
\includegraphics[trim=80 5 80 5, clip,width=0.3\textwidth]{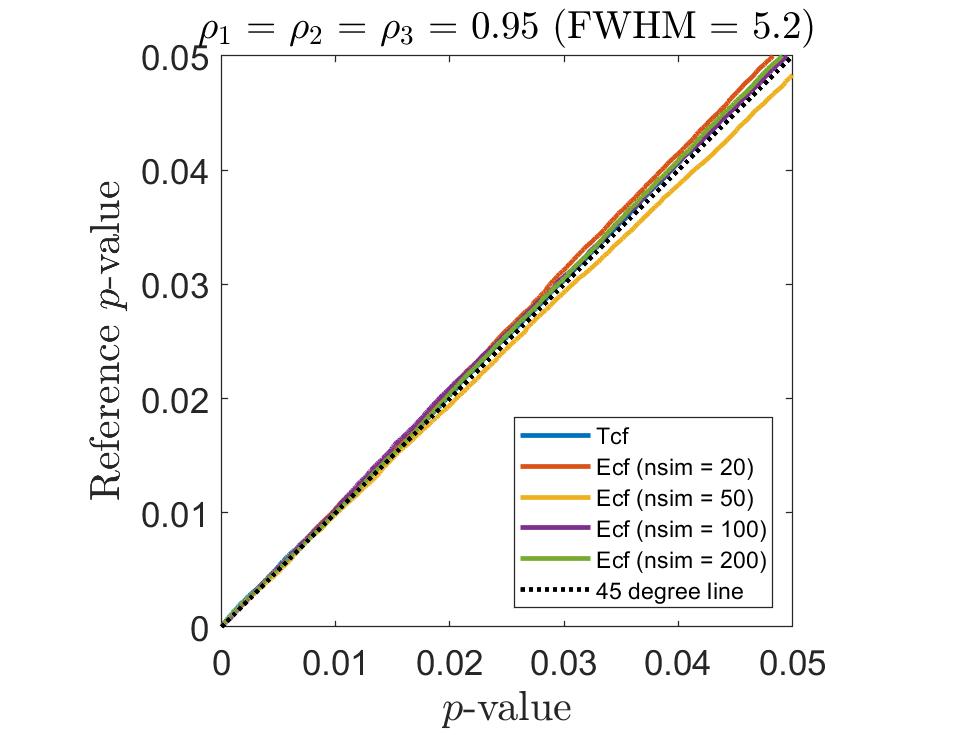}
\caption{Comparison of the peak height distribution obtained using MCDLM with different estimated neighborhood covariance for 2D and 3D isotropic Gaussian fields. The covariance functions used here are theoretical covariance function (Tcf) and empirically estimated covariance function (Ecf). The number of random fields used to estimate the covariance function is denoted using nsim. From left to right: $\rho = 0.9, 0.93, 0.95$. \label{fig.appendix.d5}}
\end{figure}

\section{Computation time table\label{appendix.g}}
The simulation times under different scenarios are shown in Table \ref{tab1}. In the case of 2D Gaussian fields, applying MCDLM with full or partial connectivity is comparable in terms of running time, but the look up table is much faster (between $5$ and $50$ times faster), becoming more efficient as the correlation increases. In the case of 3D Gaussian fields, running times are $5$ to $10$ times larger than in 2D. However, the look up table method is just as fast as in 2D, making its use very worthwhile from a computational standpoint. In the case of 2D $t$-field, the running times increase about 10 to 15 times from degrees of freedom equals $20$ to $200$. In the case of 3D $t$-field, the running times are $5$ to $10$ times larger than in 2D, making the code stop at a pre-specified threshold when degrees of freedom increases to $200$. The MCDLM with empirical covariance function has similar running times to MCDLM with theoretical covariance function, in the case of 2D Gaussian field with the number of fields to estimate the covariance small (200 or 1000). In 3D case, when the number of fields to estimate the covariance is 200, the running times are 2 times larger than the case of applying the theoretical covariance function. When the number of fields increases to 1000, the time to estimate the covariance further increases, leading to the running times 4-8 times larger than \nt{when using} the theoretical covariance function.

\begin{table}[!htp]
\small
\begin{center}
\caption{Running time (in seconds) of our MCDLM method under different scenarios. For $\rho = 0.01$ and $\rho = 0.05$,  $n = 1e6$ peak height values \nt{are simulated} and for $\rho = 0.99$, $n = 2e5$ peak height values. In some extreme cases, the code stops at a pre-specified threshold with the number of instances generated recorded in parentheses.\label{tab1}}
\begin{tabular}{|l|l|l|l|}
\hline & $\rho = 0.01$ ($n = 1e6$) & $\rho = 0.5$ ($n = 1e6$) & $\rho = 0.99$ ($n = 2e5$) \\ \hline
\multicolumn{4}{|c|}{2D Gaussian field}                                                                                          \\ \hline
Full connectivity (continuous covariance function)  & $9.81$                    & $13.83$                  & $106.41$                  \\ \hline
Full connectivity (discrete covariance function)    & $9.95$                    & $14.76$                  & $112.24$                  \\ \hline
Partial connectivity (discrete covariance function) & $7.43$                    & $12.15$                  & $110.86$                  \\ \hline
Full connectivity (look up table)                   & $1.67$                    & $1.77$                   & $1.88$                    \\ \hline
\multicolumn{4}{|c|}{2D $t$-field}                                                                                                 \\ \hline
$\nu = 20$                                          & $74.48$                   & $105.74$                 & $835.80$                  \\ \hline
$\nu = 50$                                          & $217.63$                  & $307.56$                 & $1992.58$                 \\ \hline
$\nu = 200$                                         & $1064.24$                 & $1356.55$                & $1646.86$ ($n = 37289$)     \\ \hline
\multicolumn{4}{|c|}{3D Gaussian field}                                                                                          \\ \hline
Full connectivity (continuous covariance function)  & $41.79$                   & $64.16$                  & $1395.14$ ($n = 1e5$)       \\ \hline
Full connectivity (discrete covariance function)    & $41.79$                   & $66.32$                  & $1377.69$ ($n = 1e5$)       \\ \hline
Partial connectivity (discrete covariance function) & $13.87$                   & $29.18$                  & $1753.71$                 \\ \hline
Full connectivity (look up table)                   & $1.56$                    & $2.04$                   & $1.76$                    \\ \hline
\multicolumn{4}{|c|}{3D $t$-field}                                                                                                 \\ \hline
$\nu = 20$                                          & $659.09$                  & $1043.29$                & $2199.84$ ($n = 11614$)     \\ \hline
$\nu = 50$                                          & $1719.52$                 & $2755.58$                & $2611.89$ ($n = 11244$)     \\ \hline
$\nu = 200$                                         & $5325.24$ ($n = 74403$)     & $5159.63$ ($n = 46409$)    & $11914.56$ ($n = 402$)      \\ \hline
\multicolumn{4}{|c|}{2D isotropic Gaussian field (empirical covariance case)}                                                    \\ \hline
number of fields = $200$                               & $9.56$                    & $15.13$                  & $208.33$                  \\ \hline
number of fields = $1000$                              & $9.91$                    & $14.19$                  & $145.04$                  \\ \hline
number of fields = $10{,}000$                             & $14.56$                   & $18.05$                  & $119.83$                  \\ \hline
\multicolumn{4}{|c|}{3D isotropic Gaussian field (empirical covariance case)}                                                    \\ \hline
number of fields = $200$                               & $82.94$                    & $105.66$                  & $1412.21$                  \\ \hline
number of fields = $1000$                              & $320.91$                    & $290.76$                  & $1722.92$                  \\ \hline
\end{tabular}
\end{center}
\end{table}

\end{document}